\documentclass[11pt]{article}

\usepackage{times}
\usepackage[paperwidth=199.8mm,
paperheight=297mm,centering,hmargin=20mm,vmargin=20mm]{geometry}
\usepackage{authblk} 
\usepackage[bottom]{footmisc} 

\usepackage[titletoc,title]{appendix} 
\usepackage{changepage}

\usepackage{hyperref}
\hypersetup{colorlinks=true,citecolor=myblue,linkcolor=myblue,
filecolor=myblue,urlcolor=myblue,breaklinks=true}
\usepackage{url}

\usepackage[dvipsnames]{xcolor}
\usepackage{color}
\usepackage{framed}

\definecolor{shadecolor}{rgb}{0.9,0.9,0.9}
\definecolor{mylightgray}{RGB}{100,100,100}

\definecolor{myblue}{RGB}{0, 68, 116}
\definecolor{mycyan}{RGB}{0, 97, 91}
\definecolor{mygreen}{RGB}{2, 102, 1}
\definecolor{myorange}{RGB}{240, 102, 0}
\definecolor{myred}{RGB}{172, 23, 0}
\definecolor{mymagenta}{RGB}{140,16,73}

\usepackage{mathtools}
\usepackage{amsmath}
\usepackage{amssymb}
\usepackage[shortlabels]{enumitem}
\usepackage{graphicx,epic,eepic,epsfig,latexsym,verbatim}
\usepackage{dsfont}
\usepackage{mathrsfs}

\usepackage{subcaption}
\usepackage{caption}
\usepackage{float}
\usepackage{tikz}
\usetikzlibrary{arrows.meta,positioning,calc,fit,matrix,backgrounds}
\usepackage{relsize}
\usepackage{pgfplots}
\pgfplotsset{compat=1.18}

\usepackage{amsthm}
\usepackage{tcolorbox}
\tcbuselibrary{skins}
\tcbuselibrary{breakable}

\definecolor{myblue}{RGB}{0,68,116}
\definecolor{mycyan}{RGB}{0,97,91}
\definecolor{mygreen}{RGB}{2,102,1}
\definecolor{myorange}{RGB}{240,102,0}
\definecolor{myred}{RGB}{172,23,0}

\newtheorem{theorem}{Theorem}
\newtheorem{lemma}[theorem]{Lemma}

\newtheorem{definition}[theorem]{Definition}
\newtheorem{boxedtheorem}[theorem]{Theorem}
\newtheorem{boxedlemma}[theorem]{Lemma}
\newtheorem{boxedprop}[theorem]{Proposition}
\newtheorem{boxedcorollary}[theorem]{Corollary}
\newtheorem{boxeddefinition}[theorem]{Definition}
\newtheorem{boxedassumption}[theorem]{Assumption}
\newtheorem{boxedremark}[theorem]{Remark}
\newtheorem{boxmessage}{Result}
\newtheorem{boxedexample}[theorem]{Example}

\newtheorem{remark}[theorem]{Remark}
\newtheorem{example}[theorem]{Example}

\definecolor{theoremframe}{RGB}{31,78,121}
\definecolor{theoremback}{RGB}{247,250,253}
\definecolor{lemmaframe}{RGB}{35,105,86}
\definecolor{lemmaback}{RGB}{247,252,250}
\definecolor{propframe}{RGB}{102,73,139}
\definecolor{propback}{RGB}{251,249,253}
\definecolor{corollaryframe}{RGB}{165,96,18}
\definecolor{corollaryback}{RGB}{254,251,246}
\definecolor{definitionframe}{RGB}{112,112,112}
\definecolor{definitionback}{RGB}{248,248,248}

\tcbset{
  supplement result/.style={
    enhanced jigsaw,
    breakable,
    lines before break=5,
    arc=2pt,
    outer arc=2pt,
    boxrule=0.55pt,
    left=6pt,
    right=6pt,
    oversize=0pt,
    top=6pt,
    bottom=6pt,
    before skip=9pt plus 2pt minus 1pt,
    after skip=9pt plus 2pt minus 1pt
  }
}
\tcolorboxenvironment{boxedtheorem}{
  supplement result,
  colframe=myblue,
  colback=myblue!5!white
}
\tcolorboxenvironment{boxedlemma}{
  supplement result,
  colframe=definitionframe,
  colback=definitionback
}
\tcolorboxenvironment{boxedprop}{
  supplement result,
  colframe=definitionframe,
  colback=definitionback
}
\tcolorboxenvironment{boxedcorollary}{
  supplement result,
  colframe=definitionframe,
  colback=definitionback
}
\tcolorboxenvironment{boxeddefinition}{
  supplement result,
  colframe=definitionframe,
  colback=definitionback
}
\tcolorboxenvironment{boxedassumption}{
  supplement result,
  colframe=definitionframe,
  colback=definitionback
}
\tcolorboxenvironment{boxedremark}{
  supplement result,
  colframe=myorange!70!black,
  colback=myorange!6!white
}
\tcolorboxenvironment{boxmessage}{
  supplement result,
  colframe=myorange!70!black,
  colback=myorange!6!white
}
\newtcolorbox{boxnote}{
  supplement result,
  colframe=myorange!70!black,
  colback=myorange!6!white
}
\tcolorboxenvironment{boxedexample}{
  supplement result,
  colframe=definitionframe,
  colback=definitionback
}

\newenvironment{boxtheorem}[1][]
  {\begin{boxedtheorem}[#1]}
  {\end{boxedtheorem}}
\newenvironment{boxlemma}[1][]
  {\begin{boxedlemma}[#1]}
  {\end{boxedlemma}}
\newenvironment{boxproposition}[1][]
  {\begin{boxedprop}[#1]}
  {\end{boxedprop}}
\newenvironment{boxcorollary}[1][]
  {\begin{boxedcorollary}[#1]}
  {\end{boxedcorollary}}
\newenvironment{boxdefinition}[1][]
  {\begin{boxeddefinition}[#1]}
  {\end{boxeddefinition}}
\newenvironment{boxremark}[1][]
  {\begin{boxedremark}[#1]}
  {\end{boxedremark}}
\newenvironment{boxexample}[1][]
  {\begin{boxedexample}[#1]}
  {\end{boxedexample}}

\newcommand{\prooftag}[1]{\medskip \noindent \textbf{\emph{#1}}}
 
\usepackage{colortbl}
\usepackage{pifont} 
\usepackage{booktabs}
\usepackage{makecell} 
\usepackage{diagbox}
\usepackage{multirow}

\newcommand{\nc}{\newcommand}
\nc{\rnc}{\renewcommand}

\nc{\<}{\langle}
\rnc{\>}{\rangle}
\nc{\bra}[1]{\langle#1|}
\nc{\ket}[1]{|#1\rangle}
\nc{\ketbra}[2]{|#1\rangle\!\langle#2|}
\nc{\braket}[2]{\langle#1|#2\rangle}
\nc{\braandket}[3]{\langle #1|#2|#3\rangle}
\nc{\proj}[1]{| #1\rangle\!\langle #1 |}
\nc{\avg}[1]{\langle#1\rangle}

\nc{\rank}{\operatorname{Rank}}
\nc{\id}{{\operatorname{id}}}
\nc{\iid}{{\operatorname{iid}}}
\nc{\supp}{{\operatorname{supp}}}
\nc{\smfrac}[2]{\mbox{$\frac{#1}{#2}$}}
\nc{\tr}{\operatorname{Tr}}
\nc{\ox}{\otimes}
\nc{\floor}[1]{\lfloor #1 \rfloor}
\nc{\trans}{\mathsf T}
\nc{\img}{\mathbf{i}}
\def\ve{\varepsilon}

\makeatletter

\newcommand{\solidgamearrowpicture@}[3]{%
    \tikz[x=1ex,y=1ex,line cap=round,line join=round]
      \draw[line width=#3,-{Triangle[length=0.64ex,width=0.68ex]}] (0,0) -- (#1,#2);%
}
\newcommand{\opengamearrowpicture@}[3]{%
    \tikz[x=1ex,y=1ex,line cap=round,line join=round]
      \draw[line width=#3,-{Triangle[open,length=0.64ex,width=0.68ex]}] (0,0) -- (#1,#2);%
}
\newsavebox{\solidgameuparrow@display}
\newsavebox{\solidgameuparrow@text}
\newsavebox{\solidgameuparrow@script}
\newsavebox{\solidgameuparrow@scriptscript}
\newsavebox{\solidgamedownarrow@display}
\newsavebox{\solidgamedownarrow@text}
\newsavebox{\solidgamedownarrow@script}
\newsavebox{\solidgamedownarrow@scriptscript}
\newsavebox{\opengameuparrow@display}
\newsavebox{\opengameuparrow@text}
\newsavebox{\opengameuparrow@script}
\newsavebox{\opengameuparrow@scriptscript}
\newsavebox{\opengamedownarrow@display}
\newsavebox{\opengamedownarrow@text}
\newsavebox{\opengamedownarrow@script}
\newsavebox{\opengamedownarrow@scriptscript}
\AtBeginDocument{%
  \sbox{\solidgameuparrow@display}{\solidgamearrowpicture@{0}{1.62}{0.58pt}}%
  \sbox{\solidgameuparrow@text}{\solidgamearrowpicture@{0}{1.50}{0.52pt}}%
  \sbox{\solidgameuparrow@script}{\solidgamearrowpicture@{0}{1.30}{0.44pt}}%
  \sbox{\solidgameuparrow@scriptscript}{\solidgamearrowpicture@{0}{1.12}{0.38pt}}%
  \sbox{\solidgamedownarrow@display}{\solidgamearrowpicture@{0}{-1.62}{0.58pt}}%
  \sbox{\solidgamedownarrow@text}{\solidgamearrowpicture@{0}{-1.50}{0.52pt}}%
  \sbox{\solidgamedownarrow@script}{\solidgamearrowpicture@{0}{-1.30}{0.44pt}}%
  \sbox{\solidgamedownarrow@scriptscript}{\solidgamearrowpicture@{0}{-1.12}{0.38pt}}%
  \sbox{\opengameuparrow@display}{\opengamearrowpicture@{0}{1.62}{0.58pt}}%
  \sbox{\opengameuparrow@text}{\opengamearrowpicture@{0}{1.50}{0.52pt}}%
  \sbox{\opengameuparrow@script}{\opengamearrowpicture@{0}{1.30}{0.44pt}}%
  \sbox{\opengameuparrow@scriptscript}{\opengamearrowpicture@{0}{1.12}{0.38pt}}%
  \sbox{\opengamedownarrow@display}{\opengamearrowpicture@{0}{-1.62}{0.58pt}}%
  \sbox{\opengamedownarrow@text}{\opengamearrowpicture@{0}{-1.50}{0.52pt}}%
  \sbox{\opengamedownarrow@script}{\opengamearrowpicture@{0}{-1.30}{0.44pt}}%
  \sbox{\opengamedownarrow@scriptscript}{\opengamearrowpicture@{0}{-1.12}{0.38pt}}%
}
\newcommand{\gamearrowbox@}[1]{%
  \mathrel{\vcenter{\hbox{%
    \usebox{#1}%
  }}}%
}
\DeclareRobustCommand{\soliduparrow}{%
  \mathchoice
    {\gamearrowbox@{\solidgameuparrow@display}}%
    {\gamearrowbox@{\solidgameuparrow@text}}%
    {\gamearrowbox@{\solidgameuparrow@script}}%
    {\gamearrowbox@{\solidgameuparrow@scriptscript}}%
}
\DeclareRobustCommand{\soliddownarrow}{%
  \mathchoice
    {\gamearrowbox@{\solidgamedownarrow@display}}%
    {\gamearrowbox@{\solidgamedownarrow@text}}%
    {\gamearrowbox@{\solidgamedownarrow@script}}%
    {\gamearrowbox@{\solidgamedownarrow@scriptscript}}%
}
\DeclareRobustCommand{\emptyuparrow}{%
  \mathchoice
    {\gamearrowbox@{\opengameuparrow@display}}%
    {\gamearrowbox@{\opengameuparrow@text}}%
    {\gamearrowbox@{\opengameuparrow@script}}%
    {\gamearrowbox@{\opengameuparrow@scriptscript}}%
}
\DeclareRobustCommand{\emptydownarrow}{%
  \mathchoice
    {\gamearrowbox@{\opengamedownarrow@display}}%
    {\gamearrowbox@{\opengamedownarrow@text}}%
    {\gamearrowbox@{\opengamedownarrow@script}}%
    {\gamearrowbox@{\opengamedownarrow@scriptscript}}%
}
\DeclareRobustCommand{\convdownarrow}{%
  \mathord{\overline{\emptydownarrow}}%
}
\protected\def\uparrow{\soliduparrow}
\protected\def\downarrow{\soliddownarrow}
\makeatother

\rnc{\liminf}{\mathop{\underline{\lim}}\displaylimits}
\rnc{\limsup}{\mathop{\overline{\lim}}\displaylimits}

\nc{\cA}{{\cal A}}
\nc{\cB}{{\cal B}}
\nc{\cC}{{\cal C}}
\nc{\cD}{{\cal D}}
\nc{\cE}{{\cal E}}
\nc{\cF}{{\cal F}}
\nc{\cG}{{\cal G}}
\nc{\cH}{{\cal H}}
\nc{\cI}{{\cal I}}
\nc{\cJ}{{\cal J}}
\nc{\cK}{{\cal K}}
\nc{\cL}{{\cal L}}
\nc{\cM}{{\cal M}}
\nc{\cN}{{\cal N}}
\nc{\cO}{{\cal O}}
\nc{\cP}{{\cal P}}
\nc{\cQ}{{\cal Q}}
\nc{\cR}{{\cal R}}
\nc{\cS}{{\cal S}}
\nc{\cT}{{\cal T}}
\nc{\cV}{{\cal V}}
\nc{\cU}{{\cal U}}
\nc{\cX}{{\cal X}}
\nc{\cY}{{\cal Y}}
\nc{\cZ}{{\cal Z}}
\nc{\cW}{{\cal W}}

\nc{\bp}{\boldsymbol{p}}
\nc{\bq}{\boldsymbol{q}}
\nc{\brho}{\boldsymbol{\rho}}
\nc{\bsigma}{\boldsymbol{\sigma}}
\nc{\bomega}{\boldsymbol{\omega}}
\nc{\bmu}{\boldsymbol{\mu}}
\nc{\bT}{\boldsymbol{T}}
\nc{\bS}{\boldsymbol{S}}

\nc{\RR}{{{\mathbb R}}}
\nc{\CC}{{{\mathbb C}}}
\nc{\FF}{{{\mathbb F}}}
\nc{\NN}{{{\mathbb N}}}
\nc{\ZZ}{{{\mathbb Z}}}
\nc{\QQ}{{{\mathbb Q}}}
\nc{\UU}{{{\mathbb U}}}
\nc{\EE}{{{\mathbb E}}}

\nc{\bH}{{\mathfrak{H}}}

\nc{\sK}{{{\mathscr{K}}}}
\nc{\sS}{{{\mathscr{S}}}}
\nc{\sT}{{{\mathscr{T}}}}
\nc{\sA}{{{\mathscr{A}}}}
\nc{\sB}{{{\mathscr{B}}}}
\nc{\sC}{{{\mathscr{C}}}}
\nc{\sE}{{{\mathscr{E}}}}
\nc{\sL}{{{\mathscr{L}}}}
\nc{\sG}{{{\mathscr{G}}}}
\nc{\sF}{{{\mathscr{F}}}}
\nc{\sI}{{{\mathscr{I}}}}
\nc{\sN}{{{\mathscr{N}}}}
\nc{\sM}{{{\mathscr{M}}}}

\nc{\Choi}{Choi-Jamio\l{}kowski }
\nc{\reg}{\infty}
\nc{\mx}{\text{\rm mx}}
\nc{\amo}{\text{\rm amo}}
\nc{\Renyi}{R\'{e}nyi }
\nc{\Stein}{\mathsf{Stein}}

\nc{\conv}{\operatorname{conv}}
\nc{\cvxset}{\mathscr{C}}

\nc{\RM}{{{\mathscr{R}}}}

\nc{\END}{\operatorname{End}}
\nc{\PERM}{\mathfrak{\sigma}}

\nc{\Cone}{\text{\rm Cone}}
\nc{\sep}{{\SEP}}

\nc{\DD}{{{\mathbb D}}}
\nc{\BS}{{\scriptscriptstyle \rm {BS}}}
\nc{\Sand}{{\scriptscriptstyle  \rm S}}
\nc{\Hypo}{{\scriptscriptstyle  \rm H}}
\nc{\Meas}{{\scriptscriptstyle \rm M}}
\nc{\Proj}{{{\scriptscriptstyle \rm P}}}
\nc{\KL}{{{\scriptscriptstyle \rm KL}}}
\nc{\IS}{{{\scriptscriptstyle \rm IS}}}

\nc{\suchthat}{\text{\rm s.t.}}

\nc{\pl}{{\scalebox{0.7}{+}}}
\nc{\HERM}{\mathscr{H}}
\nc{\PSD}{\HERM_{\pl}}
\nc{\PD}{\HERM_{\pl\pl}}
\nc{\density}{\mathscr{D}}
\nc{\subdensity}{\mathscr{D}_\bullet}

\nc{\polarPSD}[1]{{#1}_{\pl}^{\circ}}
\nc{\polarPSDre}[1]{{#1}_{\pl}^{\star}}
\nc{\polarPD}[1]{{#1}_{\pl\pl}^{\circ}}

\nc{\PPT}{\text{\rm PPT}}
\nc{\Rains}{\text{\rm Rains}}
\nc{\WD}{\text{\rm WD}}
\nc{\SEP}{\text{\rm SEP}}
\nc{\PSEP}{\text{\rm PSEP}}
\nc{\CPTP}{\text{\rm CPTP}}
\nc{\POVM}{\text{\rm POVM}}
\nc{\PVM}{\text{\rm PVM}}
\nc{\CP}{\text{\rm CP}}
\nc{\adv}{\text{\rm adv}}
\nc{\spec}{\text{\rm spec}}
\nc{\poly}{\text{\rm poly}}
\nc{\End}{\operatorname{End}}
\nc{\Par}{\operatorname{Par}}
\nc{\RNG}{\operatorname{RNG}}
\nc{\STAB}{\text{\rm STAB}}
\nc{\epi}{\boldsymbol{\operatorname{epi}}}
\nc{\op}{\boldsymbol{\operatorname{op}}}

\makeatletter
\newcommand*\rel@kern[1]{\kern#1\dimexpr\macc@kerna}
\newcommand*\widebar[1]{%
  \begingroup
  \def\mathaccent##1##2{%
    \rel@kern{0.8}%
    \overline{\rel@kern{-0.8}\macc@nucleus\rel@kern{0.2}}%
    \rel@kern{-0.2}%
  }%
  \macc@depth\@ne
  \let\math@bgroup\@empty \let\math@egroup\macc@set@skewchar
  \mathsurround\z@ \frozen@everymath{\mathgroup\macc@group\relax}%
  \macc@set@skewchar\relax
  \let\mathaccentV\macc@nested@a
  \macc@nested@a\relax111{#1}%
  \endgroup
}
\makeatother

\renewcommand{\mathbf}{\boldsymbol}
\renewcommand{\ge}{\geqslant}
\renewcommand{\geq}{\geqslant}
\renewcommand{\le}{\leqslant}
\renewcommand{\leq}{\leqslant}

\newcommand*{\set}[1]{\mathcal{#1}} 

\newcommand{\reals}{\mathbb{R}}

\newcommand{\LinOp}{\mathscr{L}}
\newcommand{\HermOp}{\mathscr{H}}
\newcommand{\herm}{\dagger}

\DeclarePairedDelimiterX{\spr}[2]{\langle}{\rangle}{#1\delimsize\vert#2}
\DeclarePairedDelimiterX{\ceil}[1]{\lceil}{\rceil}{#1}
\DeclarePairedDelimiterX{\abs}[1]{\lvert}{\rvert}{#1}
\DeclarePairedDelimiterX{\norm}[1]{\lVert}{\rVert}{#1}
\DeclarePairedDelimiterX{\size}[1]{\lvert}{\rvert}{#1}
\DeclarePairedDelimiterX{\infdiv}[2]{(}{)}{#1\delimsize\Vert#2}
\DeclarePairedDelimiterX{\infdivc}[3]{(}{)}{#1\delimsize\Vert#2\delimsize\vert#3}
\DeclarePairedDelimiterX{\inner}[2]{\langle}{\rangle}{#1,#2}
\newcommand{\Div}{\mathbb{D}} 
\newcommand{\uDiv}{D}
\newcommand{\hDiv}[1]{D_{\Hypo, #1}}
\newcommand{\mDiv}{D_{\scriptscriptstyle  \rm M}}
\newcommand{\mrDiv}[1]{D_{{\scriptscriptstyle  \rm M}, #1}}
\newcommand{\supDiv}{\Div^{\uparrow}}

\newcommand{\infDiv}{\Div^{\downarrow}}
\newcommand{\iinfDiv}{\Div^{\downarrow\downarrow}}

\newcommand{\iinfmrDiv}[1]{\mrDiv{#1}^{\downarrow\downarrow}}

\newcommand{\new}{\text{\rm new}}
\newcommand{\PER}{\mathscr{P}}

\ExplSyntaxOn
\NewDocumentCommand{\multiadjustlimits}{m}
 {
  \group_begin:
  \multiadjustlimits_measure:n { #1 }
  \multiadjustlimits_print:n { #1 }
  \group_end:
 }

\tl_new:N  \l__multiadjustlimits_operator_tl
\tl_new:N  \l__multiadjustlimits_limit_tl

\cs_new_protected:Nn \multiadjustlimits_measure:n
 {
  \clist_map_function:nN { #1 } \__multiadjustlimits_measure:n
 }
\cs_new_protected:Nn \__multiadjustlimits_measure:n
 {
  \__multiadjustlimits_measure:NNn #1
 }
\cs_new_protected:Nn \__multiadjustlimits_measure:NNn
 {
  \tl_put_right:Nn \l__multiadjustlimits_operator_tl { #1 }
  \tl_put_right:Nn \l__multiadjustlimits_limit_tl { #3 }
 }

\cs_new_protected:Nn \multiadjustlimits_print:n
 {
  \clist_map_function:nN { #1 } \__multiadjustlimits_print:n
 }
\cs_new_protected:Nn \__multiadjustlimits_print:n
 {
  \__multiadjustlimits_print:NNn #1
 }
\cs_new_protected:Nn \__multiadjustlimits_print:NNn
 {
  \mathop { \vphantom{\l__multiadjustlimits_operator_tl} \mathopen{} #1 }
  \limits
  \sb{ \vphantom{\cramped{\l__multiadjustlimits_limit_tl}} #3 }
 }

\ExplSyntaxOff
\newcommand\ie{\textit{i.e.}}

\newcommand\cf{\textit{cf.}}

\usepackage{framed}
\usepackage[dvipsnames]{xcolor} \usepackage{color}
\definecolor{shadecolor}{rgb}{0.9,0.9,0.9}

\usepackage{cite}
\usepackage{array}
\usepackage{tabularx}
\usepackage{arydshln}

\begin{document}

\title{\LARGE \textbf{Minimax games for quantum channel discrimination} \bigskip}


\author[1]{Kun Fang \thanks{kunfang@cuhk.edu.cn}}
\author[2,3]{Michael X. Cao}
\author[4,5,6,7,8]{Hao-Chung Cheng}
\author[9,10]{\\ Li Gao}
\author[1,11,12]{Masahito Hayashi \thanks{hmasahito@cuhk.edu.cn}}

\affil[1]{\small{School of Data Science, The Chinese University of Hong Kong, Shenzhen,
Guangdong, 518172, China}}
\affil[2]{\small{Institute for Quantum Information, RWTH Aachen University, Aachen, Germany}}
\affil[3]{\small{City University of Hong Kong (Dongguan), Guangdong, China}}
\affil[4]{\small{Department of Electrical Engineering and Graduate Institute of
Communication Engineering,\protect \\ National Taiwan University, Taiwan}}
\affil[5]{\small{Department of Mathematics, National Taiwan University, Taiwan}}
\affil[6]{\small{Center for Quantum Science and Engineering, National Taiwan University, Taiwan}}
\affil[7]{\small{Physics/Mathematics Division, National Center for Theoretical Sciences, Taiwan}}
\affil[8]{\small{Hon Hai (Foxconn) Quantum Computing Center, Taiwan}}
\affil[9]{\small{School of Mathematics and Statistics, Wuhan University, Wuhan, 430072, China}}
\affil[10]{\small{Wuhan Institute of Quantum
Science and Technology, Wuhan 430075, China}}
\affil[11]{\small{International Quantum Academy, Futian District, Shenzhen 518048, China}}
\affil[12]{\small{Graduate School of Mathematics, Nagoya University, Chikusa-ku, Nagoya 464–8602, Japan}}

\date{\today}

\maketitle

\begin{abstract}
Quantum channel discrimination is a primitive task for identifying, verifying, and benchmarking quantum dynamics. Previous studies have primarily considered either the best-case tester-input setting or the worst-case jammer-input setting. Here, we introduce a game-theoretic framework in which both the tester and jammer control separate inputs. Combining three input structures, characterized by whether the tester and jammer use entangled inputs or IID inputs across channel uses, with four information patterns, determined by the visibility of the jammer's strategy and its knowledge of the true hypothesis, yields twelve game models. We provide exact finite blocklength hypothesis testing characterizations of all twelve models in terms of nine minimax hypothesis testing divergences and derive their asymptotic Stein exponents. Notably, for entangled jammers, neither the visibility of the jammer's strategy nor its knowledge of the true hypothesis affects the asymptotic Stein exponent, whereas the information pattern remains consequential for IID jammers. As an example, we study the discrimination of a general channel from a replacer channel and show that all asymptotic Stein exponents coincide with the same additive, single letter quantity. We further develop a general argument that upgrades achievability results to strong converse results, thereby establishing strong converse properties for several game models, resolving an open problem in composite hypothesis testing posed by Berta et al. [Commun. Math. Phys. 385, 55 (2021)], and strengthening several recent results of Lami [arXiv:2510.06340]. The framework and techniques developed here may support future studies of quantum information tasks involving competing roles.
\end{abstract}

\newpage
{
\setcounter{tocdepth}{2}
\tableofcontents
}

\newpage
\section{Introduction}

\subsection{Background and motivation}

Discrimination is a fundamental primitive in information processing, with
applications ranging from signal detection to pattern recognition and machine
learning. In quantum information theory, its simplest form is quantum state
discrimination, where an unknown state is identified through a hypothesis test.
Quantum Stein's lemma is a cornerstone of this setting: in the asymptotic
regime, the optimal type-II error exponent under a fixed type-I error
constraint is given by the quantum relative entropy~\cite{hiai1991proper,Ogawa2000}.
This operational interpretation, together with its further extensions, continues to motivate the study of quantum hypothesis
testing~\cite{brandao2010generalization,Hayashi2025,lami2024solutiongeneralisedquantumsteins,fang2024generalized}.

Quantum channel discrimination extends this hypothesis-testing perspective from
static states to dynamical quantum resources. The task is to distinguish two
candidate channels by preparing suitable inputs and testing the corresponding
outputs. The additional freedom to choose the input, absent in state
discrimination, allows the tester to use ancillary systems, entanglement, or
adaptive strategies to improve distinguishability. This framework has been extensively studied and underlies a wide
range of applications~\cite{dariano_using_2001,chiribella_memory_2008,hayashi_discrimination_2009,harrow_adaptive_2010,wilde2020amortized,cooney2016strong,pirandola_fundamental_2019,WW2019,bergh_parallelization_2024},
including studies of quantum advantage and entanglement~\cite{piani2009all,takagi_operational_2019,bae_more_2019,skrzypczyk_robustness_2019},
quantum communication and channel capacities~\cite{wang2012one,datta_smooth_2013,Wang2019,wang2019converse,fang2021geometric,fang2025towards},
quantum sensing and reading~\cite{pirandola_advances_2018,zhuang_ultimate_2020},
and quantum biology~\cite{spedalieri_detecting_2020,pereira2020quantum}.

Previous studies on quantum channel discrimination have focused primarily on the conventional \emph{tester-input} scenario~\cite{acin2001statistical,duan2007entanglement,chiribella_memory_2008,piani2009all,duan2009perfect,bae_more_2019,takagi_operational_2019,skrzypczyk_robustness_2019,fang2020chain,zhuang_ultimate_2020,bavaresco2021strict,debry2023experimental,sugiura2024power}, an optimistic, best-case setting in which the tester fully controls and can therefore trust the input state supplied to the channel and the final measurement (see Figure~\ref{fig:parallel-parallel}(a)). More recently, a \emph{jammer-input} adversarial framework was introduced~\cite{fang2025adversarial}, representing the opposite, worst-case setting in which an untrusted device controls the input. Motivated by quantum state verification~\cite{zhu2019efficient}, this framework models a state-preparation device that should produce a specified resource state but may instead output arbitrary junk if faulty or malicious; the tester designs the measurement, while the jammer controls the input states (see Figure~\ref{fig:parallel-parallel}(b)).

\begin{figure}[H]
  \centering
  \includegraphics[width=\textwidth]{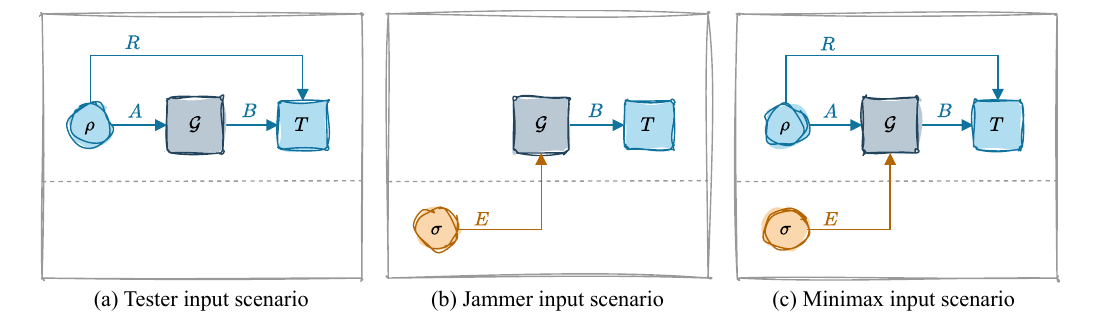}
  \caption{Tester-input, jammer-input, and minimax-input configurations. The central grey box $\cG$ denotes the channel under test. Tester systems are shown in blue, and jammer systems in orange. In (a), the tester controls the channel input and final  test, while the jammer input is absent. In (b), the jammer controls the channel input, while the tester performs the final  test and has no input control. In (c), the tester and jammer control separate channel inputs, and the tester performs the final  test.}
  \label{fig:parallel-parallel}
\end{figure}

These two extreme models provide useful idealizations, corresponding to settings in which the input is either fully controlled by the tester or fully determined by an external source. Many physically relevant scenarios, however, lie between these extremes, especially when environmental noise, hardware fluctuations, interference, or other unpredictable disturbances cannot be neglected. A motivating example is quantum arbitrarily varying-channel communication, in which the legitimate sender prepares the message-bearing quantum input, while a jammer or uncontrolled environment supplies a disturbance state, and the receiver observes only the resulting joint output~\cite{boche2018fully,belzig2024fully,Dasgupta2025,cao2025channel}. This model aims to determine when reliable communication remains possible under worst-case channel variations and to design protocols that are robust to unpredictable or adversarial behavior, as required in secure and anti-jamming communication systems. More broadly, the same perspective applies, in principle, to any quantum information processing task that requires performance guarantees against external uncertainties. In the setting of channel discrimination, this motivates the central question of this work: 

\begin{center}
  \emph{What are the fundamental limits of quantum channel discrimination when the tester and the jammer have separate control over their inputs and play competing roles?}
\end{center}

To address this question, we introduce a minimax game-theoretic framework involving the competing roles of a tester and a jammer, who control separate channel inputs; see Figure~\ref{fig:parallel-parallel}(c). By making either input trivial, the framework recovers the conventional tester-input and jammer-input settings as special cases. It therefore unifies these two endpoints while providing a more versatile model for channel discrimination under partial control and external uncertainty.

\subsection{Game models}

We formulate channel discrimination with split input control as an asymmetric
quantum hypothesis-testing game. Because the channel inputs may be chosen by two
competing parties, each game model is specified by two independent ingredients.
The first is the \emph{information pattern}, which specifies what the tester
knows about the jammer's input strategy and whether the jammer knows the active
hypothesis. The second is the \emph{input structure}, which specifies how the
tester and jammer may prepare and correlate their inputs across the $n$ parallel
channel uses, including entangled and IID inputs.

\begin{figure}[H]
  \centering
  \includegraphics[width=0.9\textwidth]{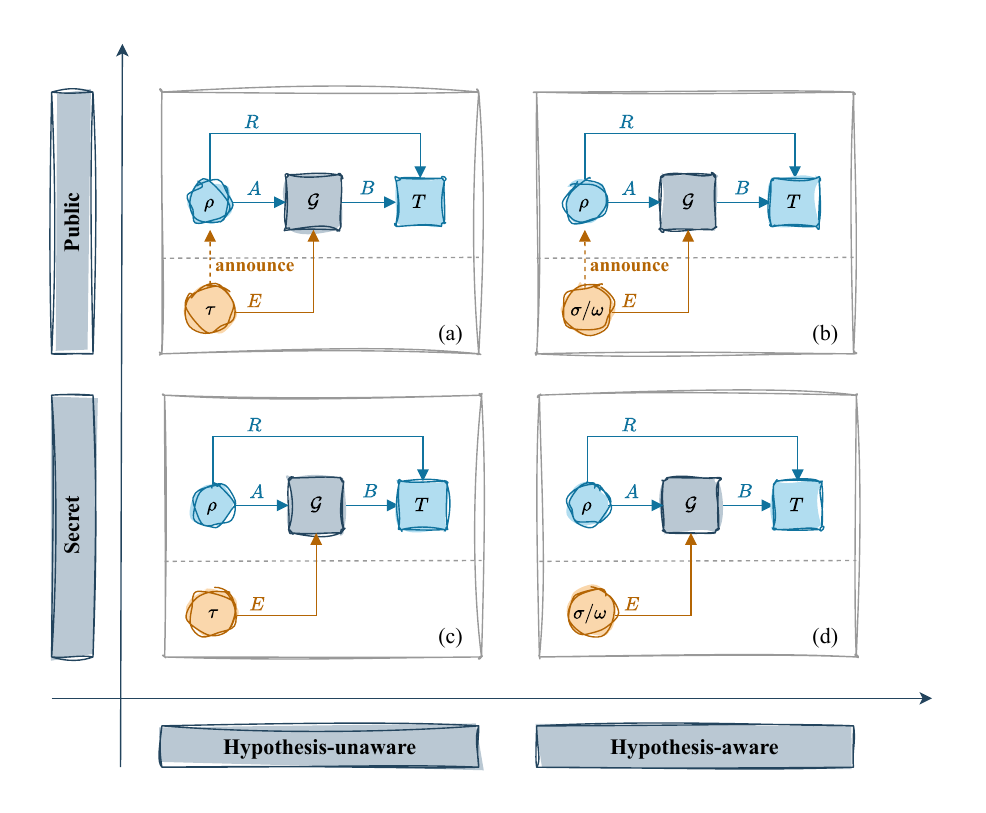}
  \caption{Four information patterns of the jammer. Rows indicate whether the jammer's input-selection rule is public or secret, while columns indicate whether the jammer is hypothesis-unaware or hypothesis-aware. In the public row, the dashed orange arrow denotes that the jammer's rule is disclosed before the tester acts. In the hypothesis-aware column, the jammer may use $\sigma$ when $\cN$ is active and $\omega$ when $\cM$ is active; in the hypothesis-unaware column, it must use the same state $\tau$ under both hypotheses.}
  \label{fig:information-structures}
\end{figure}

\paragraph{Four information patterns.}

The jammer's information pattern is specified by two binary axes as shown in Figure~\ref{fig:information-structures}.
The first axis concerns the visibility of the jammer's rule. A \emph{public}
jammer discloses its rule before the tester chooses the probing state and final
binary test, whereas a \emph{secret} jammer keeps the rule hidden. The second
axis concerns hypothesis awareness. A \emph{hypothesis-unaware} jammer must
inject the same state $\tau$ under both hypotheses, whereas a
\emph{hypothesis-aware} jammer may use $\sigma$ when $\cN$ is active and
$\omega$ when $\cM$ is active. Note that in the public, hypothesis-aware setting, the tester knows the complete
state-selection rule before choosing $\rho$ and $T$, but does not learn which
branch is realized in a particular run. Thus, only one jammer state is injected
in each run. Revealing the realized branch would provide side
information and expose the active hypothesis, thereby trivializing the game.

\paragraph{Three input structures.}
Across $n$ parallel rounds, the tester prepares an input state $\rho_n$ on
$R^nA^n$, where $R^n$ is a reference system retained for the final test. The jammer selects a state $\tau_n$ (or $\sigma_n$, $\omega_n$ depending on the information pattern)
on $E^n$. For each hypothesis, the tester's and jammer's inputs are required to
tensorize, although each player may correlate its own systems across channel
uses. We
distinguish between \emph{entangled} and \emph{IID} inputs. An entangled input
may be any state in $\density(X^n):=\{\tau\geq0:\tr\tau=1\}$, including states
with quantum entanglement or classical correlations across the $n$ uses. An
IID input, by contrast, is restricted to a deterministic tensor-power state in
$\sI(X^n):=\{\tau^{\ox n}:\tau\in\density(X)\}$. Here $X=AR$ for the tester
and $X=E$ for the jammer. We consider three input structures as shown in Figure~\ref{fig: input structure}: an entangled
tester against an entangled jammer, an IID tester against an entangled jammer,
and an IID tester against an IID jammer. This choice reflects the conservative
viewpoint common in quantum cryptography, in which the jammer is granted at
least as much input power as the tester.

\begin{figure}[H]
  \centering
  \includegraphics[width=\textwidth]{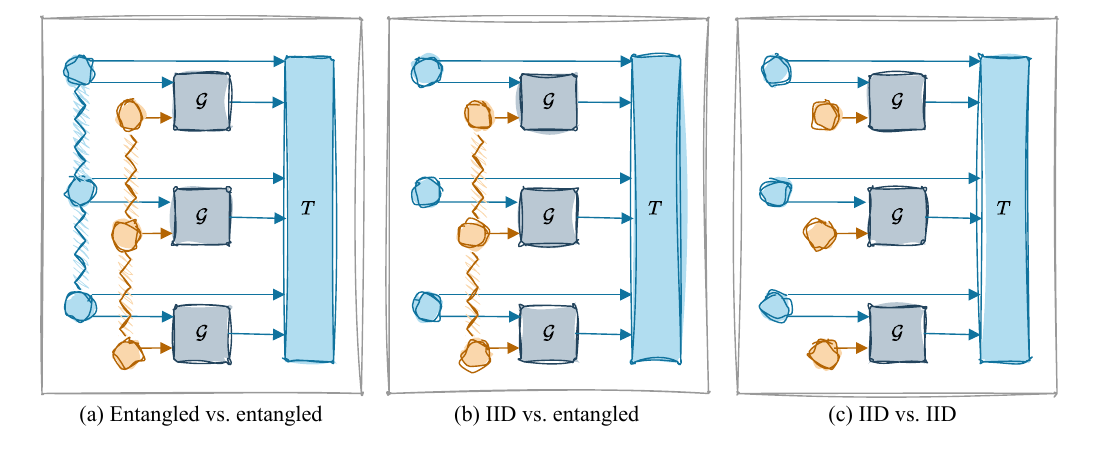}
  \caption{Three input structures for $n$ parallel channel uses. The tester's systems are shown in blue and the jammer's systems in orange. In (a), both players may prepare entangled inputs. In (b), the tester is restricted to IID inputs, whereas the jammer may still prepare entangled inputs. In (c), both players are restricted to IID inputs. All tests are performed on the tester's systems and can be entangled.}
  \label{fig: input structure}
\end{figure}

\paragraph{Twelve game models.}
Combining the four information patterns with the three input structures
yields twelve game models. Each model is an asymmetric binary
hypothesis-testing game played over $n$ parallel channel uses. The null channel
is $\cN_{AE\to B}$, the alternative channel is $\cM_{AE\to B}$, and the tester
and jammer control the input systems $A$ and $E$, respectively. The tester may
also retain a reference system $R$ and apply a final binary test to $R^nB^n$.
Let $\sE(X):=\{T:0\le T\le I_X\}$ denote the set of test effects on a system
$X$. We interpret $T$ as the effect corresponding to acceptance of the null
hypothesis, so that $I-T$ corresponds to rejection. For a tester input $\rho_n$,
jammer inputs $\sigma_n$ ($\omega_n$), and a test effect
$T\in\sE(B^nR^n)$, define the type-I and type-II error probabilities by
\begin{align}
\alpha_T(\rho_n,\sigma_n)
&:=
\tr\!\left[
(\cN^{\ox n}\ox\id_{R^n})(\rho_n\ox\sigma_n)(I-T)
\right],\\
\beta_T(\rho_n,\omega_n)
&:=
\tr\!\left[
(\cM^{\ox n}\ox\id_{R^n})(\rho_n\ox\omega_n)T
\right].
\label{eq:operational-error-functionals}
\end{align}
Here $\alpha_T$ is the type-I error, namely, the probability of rejecting
$\cN$ when $\cN$ is the true channel. Similarly, $\beta_T$ is the type-II
error, namely, the probability of accepting $\cN$ when $\cM$ is the true
channel. The tester seeks to maximize the payoff $-\log\beta_T$, subject to the
constraint $\alpha_T\leq\ve$, thereby making the two channels as distinguishable
as possible. The jammer has the opposite objective: it chooses its allowed input
states to minimize the tester's payoff.

\begin{table}[H]
\centering
\small
\setlength{\tabcolsep}{1.5pt}
\renewcommand{\arraystretch}{1.1}
\renewcommand{\tabularxcolumn}[1]{m{#1}}
\newcommand{\gamecell}[1]{%
  \raisebox{-1ex}[0pt][0pt]{\shortstack[c]{#1}}}
\begin{tabularx}{\textwidth}{
@{}
>{\centering\arraybackslash}m{0.13\textwidth}
>{\centering\arraybackslash}m{0.12\textwidth}
>{\centering\arraybackslash}m{0.1\textwidth}
>{\centering\arraybackslash}m{0.27\textwidth}
>{\centering\arraybackslash}X
>{\centering\arraybackslash}m{0.15\textwidth}
@{}}
\toprule
\noalign{\vskip1ex}
\shortstack[c]{\strut Input \\ structure\strut}
&
\shortstack[c]{\strut Information\\ pattern\strut}
&
\shortstack[c]{\strut Stein \\ exponent\strut}
&
\shortstack[c]{\strut Outer\\optimization\strut}
&
\shortstack[c]{\strut Inner\\optimization\strut}
&
\shortstack[c]{\strut Payoff \\ function \strut}
\\
\midrule
\multirow[c]{8}{*}{\shortstack[c]{Entangled \\ vs. entangled}}
&
\gamecell{Public\\unaware}
&
$\Stein_{\ve,n}^{\downarrow,\uparrow}$
&
$\displaystyle
\inf_{\tau_n\in\density(E^n)}$
&
$\displaystyle
\sup_{\substack{
  \rho_n\in\density((AR)^n)\\ T\in\sE((BR)^n)\\
  \alpha_T(\rho_n,\tau_n)\le\ve}}$
&
$-\log\beta_T(\rho_n,\tau_n)$
\\
\noalign{\vskip1ex}
\cdashline{2-6}[2pt/2pt]
&
\gamecell{Public\\aware}
&
$\Stein_{\ve,n}^{\downarrow\downarrow,\uparrow}$
&
$\displaystyle
\inf_{\sigma_n,\omega_n\in\density(E^n)}$
&
$\displaystyle
\sup_{\substack{
  \rho_n\in\density((AR)^n)\\ T\in\sE((BR)^n)\\
  \alpha_T(\rho_n,\sigma_n)\le\ve}}$
&
$-\log\beta_T(\rho_n,\omega_n)$
\\
\noalign{\vskip1ex}
\cdashline{2-6}[2pt/2pt]
&
\gamecell{Secret\\unaware}
&
$\Stein_{\ve,n}^{\uparrow,\downarrow}$
&
$\displaystyle
\sup_{\substack{
  \rho_n\in\density((AR)^n)\\ T\in\sE((BR)^n)\\
  \alpha_T(\rho_n,\tau_n)\le\ve,
  \forall\,\tau_n\in\density(E^n)}}$
&
$\displaystyle
\inf_{\tau_n\in\density(E^n)}$
&
$-\log\beta_T(\rho_n,\tau_n)$
\\
\noalign{\vskip1ex}
\cdashline{2-6}[2pt/2pt]
&
\gamecell{Secret\\aware}
&
$\Stein_{\ve,n}^{\uparrow,\downarrow\downarrow}$
&
$\displaystyle
\sup_{\substack{
  \rho_n\in\density((AR)^n)\\ T\in\sE((BR)^n)\\
  \alpha_T(\rho_n,\sigma_n)\le\ve,
  \forall\,\sigma_n\in\density(E^n)}}$
&
$\displaystyle
\inf_{\omega_n\in\density(E^n)}$
&
$-\log\beta_T(\rho_n,\omega_n)$
\\
\noalign{\vskip1ex}
\midrule 
\multirow[c]{8}{*}{\shortstack[c]{IID vs.\\ entangled}}
&
\gamecell{Public\\unaware}
&
$\Stein_{\ve,n}^{\downarrow,\emptyuparrow}$
&
$\displaystyle
\inf_{\tau_n\in\density(E^n)}$
&
$\displaystyle
\sup_{\substack{
  \rho_n\in\sI((AR)^n)\\ T\in\sE((BR)^n)\\
  \alpha_T(\rho_n,\tau_n)\le\ve}}$
&
$-\log\beta_T(\rho_n,\tau_n)$
\\
\noalign{\vskip1ex}
\cdashline{2-6}[2pt/2pt]

&
\gamecell{Public\\aware}
&
$\Stein_{\ve,n}^{\downarrow\downarrow,\emptyuparrow}$
&
$\displaystyle
\inf_{\sigma_n,\omega_n\in\density(E^n)}$
&
$\displaystyle
\sup_{\substack{
  \rho_n\in\sI((AR)^n)\\ T\in\sE((BR)^n)\\
  \alpha_T(\rho_n,\sigma_n)\le\ve}}$
&
$-\log\beta_T(\rho_n,\omega_n)$
\\
\noalign{\vskip1ex}
\cdashline{2-6}[2pt/2pt]
&
\gamecell{Secret\\unaware}
&
$\Stein_{\ve,n}^{\emptyuparrow,\downarrow}$
&
$\displaystyle
\sup_{\substack{
  \rho_n\in\sI((AR)^n)\\ T\in\sE((BR)^n)\\
  \alpha_T(\rho_n,\tau_n)\le\ve,
  \forall\,\tau_n\in\density(E^n)}}$
&
$\displaystyle
\inf_{\tau_n\in\density(E^n)}$
&
$-\log\beta_T(\rho_n,\tau_n)$
\\
\noalign{\vskip1ex}
\cdashline{2-6}[2pt/2pt]
&
\gamecell{Secret\\aware}
&
$\Stein_{\ve,n}^{\emptyuparrow,\downarrow\downarrow}$
&
$\displaystyle
\sup_{\substack{
  \rho_n\in\sI((AR)^n)\\ T\in\sE((BR)^n)\\
  \alpha_T(\rho_n,\sigma_n)\le\ve,
  \forall\,\sigma_n\in\density(E^n)}}$
&
$\displaystyle
\inf_{\omega_n\in\density(E^n)}$
&
$-\log\beta_T(\rho_n,\omega_n)$
\\
\noalign{\vskip1ex}
\midrule
\multirow[c]{8}{*}{\shortstack[c]{IID vs. IID}}
&
\gamecell{Public\\unaware}
&
$\Stein_{\ve,n}^{\emptydownarrow,\emptyuparrow}$
&
$\displaystyle
\inf_{\tau_n\in\sI(E^n)}$
&
$\displaystyle
\sup_{\substack{
  \rho_n\in\sI((AR)^n)\\ T\in\sE((BR)^n)\\
  \alpha_T(\rho_n,\tau_n)\le\ve}}$
&
$-\log\beta_T(\rho_n,\tau_n)$
\\
\noalign{\vskip1ex}
\cdashline{2-6}[2pt/2pt]
&
\gamecell{Public\\aware}
&
$\Stein_{\ve,n}^{\emptydownarrow\emptydownarrow,\emptyuparrow}$
&
$\displaystyle
\inf_{\sigma_n,\omega_n\in\sI(E^n)}$
&
$\displaystyle
\sup_{\substack{
  \rho_n\in\sI((AR)^n)\\ T\in\sE((BR)^n)\\
  \alpha_T(\rho_n,\sigma_n)\le\ve}}$
&
$-\log\beta_T(\rho_n,\omega_n)$
\\
\noalign{\vskip1ex}
\cdashline{2-6}[2pt/2pt]
&
\gamecell{Secret\\unaware}
&
$\Stein_{\ve,n}^{\emptyuparrow,\emptydownarrow}$
&
$\displaystyle
\sup_{\substack{
  \rho_n\in\sI((AR)^n)\\ T\in\sE((BR)^n)\\
  \alpha_T(\rho_n,\tau_n)\le\ve,
  \forall\,\tau_n\in\sI(E^n)}}$
&
$\displaystyle
\inf_{\tau_n\in\sI(E^n)}$
&
$-\log\beta_T(\rho_n,\tau_n)$
\\
\noalign{\vskip1ex}
\cdashline{2-6}[2pt/2pt]
&
\gamecell{Secret\\aware}
&
$\Stein_{\ve,n}^{\emptyuparrow,\emptydownarrow\emptydownarrow}$
&
$\displaystyle
\sup_{\substack{
  \rho_n\in\sI((AR)^n)\\ T\in\sE((BR)^n)\\
  \alpha_T(\rho_n,\sigma_n)\le\ve,
  \forall\,\sigma_n\in\sI(E^n)}}$
&
$\displaystyle
\inf_{\omega_n\in\sI(E^n)}$
&
$-\log\beta_T(\rho_n,\omega_n)$
\\
\noalign{\vskip1ex}
\bottomrule
\end{tabularx}
\caption{Twelve Stein exponents for minimax channel-discrimination games.}
\label{tab:finite-block-Stein-quantities}
\label{tab:arrow-code-dictionary}
\label{tab:scenario-summary}
\end{table}

Because the framework contains several game models, we use the arrow code in
Table~\ref{tab:finite-block-Stein-quantities} to specify the information
pattern and input structure of each model. Each code can be decoded in three
steps: the direction identifies the optimizing player, the arrow style identifies the allowed input class and the jammer's hypothesis
awareness, and the order identifies the jammer's visibility.
\begin{itemize}
\item \emph{Arrow directions.} An upward arrow, such as ``$\uparrow$'' or
``$\emptyuparrow$'', denotes the tester's maximization. A downward arrow, such as
``$\downarrow$'', ``$\emptydownarrow$'', ``$\downarrow\downarrow$'', or
``$\emptydownarrow\emptydownarrow$'', denotes the jammer's minimization.
\item \emph{Arrow styles.} Solid arrows, such as ``$\uparrow$'', ``$\downarrow$'', and
``$\downarrow\downarrow$'', represent entangled inputs. Hollow arrows, such as
``$\emptyuparrow$'', ``$\emptydownarrow$'', and
``$\emptydownarrow\emptydownarrow$'', represent IID inputs. For the jammer, a
single downward arrow denotes hypothesis-unawareness, where the same
input is used under both hypotheses. Two consecutive downward arrows denote
hypothesis-awareness, where different inputs may be chosen.

\item \emph{Arrow orders.} The arrows are listed from left to right in the order
of optimization, from the outer optimization to the inner one, encoding the jammer's visibility. In a public-jammer
game, the tester can condition its strategy on the jammer's input, so the code
begins with the jammer's minimization and then the tester's maximization, e.g.,``$\downarrow,\uparrow$''. In a secret-jammer game, the tester must choose
without observing the jammer's input, so the order is reversed, e.g.,
``$\uparrow,\downarrow$''.
\end{itemize}  

For each game with arrow code ``$(\star)$'', the corresponding
finite-blocklength Stein exponent is denoted by
$\Stein_{\ve,n}^{(\star)}(\cN,\cM)$ and defined in
Table~\ref{tab:finite-block-Stein-quantities}. As an example, consider the public,
hypothesis-unaware game with entangled inputs for both players. The jammer
selects one state $\tau_n\in\density(E^n)$ and uses it under both hypotheses.
Because the jammer input is public, the tester observes $\tau_n$ before choosing
the tester input state $\rho_n$ and the test $T$. The tester then maximizes the
type-II error exponent subject to the type-I constraint. This gives
\begin{align}
\Stein_{\ve,n}^{\downarrow,\uparrow}(\cN,\cM):= \quad \inf_{\tau_n \in \density(E^n)} \quad \sup_{\substack{\rho_n \in \density((AR)^n)\\ T \in \sE((BR)^n)\\ \alpha_T(\rho_n,\tau_n)\le\ve}} \quad -\log\beta_T(\rho_n,\tau_n),
\end{align}
which corresponds to the first row of the table. The remaining exponents are defined in a similar way.

\subsection{Main results}

Our main objective is to characterize the asymptotic Stein exponent of each game, thereby establishing quantum Stein's lemmas for the corresponding models and understanding their fundamental limits.
For a discrimination
game with arrow code ``$(\star)$'', we denote its asymptotic rate by
\begin{align}\label{eq: stein rate}
  \Stein^{(\star)}(\cN,\cM):= \lim_{\ve\to 0^+} \liminf_{n\to \infty} \frac{1}{n} \Stein_{\ve,n}^{(\star)}(\cN,\cM).
\end{align} 
Since these games differ in their information patterns and input structures, it is natural to expect them to exhibit distinct operational behavior. Somewhat surprisingly, we find that many of these settings turn out to coincide in the asymptotic regime (Theorem~\ref{thm:general-general-Stein},~\ref{thm:iid-tester-general-jammer-Stein},~\ref{thm: public hypo hypothesis-unaware iid} and~\ref{thm: secret hypo hypothesis-aware iid}).

\begin{boxmessage}[Operational collapse, informal]
If the jammer may use entangled inputs, all four information patterns lead to
the same asymptotic Stein exponent. Thus, neither the visibility of the jammer
input nor the jammer's hypothesis awareness changes the asymptotic rate:
  \begin{align}
    \Stein^{\downarrow,\uparrow} = \Stein^{\downarrow\downarrow,\uparrow} = \Stein^{\uparrow,\downarrow} = \Stein^{\uparrow,\downarrow\downarrow},\label{eq: messsage 1 tmp1}\\
    \Stein^{\downarrow,\emptyuparrow} = \Stein^{\downarrow\downarrow,\emptyuparrow} = \Stein^{\emptyuparrow,\downarrow}= \Stein^{\emptyuparrow,\downarrow\downarrow}.\label{eq: messsage 1 tmp2}
  \end{align}
By contrast, if both players are restricted to IID inputs, the information
pattern indeed matters. In general, the four games do
not all collapse, although the two secret-jammer exponents coincide:
  \begin{align}
    \Stein^{\emptydownarrow,\emptyuparrow} \neq \Stein^{\emptydownarrow\emptydownarrow,\emptyuparrow} \neq \Stein^{\emptyuparrow,\emptydownarrow}= \Stein^{\emptyuparrow,\emptydownarrow\emptydownarrow}.
  \end{align}
\end{boxmessage}

Establishing these collapses requires new techniques to compare games with
different optimization orders and input constraints. For entangled inputs, a minimax
identity (Propositions~\ref{prop: jammer divergence minimax}) removes the public--secret distinction, while a mixing construction (Propositions~\ref{prop:asymptotic_equivalence_variants})
replaces hypothesis-dependent jammer inputs with a common input at vanishing
asymptotic cost. Together, these establish Eq.~\eqref{eq: messsage 1 tmp1}. The IID tester constraint
introduces a subtler obstacle, as the nonconvex input family prevents the
same minimax argument. We therefore establish the collapse directly at the
operational level by constructing an explicit jammer strategy. The difficulty
is that a jammer input tailored to one tester choice need not be effective
against the others. Our construction overcomes this by combining these
tailored inputs into one common input, independent of both the tester's choice
and the hypothesis. The resulting bound holds uniformly over all IID tester
inputs, establishing Eq.~\eqref{eq: messsage 1 tmp2} without loss in the
asymptotic Stein exponent. When the jammer is also IID, however, different
information patterns can lead to different divergence characterizations and
require separate arguments. As a concrete structural specialization, we show that when the alternative channel is a replacer, all twelve vanishing-error Stein exponents reduce to the same additive, single-letter value (Theorem~\ref{thm:replacer-universal-single-letter}).

Our second main message concerns the sharpness of these exponents, as captured by the strong-converse property. In hypothesis testing, a strong converse is a desirable strengthening of the usual Stein-type statement: the convergence in Eq.~\eqref{eq: stein rate} does not rely on imposing a vanishing-error constraint. Such a result is known for jammer-input channel discrimination in general~\cite{fang2025adversarial}, whereas the corresponding tester-input problem remains open~\cite{fang2025towards}.
To establish strong converses for our minimax models, we develop a general
principle that upgrades a vanishing-error achievability result into a
strong-converse statement
(Theorem~\ref{thm:quantum-info-spectrum-strong-conv-general-scale}).  This is
conceptually surprising because achievability and strong converses concern
opposite sides of the operational problem and are typically proved by
separate arguments. 

\begin{boxmessage}[Achievability-to-strong-converse upgrade, informal]
Consider arbitrary sequences of quantum hypotheses, consisting either of
individual states or of sets of states.  If the vanishing-error Stein exponent
reaches the corresponding relative entropy rate, then this rate is also a strong converse.
\end{boxmessage}

This principle yields strong converses for an IID
tester against a secret entangled jammer and a trivial tester against a
secret IID jammer
(Theorem~\ref{thm: secret hypo hypothesis-aware iid vs general} and
Corollary~\ref{cor:trivial-tester-IID-jammer-strong-converse}). As byproducts, for composite
hypotheses, it resolves the strong-converse question posed by Berta,
Brand\~ao, and Hirche~\cite[Remark~2.2]{berta2021composite}
(Corollary~\ref{cor:BBH-composite-state-strong-converse}). It also strengthens
Lami's Stein lemmas for correlated and IID/arbitrarily varying
hypotheses~\cite{Lam25_Sanov} to strong-converse statements
(Corollary~\ref{cor:Lami-composite-state-strong-converse}).

Beyond these operational results, this work also develops minimax channel divergences with their basic structural properties. The framework and techniques developed here may
support future studies of quantum information tasks involving competing roles.  A summary of results is given in Figure~\ref{fig: summary}.

\begin{figure}[!ht]
  \centering
  \includegraphics[width=0.9\textwidth]{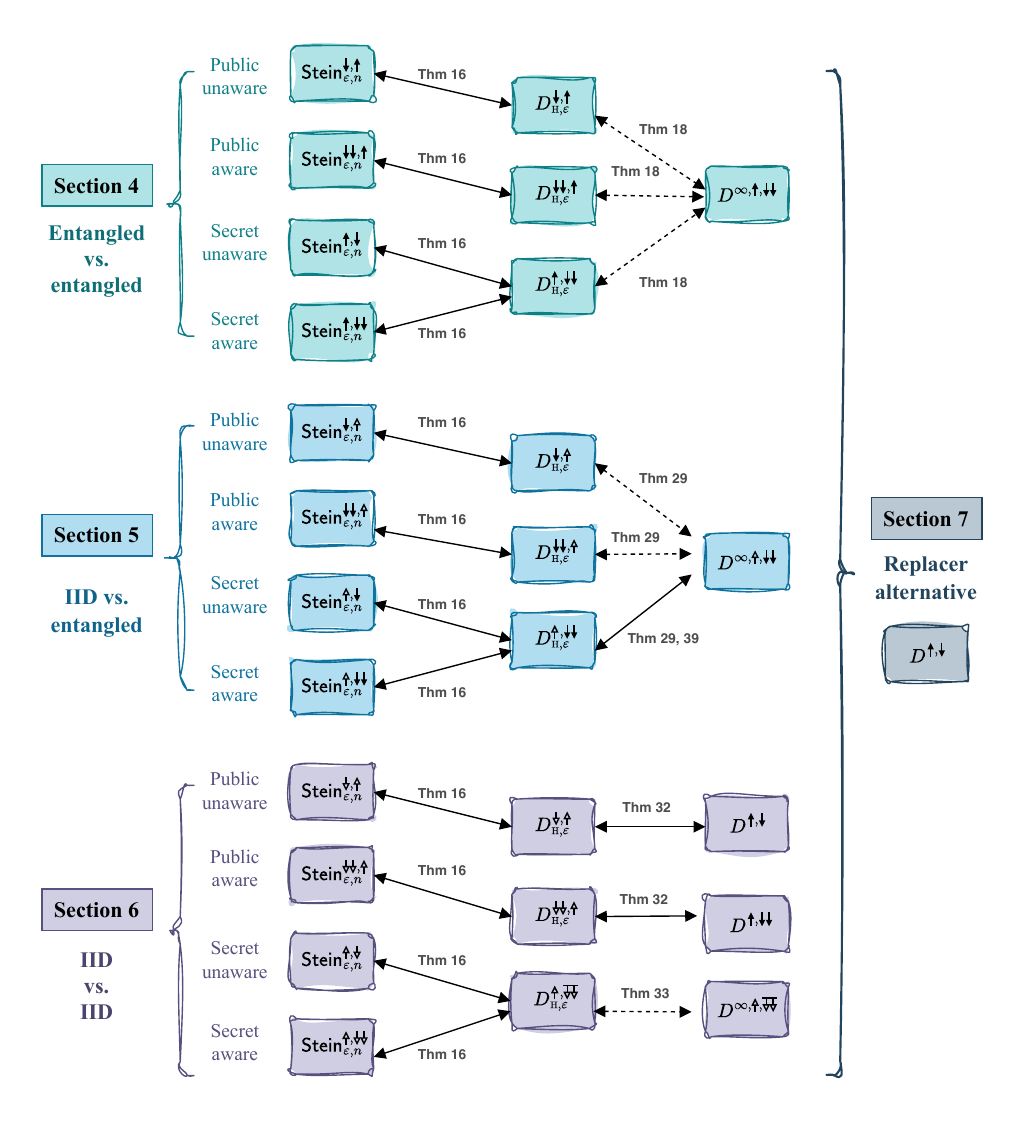}
  \caption{Summary of the main results. At finite blocklength, the twelve game
  models reduce to nine hypothesis-testing divergences
  (Definition~\ref{def:resource-dependent-block-divergences}). Asymptotically,
  they reduce to five Umegaki relative entropy rates
  (Definition~\ref{def:channel-divergence-regularization}). Solid and dashed
  connectors indicate strong- and weak-converse results, respectively. For
  a replacer alternative, all Stein exponents reduce to the same
  single-letter value.}
  \label{fig: summary}
\end{figure}

\subsection{Organization of the paper}

The remainder of the paper is organized as follows.
Section~\ref{sec: preliminaries} fixes the notation and reviews the state
divergences and one-shot hypothesis-testing bounds used throughout.
Section~\ref{sec: minimax divergence} develops tester-input, jammer-input, maximin,
and minimax channel divergences and establishes their finite-block operational
correspondence with the twelve discrimination games.
Section~\ref{sec:Channel discrimination: general tester} analyzes an entangled
tester against an entangled jammer and proves that the four information
patterns have a common Stein exponent given by a regularized divergence.
Section~\ref{sec:iid-tester-unrestricted-jammer} treats an IID tester
against an entangled jammer, for which the four operational exponents still
coincide.
Section~\ref{sec: iid inputs} studies an IID tester against an IID
jammer, deriving single-letter fixed-error exponents for public jammers and
  regularized convexified-domain vanishing-error exponents for secret jammers.
Section~\ref{sec:strong-converse-enhancement} develops an
achievability-to-strong-converse principle for state sequences and extends it
to sequences of state sets before applying these results to selected minimax
games and composite hypothesis testing.  The resulting
specializations include the trivial-tester game and fixed-error refinements of
known composite Stein lemmas.
Section~\ref{sec:replacer-alternatives} proves additivity and
single-letterization for replacer alternatives and briefly notes its implication in quantum illumination.
Finally, Section~\ref{sec:conclusion} summarizes the main
conclusions  and outlines directions for future work.

\section{Preliminaries} \label{sec: preliminaries}

This section fixes notation and collects the state divergences, several commonly used properties, and the one-shot hypothesis-testing bounds used throughout the paper.

\subsection{Notation and conventions}

All Hilbert spaces are finite dimensional, and all logarithms are taken to
base two.  We use calligraphic letters for maps and named sets, capital Roman
letters for quantum systems, and boldface letters for finite strings.  Thus,
$\mathbf{x}=(x_1,\ldots,x_n)$ denotes a finite string, whereas
$A^n=A_1\cdots A_n$ denotes a composite quantum system.  The Hilbert space
associated with a system $A$ is denoted by $\cH_A$.  We write $\LinOp(A)$,
$\HermOp(A)$, and $\PSD(A)$ for the sets of linear, Hermitian, and positive
semidefinite operators on $\cH_A$, respectively, and $\density(A)$ for the
set of density operators on $\cH_A$.
For $n$ copies of $A$, we use
$\sI(A^n):=\{\xi^{\ox n}:\xi\in\density(A)\}$ to denote the set of
tensor-power states and $\PER(A^n)$ the states
that are invariant under permutations of the $n$ tensor factors.  For
$\rho,\sigma\in\PSD(A)$, the notation $\rho\ll\sigma$ means that
$\operatorname{supp}\rho\subseteq\operatorname{supp}\sigma$.  Let
$\sE(A):=\{T:0\leq T\leq I_A\}$ denote the set of effects on $A$. A positive operator-valued measure (POVM) on $A$ is a set of effects $\Lambda=\{\Lambda_x\}_{x\in\cX}\subseteq\sE(A)$ such that $\sum_{x\in\cX}\Lambda_x=I_A$.

The sets of completely positive trace-preserving maps and completely positive
maps from $A$ to $B$ are denoted by $\CPTP(A\!:\!B)$ and $\CP(A\!:\!B)$,
respectively.  For $\cN\in\CP(AE\!:\!B)$ and a state
$\rho\in\density(RA)$, fixing the $RA$ input of $\cN$ to
$\rho$ induces a map $\cN_\rho:E\to RB$, defined by
\begin{align}\label{eq: induced channel}
\cN_\rho(\sigma)
:=\cN_{AE\to B}\ox \id_R(\rho_{RA}\ox\sigma_E),
\qquad \sigma\in\density(E).
\end{align}

\subsection{Quantum state divergences}

A functional $\DD:\density\times\PSD\to\reals\cup\{+\infty\}$ is a
\emph{quantum divergence} if it is monotone under CPTP maps. That is, for every CPTP map $\mathcal{E}$, every $\rho\in\density$, and
every $\sigma\in\PSD$, it satisfies the data-processing inequality
$\DD\infdiv*{\mathcal{E}(\rho)}{\mathcal{E}(\sigma)}
\leq\DD\infdiv*{\rho}{\sigma}$.

\begin{definition}\label{def:divergence-properties}
We record several common properties of quantum divergences for later reference, noting that a given divergence need not satisfy all of them.
\begin{itemize}[label={},leftmargin=0pt,labelsep=0pt,
                topsep=0.3em,itemsep=0.2em]
\item \emph{Direct-sum equality.}  For every
distribution $p_X$, every
$\{\rho_x\}_x\subseteq\density$, and every
$\{\sigma_x\}_x\subseteq\PSD$,
\begin{align}
\DD\bigg(\sum_{x} p_X(x)\,\proj{x}\otimes \rho_x\bigg\|\sum_{x} p_X(x)\,\proj{x}\otimes \sigma_x\bigg)
= \sum_{x} p_X(x)\,\DD\infdiv*{\rho_x}{\sigma_x}.
\end{align}
\item \emph{Normalization.}  For every
$\rho\in\density$, $\sigma\in\PSD$, and
$t>0$, $\DD(\rho\|t \sigma) = \DD(\rho\|\sigma) - \log t.$
\item \emph{Operator dominance.}  For every
$\rho\in\density$ and
$\sigma,\omega\in\PSD$ satisfying $\omega\leq\sigma$,
$\DD(\rho\|\sigma) \leq \DD(\rho\|\omega).$
\item \emph{Joint convexity.}  For every
$\lambda\in[0,1]$,
$\rho_0,\rho_1\in\density$, and $\sigma_0,\sigma_1\in\PSD$,
\begin{align}
    \DD\infdiv*{\lambda \rho_0 + (1-\lambda)\rho_1}{\lambda \sigma_0 + (1-\lambda)\sigma_1}
    \leq
    \lambda \DD\infdiv*{\rho_0}{\sigma_0} + (1-\lambda)\DD\infdiv*{\rho_1}{\sigma_1}.
\end{align}
\item \emph{Boundedness.}  For every
$\rho\in\density$ and $\sigma\in\PSD$, $\DD(\rho\|\sigma) \leq D_{\max}(\rho\|\sigma).$
\item \emph{Lower semicontinuity.}  Whenever
$\rho_n\to\rho$ and $\sigma_n\to\sigma$,
$\rho\in\density$, and
$\sigma\in\PSD$,
\begin{align}
    \DD(\rho\|\sigma) \leq \liminf_{n\to\infty} \DD(\rho_n\|\sigma_n),
\end{align}
provided that $\rho_n\in\density$ and $\sigma_n\in\PSD$ for every $n$.
\end{itemize}
\end{definition}

We recall the particular divergences
needed below.

\begin{definition}
For $\rho\in\density$ and $\sigma\in\PSD$, the Umegaki relative entropy is
defined by~\cite{umegaki1954conditional}
\begin{align}\label{eq: Umegaki}
    D(\rho\|\sigma):=
        \tr\!\left[\rho(\log \rho - \log \sigma)\right],
\end{align}
if $\rho\ll\sigma$, and $+\infty$ otherwise.
\end{definition}

\begin{definition}
For $\rho\in\density$, $\sigma\in\PSD$, the sandwiched R\'enyi divergence
is defined by~\cite{muller2013quantum,wilde2014strong}
\begin{align}
D_{\Sand,\alpha}(\rho\|\sigma)
:=
\dfrac{1}{\alpha-1}
\log\tr\!\left[
\left(\sigma^{\frac{1-\alpha}{2\alpha}}\rho
\sigma^{\frac{1-\alpha}{2\alpha}}\right)^\alpha
\right]
\end{align}
when either $0<\alpha<1$ and $\rho\not\perp\sigma$, or $\alpha>1$ and
$\rho\ll\sigma$; in all other cases, it is defined by $+\infty$.
\end{definition}

\begin{definition}
For $\rho\in\density$ and $\sigma\in\PSD$, the measured R\'enyi divergence is
defined by~\cite{Berta2017}
\begin{align}\label{eq: definition DM alpha}
D_{\Meas, \alpha} (\rho\|\sigma)
:= \sup_{(\cX,\Lambda)}
D_{\alpha}(P_{\rho,\Lambda}\|P_{\sigma,\Lambda}),
\end{align}
where $D_{\alpha}$ denotes the classical R\'enyi divergence and the supremum
is taken over all finite outcome sets $\cX$ and POVMs
$\Lambda=\{\Lambda_x\}_{x\in\cX}$.  The measurement outcome distributions are
$P_{\rho,\Lambda}(x):=\tr[\Lambda_x\rho]$ and
$P_{\sigma,\Lambda}(x):=\tr[\Lambda_x\sigma]$.
The measured relative entropy is similarly defined as~\cite{donald1986relative,hiai1991proper}
\begin{align}
D_{\Meas} (\rho\|\sigma) := \sup_{(\cX,\Lambda)}
D(P_{\rho,\Lambda}\|P_{\sigma,\Lambda}),
\end{align}
where $D$ denotes the classical relative entropy, or Kullback--Leibler
divergence.
\end{definition}

Due to the continuity, we will often regard $D_{\Meas,1} = D_{\Meas}$ and $D_{\Sand,1} = D$ throughout this work.

\begin{definition}
For $\rho\in\density$ and $\sigma\in\PSD$, the max-relative entropy is defined by~\cite{datta2009min}
\begin{align}
D_{\max}(\rho\|\sigma)
:=\inf\{\lambda\in\reals:\rho\leq 2^\lambda\sigma\},
\end{align}
if $\rho\ll\sigma$, and $+\infty$ otherwise.
\end{definition}

\begin{definition}
Let $\ve\in[0,1]$, $\rho\in\density$, and
$\sigma\in\PSD$.  The quantum hypothesis-testing relative entropy is defined by
$D_{\Hypo,\ve}(\rho\|\sigma):=-\log\beta_\ve(\rho\|\sigma)$, where
\begin{align}
\beta_\ve(\rho\|\sigma)
:=
\min_{0\leq T \leq I}
\bigl\{\tr[\sigma T]:\tr[\rho(I-T)]\le\ve\bigr\}.
\end{align}
\end{definition}

We recall the bounds that connect the hypothesis-testing divergence to
the Umegaki and sandwiched R\'enyi divergences. The lower bound follows by
combining~\cite[Proposition~3]{qi2018applications} with the
Araki--Lieb--Thirring inequality; see also~\cite{muller2013quantum}.
The two upper bounds are due to~\cite[Eq.~(2)]{wang2012one}
and~\cite[Lemma 5]{cooney2016strong}, respectively.

\begin{lemma}
\label{lem:one-shot-hypothesis-testing-bounds}
Let $1/2<\alpha<1<\gamma$, $\ve\in(0,1)$,
$\rho\in\density$, and $\sigma\in\PSD$.  Then
\begin{align}
D_{\Sand,\alpha}(\rho\|\sigma)
+\frac{\alpha}{\alpha-1}\log\frac1\ve
&\le D_{\Hypo,\ve}(\rho\|\sigma)
\le 
\begin{cases}
  \displaystyle\frac{D(\rho\|\sigma)+h(\ve)+\ve\log\tr\sigma}{1-\ve},
\\[0.2cm]
  \displaystyle D_{\Sand,\gamma}(\rho\|\sigma)
  +\frac{\gamma}{\gamma-1}\log\frac1{1-\ve},
\end{cases}
\label{eq:unified-one-shot-hypothesis-testing-bounds}
\end{align}
where
$h(\ve):=-\ve\log\ve-(1-\ve)\log(1-\ve)$ is the binary entropy.
\end{lemma}

The next lemma compares the measured and sandwiched R\'enyi divergences up to
a correction determined by the spectrum of the second argument. It follows by combining several known results in the literature; see
e.g.,~\cite[Lemma 16, 17]{fang2024generalized}.
\begin{lemma}\label{Lemma: DM and Sandwiched relation}
Let $\alpha\in[1/2,\infty)$.  For any $\rho\in\density$ and
$\sigma\in\PSD$,
\begin{align}
    D_{\Meas,\alpha}(\rho\|\sigma) \leq D_{\Sand,\alpha}(\rho\|\sigma) \leq D_{\Meas,\alpha}(\rho\|\sigma) + 2\log |\spec(\sigma)|,
\end{align}
where $|\spec(\sigma)|$ denotes the number of distinct eigenvalues of
$\sigma$. In particular, let $X\in\LinOp(\cH^{\otimes n})$ be permutation invariant, with
$\dim\cH=d$.  Then the number of distinct eigenvalues of $X$ satisfies
\begin{align}
    \size{\spec(X)} \leq (n+1)^{d} (n+d)^{d^2} = \poly(n).
\end{align}
\end{lemma}

\section{Minimax quantum channel divergences} \label{sec: minimax divergence}

This section introduces the minimax divergences used to analyze channel discrimination games. We first recall the tester-input and jammer-input channel divergences used in previous studies. We then introduce maximin and minimax divergences for dual-input channels, in which the tester and the jammer optimize over different input systems. Finally, we establish exact operational correspondences between these minimax divergences and the twelve channel-discrimination games introduced in Table~\ref{tab:finite-block-Stein-quantities}.

\subsection{Two extreme channel divergences}

\begin{definition}[Tester-input channel divergence~\cite{leditzky2018approaches}]
Let $\DD$ be a quantum divergence.
For any $\cN \in \CPTP({A}\!:\!{B})$ and $\cM \in \CP({A}\!:\!{B})$, the tester-input channel divergence is defined by
\begin{align}
    \supDiv\infdiv*{\cN}{\cM}:=
    \sup_{\rho \in \density(AR)} \DD\infdiv*{\cN(\rho)}{\cM(\rho)},
\end{align}
where ${R}$ is a reference system of arbitrary dimension.
\end{definition}

As a consequence of purification, data processing, and the Schmidt decomposition, the supremum can be restricted to be with respect to pure states and the reference system $R$ is isomorphic to $A$.

\begin{definition}[Jammer-input channel divergence~\cite{fang2025adversarial}]
\label{def:jammer-input-channel-divergences}
    Let $\DD$ be a quantum divergence.
For any $\cN \in \CPTP({E}\!:\!{B})$ and
$\cM \in \CP({E}\!:\!{B})$, the jammer-input channel
divergences are defined by
\begin{align}
    \iinfDiv\infdiv*{\cN}{\cM}
    &:= \ \inf_{\substack{\sigma,\,\omega\in\density(E)}}
   \ \DD\infdiv*{\cN(\sigma)}{\cM(\omega)},\\
    \infDiv\infdiv*{\cN}{\cM}
    &:= \ \ \ \inf_{\sigma\in\density(E)}
    \ \ \ \DD\infdiv*{\cN(\sigma)}{\cM(\sigma)}.
\end{align}
Here, double downward arrows indicate optimization over distinct variables, while a single downward arrow indicates optimization over a common variable.
\end{definition}

  In the jammer-input channel divergences, introducing an auxiliary reference system ${R}$ is unnecessary, as it does not change the value of the divergences.
  More precisely, we have
  \begin{align}
      \iinfDiv\infdiv*{\cN}{\cM} = \inf_{\substack{\sigma,\,\omega\in\density({ER})}}
  \DD\infdiv*{\cN(\sigma)}{\cM(\omega)},
  \end{align}
  where ${R}$ is a reference system of arbitrary dimension.
  The ``$\geq$'' direction is trivial, as we can choose $|R| = 1$. The ``$\leq$'' direction is a direct consequence of the data-processing inequality, by applying the partial trace over the reference system ${R}$. Similarly, we have
  \begin{align}
      \infDiv\infdiv*{\cN}{\cM} = \inf_{\sigma\in\density({ER})}
  \DD\infdiv*{\cN(\sigma)}{\cM(\sigma)},
  \end{align}
  where ${R}$ is a reference system of arbitrary dimension.

The variant
$\DD^{\downarrow\downarrow}(\cN\|\cM)$ was studied in~\cite{fang2025adversarial} whereas $\DD^{\downarrow}(\cN\|\cM)$ is introduced here. In general, the two quantities can differ, as the following example shows for the quantum relative entropy.
\begin{boxexample}\label{exam:1}
Let $\{\ket{0},\ket{1}\}$ be the computational basis. Consider two quantum channels $\cN,\cM\in \CPTP(E\!:\!B)$ with $\dim \cH_E=\dim \cH_B=2$ as follows:
\begin{align}
\cN(\rho)&:= \langle 0|\rho|0\rangle |0\rangle \langle 0|
+\langle 1|\rho|1\rangle
 \pi_2, \\
\cM(\rho)&:= \langle 0|\rho|0\rangle \pi_2
+\langle 1|\rho|1\rangle |1\rangle \langle 1|,
\end{align}
where $\pi_2$ is the maximally mixed state of dimension $2$.
In this case,
$D^{\downarrow}(\cN\|\cM )  > 
D^{\downarrow\downarrow}(\cN\|\cM ) =0.
$
\end{boxexample}

\subsection{Maximin and minimax channel divergences}

For the analysis of the discrimination games, we introduce in a unified framework the minimax and maximin divergences used throughout the paper; these are summarized in Definition~\ref{def:resource-dependent-block-divergences}. 

\begin{boxdefinition}[Maximin and minimax channel divergences]
\label{def:resource-dependent-block-divergences}
Let $\cN\in\CPTP(AE\!:\!B)$, $\cM\in\CP(AE\!:\!B)$,
$n\in\NN$. For states
$\rho_n\in\density((AR)^n)$ and
$\sigma_n,\omega_n\in\density(E^n)$, denote
\begin{align}
\DD_n(\rho_n;\sigma_n,\omega_n)
&:=
\DD\!\left(
\cN^{\ox n}(\rho_n\ox\sigma_n)
\big\|\,
\cM^{\ox n}(\rho_n\ox\omega_n)
\right).
\end{align}
The following table lists eleven different variants of channel divergences. The arrow codes follow the same pattern as in Table~\ref{tab:finite-block-Stein-quantities}. An upward arrow denotes the tester's supremum and a downward arrow the jammer's infimum. Double downward arrows indicate distinct optimization variables and a single downward arrow a common one. Solid arrows denote optimization over all density matrices $\density$, hollow arrows over tensor-power states $\sI$, and barred hollow arrows over their convex hull $\conv(\sI)$.  

\vspace{0.4cm}
For example, the first row reads
\begin{align}
\DD^{\uparrow,\downarrow}(\cN^{\ox n}\|\cM^{\ox n})
:=
\sup_{\rho_n\in\density((AR)^n)}\ \ \
\inf_{\tau_n\in\density(E^n)}\ \ \
\DD_n(\rho_n;\tau_n,\tau_n).
\end{align}

\begin{table}[H]
\centering
\normalsize
\setlength{\tabcolsep}{2pt}
\renewcommand{\arraystretch}{1}
\begin{tabularx}{\linewidth}{
@{}
>{\centering\arraybackslash}m{0.22\linewidth}
>{\centering\arraybackslash}m{0.25\linewidth}
>{\centering\arraybackslash}m{0.25\linewidth}
>{\centering\arraybackslash}X
@{}}
\toprule
Channel divergence
&
Outer optimization
&
Inner optimization
&
Objective
\\
\midrule
$\DD^{\uparrow,\downarrow}$
& $\displaystyle\sup_{\rho_n\in\density((AR)^n)}$
& $\displaystyle\inf_{\tau_n\in\density(E^n)}$
& $\DD_n(\rho_n;\tau_n,\tau_n)$
\\
$\DD^{\uparrow,\downarrow\downarrow}$
& $\displaystyle\sup_{\rho_n\in\density((AR)^n)}$
& $\displaystyle\inf_{\sigma_n,\omega_n\in\density(E^n)}$
& $\DD_n(\rho_n;\sigma_n,\omega_n)$
\\
$\DD^{\downarrow,\uparrow}$
& $\displaystyle\inf_{\tau_n\in\density(E^n)}$
& $\displaystyle\sup_{\rho_n\in\density((AR)^n)}$
& $\DD_n(\rho_n;\tau_n,\tau_n)$
\\
$\DD^{\downarrow\downarrow,\uparrow}$
& $\displaystyle\inf_{\sigma_n,\omega_n\in\density(E^n)}$
& $\displaystyle\sup_{\rho_n\in\density((AR)^n)}$
& $\DD_n(\rho_n;\sigma_n,\omega_n)$
\\
\midrule
$\DD^{\emptyuparrow,\downarrow}$
& $\displaystyle\sup_{\rho_n \in \sI((AR)^n)}$
& $\displaystyle\inf_{\tau_n\in\density(E^n)}$
& $\DD_n(\rho_n;\tau_n,\tau_n)$
\\
$\DD^{\emptyuparrow,\downarrow\downarrow}$
& $\displaystyle\sup_{\rho_n \in \sI((AR)^n)}$
& $\displaystyle\inf_{\sigma_n,\omega_n\in\density(E^n)}$
& $\DD_n(\rho_n;\sigma_n,\omega_n)$
\\
$\DD^{\downarrow,\emptyuparrow}$
& $\displaystyle\inf_{\tau_n\in\density(E^n)}$
& $\displaystyle\sup_{\rho_n \in \sI((AR)^n)}$
& $\DD_n(\rho_n;\tau_n,\tau_n)$
\\
$\DD^{\downarrow\downarrow,\emptyuparrow}$
& $\displaystyle\inf_{\sigma_n,\omega_n\in\density(E^n)}$
& $\displaystyle\sup_{\rho_n \in \sI((AR)^n)}$
& $\DD_n(\rho_n;\sigma_n,\omega_n)$
\\
\midrule
$\DD^{\emptydownarrow,\emptyuparrow}$
& $\displaystyle\inf_{\tau_n\in\sI(E^n)}$
& $\displaystyle\sup_{\rho_n \in\sI((AR)^n)}$
& $\DD_n(\rho_n;\tau_n,\tau_n)$
\\
$\DD^{\emptydownarrow\emptydownarrow,\emptyuparrow}$
& $\displaystyle\inf_{\sigma_n,\omega_n\in\sI(E^n)}$
& $\displaystyle\sup_{\rho_n \in\sI((AR)^n)}$
& $\DD_n(\rho_n;\sigma_n,\omega_n)$
\\
$\DD^{\emptyuparrow,\convdownarrow\convdownarrow}$
& $\displaystyle\sup_{\rho_n\in\sI((AR)^n)}$
& $\displaystyle\inf_{\sigma_n,\omega_n\in\conv(\sI(E^n))}$
& $\DD_n(\rho_n;\sigma_n,\omega_n)$
\\
\bottomrule
\end{tabularx}
\label{tab:resource-dependent-divergence-variants}
\end{table}
\end{boxdefinition}

\begin{remark}
\label{rem: jammer divergence sup of jammer-input}
The maximin divergences interpolate between the tester-input and jammer-input
channel divergences.  
Specifically, trivializing $E$ recovers the tester-input channel divergence, whereas trivializing $A$ recovers the corresponding jammer-input channel divergence.\footnote{
For any quantum divergence $\DD$,
$\DD(\rho\ox\omega\|\sigma\ox\omega)=\DD(\rho\|\sigma)$;
see e.g.,~\cite[Proposition~7.14]{khatri2024principlesquantumcommunicationtheory}.} Moreover, we have
\begin{align}
\DD^{\uparrow,(\star)}(\cN\|\cM)
&=
\sup_{\rho\in\density(AR)}
\DD^{(\star)}(\cN_\rho\|\cM_\rho),\qquad (\star) \in \{(\downarrow),(\downarrow\downarrow)\}
\label{eq: DM jammer superadditivity tmp1}
\end{align}
where the induced maps $\cN_\rho, \cM_\rho$ are defined in Eq.~\eqref{eq: induced channel}.
\end{remark}

\begin{definition}[Regularization]
\label{def:channel-divergence-regularization}
For any quantum channel divergence $\DD$, we define its lower and upper regularization limits by
\begin{align}
\underline{\DD}^{\infty}\infdiv*{\cN}{\cM}
&:=
\liminf_{n\to\infty}\frac1n
\DD\infdiv*{\cN^{\ox n}}{\cM^{\ox n}},
\\
\overline{\DD}^{\infty}\infdiv*{\cN}{\cM}
&:=
\limsup_{n\to\infty}\frac1n
\DD\infdiv*{\cN^{\ox n}}{\cM^{\ox n}},
\end{align}
respectively.
When the limits exist, simply write $\DD^\infty$.
\end{definition}

The following lemma ensures the finiteness of the minimax/maximin channel divergences that will be
used repeatedly throughout the paper.

\begin{boxlemma}
\label{lem:uniform-block-bounds}
Let $\cN\in\CPTP(AE\!:\!B)$ and $\cM\in\CP(AE\!:\!B)$ satisfy
$D_{\max}^{\uparrow}(\cN\|\cM)<\infty$.  Let $\DD$ be a quantum
divergence that vanishes on identical states and satisfies the normalization
and boundedness properties in Definition~\ref{def:divergence-properties}.
Choose $\kappa>1$ and define $\mu$ by
\begin{align}
D_{\max}^{\uparrow}(\cN\|\cM)
<\log\kappa,
\qquad
\mu:=\|\cM^\dagger(I_B)\|_\infty.
\label{eq:iid-general-kappa-mu}
\end{align}
Then, for every $n\in\NN$, $\rho_n\in\density((AR)^n)$, and
$\sigma_n,\omega_n,\tau_n\in\density(E^n)$,
\begin{align}
\DD_n(\rho_n;\sigma_n,\omega_n)
&\ge
-n\log\mu,
\quad\text{and}\quad
\DD_n(\rho_n;\tau_n,\tau_n)
\le
n\log\kappa.
\label{eq:uniform-block-pointwise-divergence-bounds}
\end{align}
Consequently, every arrow code $(\star)$ in
Definition~\ref{def:resource-dependent-block-divergences} satisfies
\begin{align}
-n\log\mu
\le
\DD^{(\star)}(\cN^{\ox n}\|\cM^{\ox n})
\le
n\log\kappa.
\label{eq:fixed-block-uniform-finite-bounds}
\end{align}
\end{boxlemma}

\begin{proof}
Evaluating the definition of
$D_{\max}^{\uparrow}(\cN\|\cM)$ on a normalized maximally entangled input
gives $J_{\cN}\le\kappa J_{\cM}$.  The Choi criterion therefore yields
$\cN\le\kappa\cM$ in the completely positive order.  Tensoring this
relation gives, for every $\rho_n\in\density((AR)^n)$ and
$\tau_n\in\density(E^n)$,
\begin{align}
\cN^{\ox n}(\rho_n\ox\tau_n)
&\le
\kappa^n\cM^{\ox n}(\rho_n\ox\tau_n).
\end{align}
Trace preservation of $\cN$ and the adjoint
relation further imply, for every $\omega_n\in\density(E^n)$,
\begin{align}
\kappa^{-n}
&\le
\tr\!\left[\cM^{\ox n}(\rho_n\ox\omega_n)\right]
\le
\left\|(\cM^\dagger(I_B))^{\ox n}\right\|_\infty
=
\mu^n.\label{eq: trace estimation}
\end{align}
In particular, $\mu>0$.
The two pointwise bounds now follow from
\begin{align}
\DD_n(\rho_n;\sigma_n,\omega_n)
&\ge
\DD\!\left(
1\,\middle\|\,
\tr\!\left[
\cM^{\ox n}(\rho_n\ox\omega_n)
\right]
\right)
=
-\log\tr\!\left[
\cM^{\ox n}(\rho_n\ox\omega_n)
\right]
\ge
-n\log\mu,
\\
\DD_n(\rho_n;\tau_n,\tau_n)
&\le
D_{\max}\!\left(
\cN^{\ox n}(\rho_n\ox\tau_n)
\,\middle\|\,
\cM^{\ox n}(\rho_n\ox\tau_n)
\right)
\le
n\log\kappa.
\end{align}
The first chain uses data processing under the trace map, normalization, and
the fact that $\DD$ vanishes on identical states.  The second uses boundedness
and the preceding completely positive domination.
\end{proof}

\subsection{Operational correspondence}

We show that the optimal error exponents in the twelve channel-discrimination games of Table~\ref{tab:finite-block-Stein-quantities} are exactly characterized by one of the maximin and minimax channel divergences induced by the quantum hypothesis-testing divergence. These correspondences are summarized in Theorem~\ref{thm:unified-operational-correspondence} and will be the starting point to analyze the asymptotic behavior of the games in later sections.

\begin{boxtheorem}[Operational correspondence]
\label{thm:unified-operational-correspondence}
Let $\cN\in\CPTP(AE\!:\!B)$, $\cM\in\CP(AE\!:\!B)$,
$n\in\NN$, and $\ve\in(0,1)$.  The correspondence between the finite-blocklength Stein exponent in each channel discrimination game and the maximin or minimax channel divergence is summarized below:
\begin{center}
\normalfont\normalsize
\setlength{\tabcolsep}{4pt}
\renewcommand{\arraystretch}{1.3}
\begin{tabularx}{\linewidth}{
@{}
>{\centering\arraybackslash}m{0.28\linewidth}
>{\centering\arraybackslash}m{0.15\linewidth}
>{\centering\arraybackslash}m{0.25\linewidth}
>{\centering\arraybackslash}X
@{}}
\toprule
\shortstack{Tester / jammer inputs}
&
Visibility
&
Hypothesis awareness
&
Operational identity
\\
\midrule
\multirow[c]{4}{*}{\shortstack{Entangled tester\\[0.2cm] vs. entangled jammer}}
&
Public
&
Unaware
&
$\Stein_{\ve,n}^{\downarrow,\uparrow}
  =D_{\Hypo,\ve}^{\downarrow,\uparrow}$
\\
&
Public
&
Aware
&
$\Stein_{\ve,n}^{\downarrow\downarrow,\uparrow}
  =D_{\Hypo,\ve}^{\downarrow\downarrow,\uparrow}$
\\
\noalign{\vskip0.5ex}
\cdashline{2-4}[2pt/2pt]
\noalign{\vskip0.5ex}
&
Secret
&
Unaware
&
$\Stein_{\ve,n}^{\uparrow,\downarrow}
  =D_{\Hypo,\ve}^{\uparrow,\downarrow\downarrow}$
\\
&
Secret
&
Aware
&
$\Stein_{\ve,n}^{\uparrow,\downarrow\downarrow}
  =D_{\Hypo,\ve}^{\uparrow,\downarrow\downarrow}$
\\
\midrule
\multirow[c]{4}{*}{\shortstack{IID tester \\[0.2cm] vs. entangled jammer}}
&
Public
&
Unaware
&
$\Stein_{\ve,n}^{\downarrow,\emptyuparrow}
  =D_{\Hypo,\ve}^{\downarrow,\emptyuparrow}$
\\
&
Public
&
Aware
&
$\Stein_{\ve,n}^{\downarrow\downarrow,\emptyuparrow}
  =D_{\Hypo,\ve}^{\downarrow\downarrow,\emptyuparrow}$
\\
\noalign{\vskip0.5ex}
\cdashline{2-4}[2pt/2pt]
\noalign{\vskip0.5ex}
&
Secret
&
Unaware
&
$\Stein_{\ve,n}^{\emptyuparrow,\downarrow}
  =D_{\Hypo,\ve}^{\emptyuparrow,\downarrow\downarrow}$
\\
&
Secret
&
Aware
&
$\Stein_{\ve,n}^{\emptyuparrow,\downarrow\downarrow}
  =D_{\Hypo,\ve}^{\emptyuparrow,\downarrow\downarrow}$
\\
\midrule
\multirow[c]{4}{*}{\shortstack{IID tester \\[0.2cm] vs. IID jammer}}
&
Public
&
Unaware
&
$\Stein_{\ve,n}^{\emptydownarrow,\emptyuparrow}
  =D_{\Hypo,\ve}^{\emptydownarrow,\emptyuparrow}$
\\
&
Public
&
Aware
&
$\Stein_{\ve,n}^{\emptydownarrow\emptydownarrow,\emptyuparrow}
  =D_{\Hypo,\ve}^{\emptydownarrow\emptydownarrow,\emptyuparrow}$
\\
\noalign{\vskip0.5ex}
\cdashline{2-4}[2pt/2pt]
\noalign{\vskip0.5ex}
&
Secret
&
Unaware
&
$\Stein_{\ve,n}^{\emptyuparrow,\emptydownarrow}
  =D_{\Hypo,\ve}^{\emptyuparrow,
    \convdownarrow\convdownarrow}$
\\
&
Secret
&
Aware
&
$\Stein_{\ve,n}^{\emptyuparrow,\emptydownarrow\emptydownarrow}
  =D_{\Hypo,\ve}^{\emptyuparrow,
    \convdownarrow\convdownarrow}$
\\
\bottomrule
\end{tabularx}
\end{center}
By our notational conventions, every finite blocklength Stein exponent in the above table is evaluated at $(\cN,\cM)$, and every hypothesis testing channel divergence is evaluated at
$(\cN^{\ox n}\|\cM^{\ox n})$.
\end{boxtheorem}

For the public games, the correspondence follows directly from the definitions.
Once the jammer's announced states are fixed, the remaining tester
optimization is exactly the fixed state hypothesis-testing problem.  Restoring
the input optimizations therefore give the result.

The secret games are more subtle because the tester must choose a single test
without knowing the jammer input.  This test must satisfy the type-I constraint
uniformly over all null inputs, while its type-II performance is assessed
against the worst alternative input.  From the tester's worst-case perspective,
the secret jammer's knowledge of the true hypothesis does not change the
resulting value.  Thus, within each input structure, the two secret games are
characterized by the same effective double-down-arrow divergence.  This
identification is substantive: Example~\ref{eg: DH single double down arrow}
shows that the single- and double-down-arrow hypothesis-testing divergences can
differ strictly.

For a secret IID jammer, one must additionally address the nonconvexity of the
tensor-power input set.  Because both error probabilities depend affinely on
the jammer state, replacing this set by its convex hull leaves the game
unchanged.  The resulting output families are compact and convex, so the
composite hypothesis-testing identity
of~\cite[Lemma~31]{fang2024generalized} applies.

\begin{proof}
  The proof is relatively direct. Nevertheless, identifying the divergence
associated with each operational game requires some care: the twelve games
yield only nine potentially distinct divergence expressions.
The public correspondences follow directly from the definitions, whereas the secret correspondences require additional steps, which we detail below.

\prooftag{Public jammer.} For a public jammer, fix its announced state or state pair.  For
every fixed tester input, the optimization over the test is precisely the
state hypothesis-testing error,
\begin{align}
\sup_{\substack{T\in\sE((BR)^n)\\
\alpha_T(\rho_n,\sigma_n)\le\ve}}
-\log\beta_T(\rho_n,\omega_n)
=
D_{\Hypo,\ve}\!\left(
\cN^{\ox n}(\rho_n\ox\sigma_n)
\,\|\,
\cM^{\ox n}(\rho_n\ox\omega_n)
\right).
\label{eq:public-fixed-input-hypothesis-testing-correspondence}
\end{align}
Setting $\sigma_n=\omega_n=\tau_n$ gives the hypothesis-unaware case.  Restoring the
inner tester supremum and outer jammer infimum over the domains proves all six public formulas.

\prooftag{Secret jammer.}
Fix the tester input $\rho_n$.  Let $\cK_n=\density(E^n)$ for an entangled
jammer and $\cK_n=\sI(E^n)$ for an IID jammer.  The fixed-input value
of the secret-hypothesis-unaware game can be written as
\begin{align}
V_{\ve,n}^{\rm sec}(\rho_n)
&:=
\sup_{\substack{T\in\sE((BR)^n)\\
\alpha_T(\rho_n,\tau_n)\le\ve,\
\forall\,\tau_n\in\cK_n}}
\ \inf_{\tau_n\in\cK_n}
\bigl[-\log\beta_T(\rho_n,\tau_n)\bigr]\\
&=
\sup_{\substack{T\in\sE((BR)^n)\\
\alpha_T(\rho_n,\sigma_n)\le\ve,\
\forall\,\sigma_n\in\cK_n}}
\ \inf_{\omega_n\in\cK_n}
\bigl[-\log\beta_T(\rho_n,\omega_n)\bigr].
\label{eq:secret-hypothesis-unaware-hypothesis-aware-fixed-input}
\end{align}
The variable $\tau_n$ in the uniform type-I constraint and the variable
$\tau_n$ in the type-II infimum have disjoint quantifier scopes.  Renaming
them separately as $\sigma_n$ and $\omega_n$ gives the second line, which is
exactly the fixed-input value of the secret-hypothesis-aware game.  Thus, the two secret
information patterns define the same robust composite test under the
present jammer-input error criterion.

Set $\overline{\cK}_n:=\conv(\cK_n)$.  For every fixed test, the type-I and
type-II errors are affine in the jammer state.  Hence the uniform type-I
constraint is unchanged when $\cK_n$ is replaced by
$\overline{\cK}_n$.  Moreover,
$\inf_{\omega_n}[-\log\beta_T(\rho_n,\omega_n)]
=-\log\sup_{\omega_n}\beta_T(\rho_n,\omega_n)$, so the jammer-input type-II
objective is also unchanged by convexification.  We therefore obtain
\begin{align}
V_{\ve,n}^{\rm sec}(\rho_n)
&=
\sup_{\substack{T\in\sE((BR)^n)\\
\alpha_T(\rho_n,\sigma_n)\le\ve,\
\forall\,\sigma_n\in\overline{\cK}_n}}
\ \inf_{\omega_n\in\overline{\cK}_n}
\bigl[-\log\beta_T(\rho_n,\omega_n)\bigr]\\
&=
-\log
\inf_{\substack{T\in\sE((BR)^n)\\
\alpha_T(\rho_n,\sigma_n)\le\ve,\
\forall\,\sigma_n\in\overline{\cK}_n}}
\ \sup_{\omega_n\in\overline{\cK}_n}
\beta_T(\rho_n,\omega_n)\\
&=
\inf_{\sigma_n,\omega_n\in\overline{\cK}_n}
D_{\Hypo,\ve}\!\left(
\cN^{\ox n}(\rho_n\ox\sigma_n)
\,\|\,
\cM^{\ox n}(\rho_n\ox\omega_n)
\right),
\label{eq:secret-error-composite-rewriting}
\end{align}
where the last equality follows from the composite hypothesis-testing
identity in~\cite[Definition~18 and Lemma~31]{fang2024generalized}, applied
to the sets
\begin{align}
  \left\{
\cN^{\ox n}(\rho_n\ox\sigma_n):
\sigma_n\in\overline{\cK}_n
\right\}, \qquad \left\{
\cM^{\ox n}(\rho_n\ox\omega_n):
\omega_n\in\overline{\cK}_n
\right\}.\label{eq:secret-output-families}
\end{align}
These output sets are compact and convex.  
For an entangled jammer, $\overline{\cK}_n=\density(E^n)$, which gives the
double solid-down-arrow divergence.  For an IID jammer,
$\overline{\cK}_n=\conv(\sI(E^n))$, which gives the double barred
hollow-down-arrow divergence.  Taking the tester supremum over
$\density((AR)^n)$ or $\sI((AR)^n)$, as appropriate, proves all six secret
formulas.
\end{proof}

\begin{example}
\label{eg: DH single double down arrow}
The single- and double-down-arrow hypothesis-testing divergences need not
coincide at finite blocklength.  Consider the channels in
Example~\ref{exam:1} at $n=1$ and $\ve=1/2$.  For a common jammer input
$\tau$, let $p:=\langle 0|\tau|0\rangle$.  Direct evaluation gives
\begin{align}
D_{\Hypo,1/2}\bigl(\cN(\tau)\|\cM(\tau)\bigr)
=\log\frac{2(1+p)}{p},
\end{align}
with value $+\infty$ at $p=0$.  The right-hand side is minimized at $p=1$,
so $D_{\Hypo,1/2}^{\uparrow,\downarrow}(\cN\|\cM)=2$.  By contrast, the
independent inputs $\sigma=|1\rangle\langle 1|$ and
$\omega=|0\rangle\langle 0|$ satisfy
$\cN(\sigma)=\cM(\omega)=\pi_2$, and hence their hypothesis-testing
divergence equals $1$.  Moreover, the test $T=I/2$ is feasible for every pair
of normalized states and implies $D_{\Hypo,1/2}\geq 1$, so this choice is
optimal.  Therefore,
\begin{align}
D_{\Hypo,1/2}^{\uparrow,\downarrow}(\cN\|\cM)
=2>1
=D_{\Hypo,1/2}^{\uparrow,\downarrow\downarrow}(\cN\|\cM).
\label{eq:hypothesis-testing-single-double-separation}
\end{align}
This shows why the secret-game identities in
Theorem~\ref{thm:unified-operational-correspondence} must use the double-down-arrow divergence.
\end{example}

\section{Entangled tester against entangled jammer}
\label{sec:Channel discrimination: general tester}

In this section, we consider an entangled tester against an entangled jammer.  We show that all four information patterns, regardless of the jammer’s visibility and awareness, yield the same optimal type-II error exponent for the tester.
The following theorem is the main result of this section.

\begin{boxtheorem}[Entangled tester against entangled jammer]
\label{thm:general-general-Stein}
Let $\cN\in\CPTP(AE\!:\!B)$ and $\cM\in\CP(AE\!:\!B)$ satisfy
$D_{\max}^{\uparrow}(\cN\|\cM)<\infty$. For each of the four information patterns
\begin{align}
(\star)\in
\bigl\{
(\downarrow,\uparrow),
(\downarrow\downarrow,\uparrow),
(\uparrow,\downarrow),
(\uparrow,\downarrow\downarrow)
\bigr\},
\end{align}
the asymptotic Stein exponent is given by
\begin{align}
&\lim_{\ve\to 0^+}\liminf_{n\to\infty}
\frac1n
\Stein_{\ve,n}^{(\star)}(\cN,\cM)
=
\lim_{\ve\to 0^+}\limsup_{n\to\infty}
\frac1n
\Stein_{\ve,n}^{(\star)}(\cN,\cM)
=D^{\infty,\uparrow,\downarrow\downarrow}\infdiv*{\cN}{\cM}.
\label{eq:general-general-Stein}
\end{align}
\end{boxtheorem}

The coincidence asserted above comes from the interplay of two structural
properties of the minimax channel divergences. First, the public and secret settings differ in whether the jammer's
input-selection rule is revealed to the tester. The minimax property in
Proposition~\ref{prop: jammer divergence minimax} shows that this distinction
does not affect the value, because the tester and jammer optimization orders can
be exchanged. Second, the hypothesis-aware and hypothesis-unaware settings
differ in whether the jammer may use different input states under the two
hypotheses. The asymptotic equivalence in
Proposition~\ref{prop:asymptotic_equivalence_variants} shows that, for any
hypothesis-aware input pair, a suitable convex combination can yield an almost
equivalent hypothesis-unaware performance. Together, these two properties imply
the same asymptotic exponent for all four settings, as summarized in
Remark~\ref{rem: four divergence conincidence}.

In the following, we first establish the two key ingredients in Subsection~\ref{sec: Minimax property and asymptotic equivalence}, then several other technical
preparations in Subsection~\ref{sec: Other technical preparation}, and finally prove Theorem~\ref{thm:general-general-Stein} in
Subsection~\ref{sec:proof-of-general-general-Stein}.

\subsection{Minimax property and asymptotic equivalence}
\label{sec: Minimax property and asymptotic equivalence}

\subsubsection*{Minimax identity}

Under mild assumptions on the underlying divergence $\DD$, the maximin and minimax channel
divergences coincide. Thus the discrimination game has a well-defined value, and the order of optimization may be interchanged. This underlies the operational equivalence between the public and secret games later. 

\begin{boxproposition}[Minimax property]
\label{prop: jammer divergence minimax}
Let $\DD$ be a quantum divergence satisfying direct-sum equality and lower
semicontinuity as specified in Definition~\ref{def:divergence-properties}.
For any $\cN \in \CPTP({AE}\!:\!{B})$ and $\cM \in \CP({AE}\!:\!{B})$, the following minimax identities hold
\begin{align}
    \DD^{\uparrow,\downarrow\downarrow}
    \infdiv*{\cN}{\cM}
    & = \DD^{\downarrow\downarrow,\uparrow}  \infdiv*{\cN}{\cM},\qquad
    \DD^{\uparrow,\downarrow}
    \infdiv*{\cN}{\cM}
     = \DD^{\downarrow,\uparrow}  \infdiv*{\cN}{\cM}. \label{eq: maximin equal minimax 1}
\end{align}
In particular, these hold for Umegaki relative entropy $D$ and measured relative entropy $D_{\Meas}$. They also hold for the measured R\'enyi divergence
$D_{\Meas,\alpha}$ for
$\alpha\in(0,1)\cup(1,\infty)$ and the sandwiched R\'enyi divergence
$D_{\Sand,\alpha}$ for
$\alpha\in[1/2,1)\cup(1,\infty)$ by going through the same proof via their quasi-divergences.
\end{boxproposition}

\begin{proof}
We prove the case for ``$\downarrow\downarrow$''; the other case
follows by the same argument after restricting the two hypotheses to a common jammer input.
Thus, it suffices to establish the following chain of equalities, which first removes the
reference system and then interchanges the tester and jammer optimizations:
\begin{align}
    \DD^{\uparrow,\downarrow\downarrow}
    \infdiv*{\cN}{\cM}
     :=&
    \sup_{\rho \in \density(AR)} \ \ \inf_{\sigma,\,\omega \in \density(E)} \ \  \DD\infdiv*{\cN(\rho\ox\sigma)}{\cM(\rho\ox\omega)} \label{eq: minimax 1}\\
    = & \sup_{\rho_{A} \in \density(A)} \ \  \inf_{\sigma,\,\omega \in \density(E)} \ \  \DD\infdiv*{\cN(\proj{\rho}\ox\sigma)}{\cM(\proj{\rho}\ox\omega)} \label{eq: minimax 2}\\
    = &  \inf_{\sigma,\,\omega \in \density(E)} \ \  \sup_{\rho_{A} \in \density(A)} \ \  \DD\infdiv*{\cN(\proj{\rho}\ox\sigma)}{\cM(\proj{\rho}\ox\omega)} \label{eq: minimax 3}\\
    = &  \inf_{\sigma,\,\omega \in \density(E)} \ \  \sup_{\rho \in \density(AR)}  \ \  \DD\infdiv*{\cN(\rho\ox\sigma)}{\cM(\rho\ox\omega)} \label{eq: minimax 4}\\
    = & \ \ \DD^{\downarrow\downarrow,\uparrow}  \infdiv*{\cN}{\cM},
\end{align}
where $\ket{\rho}$ is a purification of $\rho_{A}$ with $R \simeq A$. Due to the isometric relation between purifications and isometric invariance of the divergence, the choice of purification here does not change the value.

The above claim can be proved as follows.
    For any fixed $\sigma,\omega \in \density(E)$, we have
    \begin{align}
        \sup_{\rho \in \density(AR)} & \DD\infdiv*{\cN(\rho\ox\sigma)}{\cM(\rho\ox\omega)} =  \sup_{\rho \in \density(AR)} \DD\infdiv*{\cN_{\sigma}(\rho)}{\cM_{\omega}(\rho)}.
    \end{align}
    As a consequence of purification, data processing and the Schmidt decomposition, we can restrict $\rho$ to pure states, and restrict the  system ${R}$ to be isomorphic to the system ${A}$, \ie,
    \begin{align}
         \sup_{\rho \in \density(AR)} \DD\infdiv*{\cN_{\sigma}(\rho)}{\cM_{\omega}(\rho)}
         & = \sup_{\rho_{A} \in \density(A)} \DD\infdiv*{\cN_{\sigma}(\proj{\rho})}{\cM_{\omega}(\proj{\rho})}.
    \end{align}
    This implies that \eqref{eq: minimax 3} equals~\eqref{eq: minimax 4}.
    We apply the minimax theorem in~\cite[Theorem~5.2]{farkas2006potential}, with compact variable $(\sigma,\omega)\in\density(E)\times \density(E)$ and convex variable $\rho_A\in\density(A)$. The set $\density(E)\times \density(E)$ is compact and convex. For each fixed $\rho_A$, lower semicontinuity of $\DD$ makes the objective lower semicontinuous in $(\sigma,\omega)$. Direct-sum equality and data processing imply joint convexity of $\DD$, and hence joint convexity of the objective in $(\sigma,\omega)$. By~\cite[Lemma II.3]{leditzky2018approaches}, they also make the objective concave in $\rho_A$. This implies that \eqref{eq: minimax 2} equals \eqref{eq: minimax 3}. It is also clear that \eqref{eq: minimax 4} is no smaller than \eqref{eq: minimax 1}, which in turn is no smaller than \eqref{eq: minimax 2}.
    For the \Renyi divergences, the same proof goes through via their quasi-divergences.
\end{proof}

\begin{remark}\label{rem: jammer divergence sup of tester-input}
Dually to the expressions in Remark~\ref{rem: jammer divergence sup of jammer-input}, the minimax identities allow us to express the maximin channel divergence in terms of the tester-input channel divergence as follows:
\begin{align}
    \DD^{\uparrow,\downarrow\downarrow}\infdiv*{\cN}{\cM} & =
    \ \inf_{\sigma,\, \omega \in \density(E)} \ \supDiv\infdiv*{\cN_{\sigma}}{\cM_{\omega}},\\
    \DD^{\uparrow,\downarrow}\infdiv*{\cN}{\cM} & =
   \ \ \ \inf_{\sigma \in \density(E)} \ \ \ \supDiv\infdiv*{\cN_{\sigma}}{\cM_{\sigma}}.
\end{align}
\end{remark}

\subsubsection*{Asymptotic equivalence}

As shown in Example~\ref{exam:1}, minimizing over a common input or over separate inputs can yield different values. In contrast, we show that this discrepancy \emph{vanishes} asymptotically, as formalized in Proposition~\ref{prop:asymptotic_equivalence_variants}.
To this end, we first establish a simple yet useful bound relating the divergence between two convex combinations to the divergences of their constituent terms. But this bound differs from the usual joint convexity property of quantum divergences.

\begin{boxlemma}[Divergence estimate for convex mixture]\label{lem: A bound for convex combinations}
Let $\DD$ be a quantum divergence satisfying normalization, operator
dominance, and joint convexity as specified in
Definition~\ref{def:divergence-properties}.
Let $\rho_1,\sigma_1\in\density$ and $\rho_2,\sigma_2\in\PSD$. For every $\ve\in(0,1)$, the following inequality holds
\begin{align}
\DD\!\left((1-\ve)\rho_1 + \ve\sigma_1 \,\big\|\,(1-\ve)\rho_2 + \ve\sigma_2\right)
\leq (1-\ve)\DD(\rho_1\|\sigma_2) + \ve\DD(\sigma_1\|\sigma_2) - \log \ve .
\end{align}
\end{boxlemma}

\begin{proof}
We first establish the following chain of inequalities:
\begin{align}
& \DD\!\left((1-\ve)\rho_1 + \ve\sigma_1 \,\big\|\,(1-\ve)\rho_2 + \ve\sigma_2\right) \notag\\
& \hspace{2cm}\le \DD\!\left((1-\ve)\rho_1 + \ve\sigma_1 \,\big\|\, \ve\sigma_2\right) \\
&\hspace{2cm}= \DD\!\left((1-\ve)\rho_1 + \ve\sigma_1 \,\big\|\, \sigma_2\right)-\log \ve \\
&\hspace{2cm}\le (1-\ve)\DD(\rho_1\|\sigma_2)
    + \ve\DD(\sigma_1\|\sigma_2)-\log \ve .
\end{align}
The first inequality follows from $\ve \sigma_2 \le
(1-\ve)\rho_2+\ve\sigma_2$ and operator dominance. The equality uses
normalization, while the final inequality follows from convexity in the first
argument. 
\end{proof}

The following result shows that allowing the jammer to choose different inputs under the two hypotheses does not change the asymptotic maximin divergence. The key idea is to mix the two hypothesis-dependent inputs into one common input, incurring only a vanishing penalty after regularization.

\begin{boxproposition}[Asymptotic equivalence]
\label{prop:asymptotic_equivalence_variants}
Let $\DD$ be a quantum divergence satisfying normalization, operator
dominance, boundedness, and joint convexity as specified in
Definition~\ref{def:divergence-properties}.
Let $\cN \in \CPTP(AE\!:\!B)$ and $\cM \in \CP(AE\!:\!B)$. Then, for any $\ve\in(0,1)$,
\begin{align}\label{eq: equivalence one and two down arrows}
    \DD^{\uparrow,\downarrow\downarrow}(\cN\|\cM)
    \leq
    \DD^{\uparrow,\downarrow}(\cN\|\cM)
    \leq
    (1-\ve)\,\DD^{\uparrow,\downarrow\downarrow}(\cN\|\cM)
    + \ve\, \DD^{\uparrow}(\cN\|\cM)
    - \log \ve .
\end{align}
In particular, suppose that $D_{\max}^{\uparrow}(\cN\|\cM)<\infty$ and that
$\DD^{\infty,\uparrow,\downarrow\downarrow}(\cN\|\cM)$ exists.  Then
\begin{align}
    \DD^{\infty,\uparrow,\downarrow\downarrow}(\cN\|\cM)
    =
    \DD^{\infty,\uparrow,\downarrow}(\cN\|\cM).
    \label{eq:asymptotic-equivalence-same-separate-jammer}
\end{align}
\end{boxproposition}

\begin{proof}
    The first inequality in Eq.~\eqref{eq: equivalence one and two down arrows} is clear from the definitions. We prove the second one. For any $\ve\in (0,1)$, $\sigma,\omega \in \density(E)$, $\cE \in \CPTP(E\!:\!B)$ and $\cF \in \CP(E\!:\!B)$,
\begin{align}
   \DD^{\downarrow}(\cE\|\cF) & \leq \DD(\cE((1-\ve)\sigma + \ve \omega)\|\cF((1-\ve)\sigma + \ve \omega))\\
    & =  \DD((1-\ve) \cE(\sigma) + \ve \cE(\omega)\| (1-\ve) \cF(\sigma) + \ve \cF(\omega))\\
    & \leq (1-\ve) \DD(\cE(\sigma)\|\cF(\omega)) + \ve \DD(\cE(\omega)\|\cF(\omega)) - \log \ve\\
    & \leq (1-\ve) \DD(\cE(\sigma)\|\cF(\omega)) + \ve \DD^{\uparrow}(\cE\|\cF) - \log \ve, \label{eq: tmp relation}
\end{align}
where the second inequality follows from Lemma~\ref{lem: A bound for convex combinations} and the last inequality follows from the definition of $\DD^{\uparrow}(\cE\|\cF)$.
Taking infimum over $\sigma$ and $\omega$, we have
\begin{align}
    \DD^{\downarrow}(\cE\|\cF) & \leq (1-\ve) \DD^{\downarrow\downarrow}(\cE\|\cF) + \ve \DD^{\uparrow}(\cE\|\cF) - \log \ve.\label{eq: down arrow relation tmp 1}
\end{align}
Then for any $\cN \in \CPTP(AE\!:\!B)$ and $\cM \in \CP(AE\!:\!B)$, taking supremum over induced channels $\cE = \cN_\rho$ and $\cF = \cM_\rho$ for $\rho \in \density(AR)$, we have
\begin{align}
    \DD^{\uparrow,\downarrow}(\cN\|\cM) & = \sup_{\rho \in \density(AR)} \DD^{\downarrow}(\cN_\rho\|\cM_\rho)\\
    & \leq \sup_{\rho \in \density(AR)} (1-\ve) \DD^{\downarrow\downarrow}(\cN_\rho\|\cM_\rho) + \ve \DD^{\uparrow}(\cN_\rho\|\cM_\rho) - \log \ve\\
    & \leq (1-\ve) \DD^{\uparrow,\downarrow\downarrow}(\cN\|\cM) + \ve \DD^{\uparrow}(\cN\|\cM) - \log \ve,
\end{align}
The first line follows from
Remark~\ref{rem: jammer divergence sup of jammer-input}, and the second follows
from Eq.~\eqref{eq: down arrow relation tmp 1}.  The last line uses
$\DD^{\uparrow}(\cN_\rho\|\cM_\rho)
\leq\DD^{\uparrow}(\cN\|\cM)$ for every $\rho$.  This proves
Eq.~\eqref{eq: equivalence one and two down arrows}.

It remains to establish
Eq.~\eqref{eq:asymptotic-equivalence-same-separate-jammer}.  The boundedness
property of $\DD$ and the additivity of the max-relative entropy under tensor
products of channels~\cite[Lemma~12]{wilde2020amortized} give, for every
$n\in\NN$,
\begin{align}
\DD^{\uparrow}(\cN^{\ox n}\|\cM^{\ox n})
\leq D_{\max}^{\uparrow}(\cN^{\ox n}\|\cM^{\ox n})
=nD_{\max}^{\uparrow}(\cN\|\cM).
\end{align}
Consequently,
\begin{align}
\overline{\DD}^{\infty,\uparrow}(\cN\|\cM)
\leq D_{\max}^{\uparrow}(\cN\|\cM)<\infty.
\end{align}

Applying Eq.~\eqref{eq: equivalence one and two down arrows} to
$\cN^{\ox n}$ and $\cM^{\ox n}$ with any fixed $\delta\in(0,1)$, and then
dividing by $n$, yields
\begin{align}
\frac{1}{n}\DD^{\uparrow,\downarrow\downarrow}
(\cN^{\ox n}\|\cM^{\ox n})
&\leq
\frac{1}{n}\DD^{\uparrow,\downarrow}
(\cN^{\ox n}\|\cM^{\ox n}) \\
&\leq
\frac{1-\delta}{n}\DD^{\uparrow,\downarrow\downarrow}
(\cN^{\ox n}\|\cM^{\ox n})
+\frac{\delta}{n}\DD^{\uparrow}
(\cN^{\ox n}\|\cM^{\ox n})
-\frac{\log\delta}{n}.
\end{align}
Taking the limit inferior in the first inequality and the limit superior in
the second gives
\begin{align}
\DD^{\infty,\uparrow,\downarrow\downarrow}(\cN\|\cM)
\leq \underline{\DD}^{\infty,\uparrow,\downarrow}(\cN\|\cM)
& \leq \overline{\DD}^{\infty,\uparrow,\downarrow}(\cN\|\cM)\notag\\
& \leq (1-\delta)\DD^{\infty,\uparrow,\downarrow\downarrow}(\cN\|\cM)
+\delta\,\overline{\DD}^{\infty,\uparrow}(\cN\|\cM).
\end{align}
Because the last term is finite, letting $\delta\to0^+$ yields
\begin{align}
\DD^{\infty,\uparrow,\downarrow\downarrow}(\cN\|\cM)
\leq \underline{\DD}^{\infty,\uparrow,\downarrow}(\cN\|\cM)
\leq \overline{\DD}^{\infty,\uparrow,\downarrow}(\cN\|\cM)
\leq \DD^{\infty,\uparrow,\downarrow\downarrow}(\cN\|\cM).
\end{align}
This completes the proof.
\end{proof}

\begin{boxremark}
\label{rem: four divergence conincidence}
Under the assumptions of Propositions~\ref{prop: jammer divergence minimax} and~\ref{prop:asymptotic_equivalence_variants}, the four regularized minimax and maximin channel
divergences coincide:
\begin{align}
    \DD^{\infty,\downarrow,\uparrow}\infdiv*{\cN}{\cM}
    = \DD^{\infty,\uparrow,\downarrow}\infdiv*{\cN}{\cM}
    = \DD^{\infty,\downarrow\downarrow,\uparrow}\infdiv*{\cN}{\cM}
    = \DD^{\infty,\uparrow,\downarrow\downarrow}\infdiv*{\cN}{\cM}.
\end{align}
\end{boxremark}

\subsection{Other technical preparations}
\label{sec: Other technical preparation}

Beyond the above two key ingredients, the proof of the main theorem requires several
additional technical lemmas that extend a few known results to the minimax/maximin setting
considered here.

\subsubsection*{Continuity in the R\'enyi order}

The following lemma proves the continuity of the maximin sandwiched and measured R\'enyi divergences at order one. This extends a few properties in~\cite[Lemma 33]{ding2023bounding},~\cite[Theorem 68]{mosonyi2023some} and~\cite[Lemma 22]{fang2024generalized} to the maximin setting.

\begin{boxlemma}[Continuity at order one]
\label{lem: DM alpha jam continuity}
    For any $\cN \in \CPTP(AE\!:\!B)$ and $\cM \in \CP(AE\!:\!B)$,
    \begin{align}
    \sup_{\alpha \in (1/2, 1)} D_{\Sand,\alpha}^{\uparrow, (\star)}(\cN\|\cM) & = \inf_{\alpha > 1} D_{\Sand,\alpha}^{\uparrow,(\star)}(\cN\|\cM) =   D^{\uparrow,(\star)}(\cN\|\cM), \quad (\star) \in \{(\downarrow), (\downarrow\downarrow)\},\label{eq: sand ddu app}\\
    \sup_{\alpha \in (1/2, 1)} D_{\Meas,\alpha}^{\uparrow, (\star)}(\cN\|\cM) & = \inf_{\alpha > 1} D_{\Meas,\alpha}^{\uparrow,(\star)}(\cN\|\cM) =   D_{\Meas}^{\uparrow,(\star)}(\cN\|\cM),  \quad (\star) \in \{(\downarrow), (\downarrow\downarrow)\}.
    \label{eq: DM alpha cont udd}
    \end{align}
\end{boxlemma}
\begin{proof}
    We prove~\eqref{eq: sand ddu app} for ``$\downarrow\downarrow$'', and the proof for ``$\downarrow$'' is similar. For the infimum case, we have
    \begin{align}
        \inf_{\alpha> 1} D_{\Sand,\alpha}^{\uparrow, \downarrow\downarrow}(\cN\|\cM)
        & = \inf_{\alpha > 1} \ \ \inf_{\sigma,\, \omega \in \density(E)} D_{\Sand,\alpha}^{\uparrow}\infdiv*{\cN_{\sigma}}{\cM_{\omega}}\\
        & =  \inf_{\sigma,\, \omega \in \density(E)}\ \ \inf_{\alpha > 1} \, D_{\Sand,\alpha}^{\uparrow}\infdiv*{\cN_{\sigma}}{\cM_{\omega}}\\
        & = \inf_{\sigma,\, \omega \in \density(E)}\ \ D^{\uparrow}\infdiv*{\cN_{\sigma}}{\cM_{\omega}}\\
        & = D^{\uparrow,\downarrow\downarrow}\infdiv*{\cN}{\cM},
    \end{align}
    where the first and last lines follow from Remark~\ref{rem: jammer divergence sup of tester-input}, the second line exchanges the order of infima and the third line follows from the continuity of tester-input channel divergence in~\cite[Lemma 33]{ding2023bounding} and~\cite[Theorem 68]{mosonyi2023some}. For the supremum case, we have
    \begin{align}
        \sup_{\alpha\in (1/2, 1)} D_{\Sand,\alpha}^{\uparrow, \downarrow\downarrow}(\cN\|\cM) & = \sup_{\alpha \in (1/2, 1)} \ \sup_{\rho \in \density(AR)} D_{\Sand,\alpha}^{\downarrow\downarrow}\infdiv*{\cN_{\rho}}{\cM_{\rho}}\\
        & = \sup_{\rho \in \density(AR)}\ \sup_{\alpha \in (1/2, 1)}  D_{\Sand,\alpha}^{\downarrow\downarrow}\infdiv*{\cN_{\rho}}{\cM_{\rho}}\\
        & = \sup_{\rho \in \density(AR)} \ D^{\downarrow\downarrow}\infdiv*{\cN_{\rho}}{\cM_{\rho}}\\
        & = D^{\uparrow,\downarrow\downarrow}\infdiv*{\cN}{\cM},
    \end{align}
    where the first and last lines follow from Remark~\ref{rem: jammer divergence sup of jammer-input}, the second line exchanges the order of suprema and the third line follows from the same argument of~\cite[Lemma 22]{fang2024generalized} for two channel image sets and sandwiched \Renyi divergence.
    The measured-\Renyi identity follows by the same argument.
\end{proof}

\subsubsection*{Superadditivity and existence of regularized rates}

Additivity properties of state divergences often follow directly from their definitions, and channel extensions typically inherit analogous tensor-product behavior, depending on the direction of optimization. For example, tester-input channel divergences are usually superadditive, whereas jammer-input channel divergences are typically subadditive.
In contrast, additivity for the minimax/maximin divergences is not immediate, since the tester and jammer optimize in opposite directions. Nevertheless, the following lemma shows that superadditivity holds for the maximin divergence induced by the measured \Renyi divergence. This extends~\cite[Lemma 21]{fang2024generalized} to the maximin setting.

\begin{boxlemma}[Superadditivity]
\label{lem:superadditive:mrDiv}
Let $\alpha \in (0,+\infty)$.
For any $\cN_1,\cN_2\in \CPTP({A}{E}\!:\!{B})$ and $\cM_1,\cM_2 \in \CP({A}{E}\!:\!{B})$, the following superadditivity relation holds:
\begin{align}\label{eq:superadditive:mrDiv}
\begin{aligned}
    \mrDiv{\alpha}^{\uparrow,\downarrow\downarrow}\infdiv*{\cN_1\ox \cN_2}{\cM_1\ox \cM_2}
    \geq
    \mrDiv{\alpha}^{\uparrow,\downarrow\downarrow}\infdiv*{\cN_1}{\cM_1} + \mrDiv{\alpha}^{\uparrow,\downarrow\downarrow}\infdiv*{\cN_2}{\cM_2}.
\end{aligned}\end{align}
Consequently, the regularized channel divergence always exists and can be expressed as:
    \begin{align}
    \mrDiv{\alpha}^{\infty,\uparrow,\downarrow\downarrow}\infdiv*{\cN}{\cM} := \lim_{n\to \infty} \frac{1}{n} \mrDiv{\alpha}^{\uparrow,\downarrow\downarrow}\infdiv*{{\cN}^{\ox n}}{{\cM}^{\ox n}}
    = \sup_{n\in\NN} \frac{1}{n} \mrDiv{\alpha}^{\uparrow,\downarrow\downarrow}\infdiv*{{\cN}^{\ox n}}{{\cM}^{\ox n}},
    \end{align}
    where the second equality follows from Fekete's lemma.
\end{boxlemma}

\begin{proof}
Write $\rho_{12}\in\density(A_1A_2R_1R_2)$ and, for $i=1,2$,
$\rho_i\in\density(A_iR_i)$.  The proof follows by
\begin{align}
  \text{LHS of ~\eqref{eq:superadditive:mrDiv}}
    & = \sup_{\rho_{12}}
    \iinfmrDiv{\alpha}\infdiv*{
    (\cN_1\otimes\cN_2)_{\rho_{12}}}
    {(\cM_1\otimes\cM_2)_{\rho_{12}}}
    \label{eq:superadditive:mrDiv:1}\\
    \label{eq:superadditive:mrDiv:2}
    & \geq \sup_{\rho_1,\rho_2}
    \iinfmrDiv{\alpha}\infdiv*{
    (\cN_1\otimes\cN_2)_{\rho_1\ox\rho_2}}
    {(\cM_1\otimes\cM_2)_{\rho_1\ox\rho_2}}
    \\
    \label{eq:superadditive:mrDiv:3}
    & = \sup_{\rho_1,\rho_2}
    \iinfmrDiv{\alpha}\infdiv*{
    (\cN_1)_{\rho_1}\otimes(\cN_2)_{\rho_2}}
    {(\cM_1)_{\rho_1}\otimes(\cM_2)_{\rho_2}}
    \\
    \label{eq:superadditive:mrDiv:4}
    & \geq \sup_{\rho_1,\rho_2}
    \left[
    \iinfmrDiv{\alpha}\infdiv*{(\cN_1)_{\rho_1}}{(\cM_1)_{\rho_1}}
    + \iinfmrDiv{\alpha}\infdiv*{(\cN_2)_{\rho_2}}{(\cM_2)_{\rho_2}}
    \right]
    \\
    \label{eq:superadditive:mrDiv:5}
    & = \text{RHS of~\eqref{eq:superadditive:mrDiv}},
\end{align}
where Eqs.~\eqref{eq:superadditive:mrDiv:1} and~\eqref{eq:superadditive:mrDiv:5} follow from Remark~\ref{rem: jammer divergence sup of jammer-input} and Proposition~\ref{prop: jammer divergence minimax}; Eq.~\eqref{eq:superadditive:mrDiv:2} is a simple restriction on the feasible set; Eq.~\eqref{eq:superadditive:mrDiv:3} follows directly from the definition of the induced channels; and Eq.~\eqref{eq:superadditive:mrDiv:4} follows by the superadditivity of the jammer-input measured \Renyi channel divergence (\ie, by applying twice the chain rule for the measured \Renyi channel divergence~\cite[Lemma~2]{fang2025adversarial}).
\end{proof}

\begin{remark}
It is worth noting that the superadditivity statement in Lemma~\ref{lem:superadditive:mrDiv} does not extend to the variant $\mrDiv{\alpha}^{\uparrow,\downarrow}$ that minimizes over a common input.
Indeed, superadditivity can already fail for the underlying jammer-input channel divergence; see the counterexample in~\cite[Supplemental Material, Section~C]{fang2025adversarial}.
This highlights the subtle distinction between the ``$\downarrow$'' and ``$\downarrow\downarrow$'' formulations.
\end{remark}

\subsubsection*{Permutation-symmetry reduction}

We now establish a symmetry property of the maximin channel divergence. Specifically, the optimizations defining the 
$n$-fold maximin channel divergences can be restricted to permutation invariant states without changing their values. This result extends~\cite[Proposition II.4]{leditzky2018approaches} to the maximin setting.

\begin{boxlemma}[Symmetry reduction]\label{lem:mimmaxDiv:perm}
Let $\DD$ be a quantum divergence satisfying the minimax property in Proposition~\ref{prop: jammer divergence minimax}. 
For any $\cN \in \CPTP({AE}\!:\!{B})$, $\cM \in \CP({AE}\!:\!{B})$, and integer $n\geq 2$, 
\begin{align}
\DD^{\uparrow,\downarrow\downarrow} \infdiv*{\cN^{\ox n}}{\cM^{\ox n}} 
= \adjustlimits \inf_{\substack{\sigma,\,\omega \in \PER(E^n)}} \sup_{\rho \in \PER(A^n)}\ \DD\infdiv*{\cN^{\ox n}(\proj{\rho}\ox\sigma)}{\cM^{\ox n}(\proj{\rho}\ox\omega)}, \label{eq: minimax symmetry reduction 3}
\end{align}
where $\PER(A^n)$ and $\PER(E^n)$ denote the set of permutation-invariant states. Here $|\rho\rangle$ is a purification of $\rho$ on $R^nA^n$, with $R\simeq A$. Similarly, the result also holds for $\DD^{\uparrow,\downarrow}$ by restricting to $\omega = \sigma$.
\end{boxlemma}

\begin{proof}
  We aim to show that
  \begin{align}
\DD^{\uparrow,\downarrow\downarrow} \infdiv*{\cN^{\ox n}}{\cM^{\ox n}} 
&= \adjustlimits \inf_{\substack{\sigma,\,\omega \in \PER(E^n)}} \sup_{\rho \in \density(A^n)}\ f_n(\sigma, \omega|\rho)= \adjustlimits \inf_{\substack{\sigma,\,\omega \in \PER(E^n)}} \sup_{\rho \in \PER(A^n)}\ f_n(\sigma, \omega|\rho), \label{eq: minimax symmetry reduction 2}
\end{align}
where the objective function is defined by
\begin{align}
    f_n(\sigma, \omega|\rho) := \DD\infdiv*{\cN^{\ox n}(\proj{\rho}\ox\sigma)}{\cM^{\ox n}(\proj{\rho}\ox\omega)}.
\end{align}
    The key is to prove the following expression, as a function of $\sigma$ and $\omega$,
    \begin{align}\label{eq: inner sup}
    f_n(\sigma, \omega) := \sup_{\rho \in \density(A^n)} f_n(\sigma, \omega|\rho)
    \end{align}
    to be permutation invariant.
    For each permutation $\pi$ in the permutation group $\set{S}_n$, we define the unitary transformation $\pi_{A^n}\in\LinOp(A^n)$ as 
    \begin{align}
    \pi_{A^n}: \ket{\psi_1}_{{A}_1}\ket{\psi_2}_{{A}_2}\cdots\ket{\psi_n}_{{A}_n}
    \mapsto \ket{\psi_{\pi(1)}}_{{A}_1}\ket{\psi_{\pi(2)}}_{{A}_2}\cdots\ket{\psi_{\pi(n)}}_{{A}_n}.
    \end{align}
    We also abuse the notation a bit, and use $\pi_{A^n}$ to denote the CPTP $\pi_{A^n}: \rho \mapsto \pi_{A^n} \rho \pi_{A^n}^\herm$ as well.
    By carrying over the permutation of the output systems of the tensor of quantum channel to its input system (\cf~\cite[Eq.~(2)]{boche2018fully}~and~\cite[Theorem~3.3]{belzig2024fully}), we have 
    \begin{align}
        \pi_{B^n}\left(\cN^{\ox n}(\proj{\rho}\ox\sigma)\right)
        &= \cN^{\ox n}\left(\pi_{A^nE^n}(\proj{\rho}\ox\sigma)\right) \\
        &= \cN^{\ox n}\left(\pi_{A^n}(\proj{\rho})\ox\pi_{E^n}(\sigma)\right)\\
        &= \cN^{\ox n}\left(\proj{\pi_{A^n}\rho}_{A^nR^n}\ox\pi_{E^n}(\sigma)\right).
    \end{align}
    Similarly, the above equalities also hold for the channel $\cM$.
    Since, quantum divergences are invariant under unitary transformations, we have 
    \begin{align}
         f_n(\sigma, \omega|\rho) 
        &= \DD\infdiv*{\pi_{B^n}\left(\cN^{\ox n}(\proj{\rho}\ox\sigma)\right)}{\pi_{B^n}\left(\cM^{\ox n}(\proj{\rho}\ox\omega)\right)} \\
        &=\DD\infdiv*{\cN^{\ox n}\left(\proj{\pi_{A^n}\rho}\!\ox\!\pi_{E^n}(\sigma)\right)}{\cM^{\ox n}\left(\proj{\pi_{A^n}\rho}\!\ox\!\pi_{E^n}(\omega)\right)}\\
        &= f_n(\pi_{E^n}(\sigma), \pi_{E^n}(\omega)|\pi_{A^n}(\rho)) .
    \end{align}
    Note that the set of all density operators on systems $A^n$ is permutation invariant.
    Hence,
    \begin{align}
        f_n(\pi_{E^n}(\sigma), \pi_{E^n}(\omega))
        &= \sup_{\rho \in \density(A^n)} f_n(\pi_{E^n}(\sigma), \pi_{E^n}(\omega)|\rho)\\
        &= \sup_{\rho \in \density(A^n)} f_n(\pi_{E^n}(\sigma), \pi_{E^n}(\omega)|\pi_{A^n}(\rho)) \\
        &= \sup_{\rho \in \density(A^n)} f_n(\sigma, \omega|\rho),
    \end{align}
    \ie, $f_n$ in Eq.~\eqref{eq: inner sup} is permutation invariant.

  By the unitary invariance and the joint convexity of $\DD$, $f_n(\sigma,\omega|\rho)$ is also unitarily invariant and joint convex on $(\sigma,\omega)$. Therefore, for any optimizing pair $(\sigma^{*}, \omega^{*})$, we can replace them by the average over all of its permutations, \ie,
    \begin{align}
    (\sigma^\new, \omega^\new )
    & := \left(\sum_{\pi\in\set{S}_n} \size{\set{S}_n}^{-1}\cdot\pi_{E^n}(\sigma^{*}), \sum_{\pi\in\set{S}_n} \size{\set{S}_n}^{-1}\cdot\pi_{E^n}(\omega^{*})\right)
    \end{align}
    which is a pair of permutation invariant states; hence, finishing the proof of the first equality in Eq.~\eqref{eq: minimax symmetry reduction 2}. Since the induced channels $\cN^{\ox n}(\cdot\ox\sigma)$ and $\cM^{\ox n}(\cdot\ox\omega)$ are permutation invariant for $\sigma,\,\omega \in \PER(E^n)$, we have the second equality in Eq.~\eqref{eq: minimax symmetry reduction 2} by applying the known permutation reduction for the tester-input divergence in~\cite[Proposition II.4]{leditzky2018approaches}.
\end{proof}

\subsubsection*{Measured-to-quantum asymptotic equivalence}

We next establish an asymptotic equivalence result: after regularization, the maximin channel divergence induced by the measured divergence coincides with that induced by its fully quantum counterpart. This extends~\cite[Lemma 28]{fang2024generalized} to the maximin setting.

\begin{boxlemma}\label{lem: DM D jam infty}
For any $\alpha \in [1/2,\infty)$, $\cN\in\CPTP({AE}\!:\!{B})$ and $\cM\in\CP({AE}\!:\!{B})$,
    \begin{align}\label{eq: DM D jam infty}
        D_{\Meas,\alpha}^{\infty,\uparrow,\downarrow\downarrow}\infdiv*{\cN}{\cM} = D_{\Sand,\alpha}^{\infty,\uparrow,\downarrow\downarrow}\infdiv*{\cN}{\cM}.
    \end{align}
Moreover, the following holds
    \begin{align}
        \sup_{\alpha \in (1/2,1)} \mrDiv{\alpha}^{\infty,\uparrow,\downarrow\downarrow}\infdiv*{\cN}{\cM} = \mDiv^{\infty,\uparrow,\downarrow\downarrow}\infdiv*{\cN}{\cM} = \uDiv^{\infty,\uparrow,\downarrow\downarrow}\infdiv*{\cN}{\cM}.\label{eq: DM alpha cont udd reg}
    \end{align}
\end{boxlemma}

\begin{proof}
For every $n$, we first establish
\begin{align}
\frac{1}{n}
D_{\Meas,\alpha}^{\uparrow,\downarrow\downarrow}
\infdiv*{\cN^{\ox n}}{\cM^{\ox n}}
&\le\frac{1}{n}
D_{\Sand,\alpha}^{\uparrow,\downarrow\downarrow}
\infdiv*{\cN^{\ox n}}{\cM^{\ox n}} \le\frac{1}{n}
D_{\Meas,\alpha}^{\uparrow,\downarrow\downarrow}
\infdiv*{\cN^{\ox n}}{\cM^{\ox n}}+\frac{1}{n}f(n),
\end{align}
where $f(n) = 2\log[(n+1)^d(n+d)^{d^2}]$ and $d = \dim(\cH_{BR})$.
The first inequality follows from the ordering between the measured and sandwiched
R\'enyi divergences. For the second inequality, Lemma~\ref{lem:mimmaxDiv:perm}
restricts both optimizations to permutation-invariant variables, after which
Lemma~\ref{Lemma: DM and Sandwiched relation} provides the correction $f(n)$.
Since $f(n)/n\to0$, taking the limit $n\to\infty$ yields
Eq.~\eqref{eq: DM D jam infty}; the existence of the limit follows from
Lemma~\ref{lem:superadditive:mrDiv}.

Next, we prove Eq.~\eqref{eq: DM alpha cont udd reg} by first writing
\begin{align}
    \sup_{\alpha \in (1/2,1)} \mrDiv{\alpha}^{\infty,\uparrow,\downarrow\downarrow}\infdiv*{\cN}{\cM}
    &= \sup_{\alpha \in (1/2,1)} \sup_{n\in\NN} \frac{1}{n} \mrDiv{\alpha}^{\uparrow,\downarrow\downarrow}\infdiv*{\cN^{\ox n}}{\cM^{\ox n}} \\
    &= \sup_{n\in\NN} \ \sup_{\alpha \in (1/2,1)} \frac{1}{n} \mrDiv{\alpha}^{\uparrow,\downarrow\downarrow}\infdiv*{\cN^{\ox n}}{\cM^{\ox n}} \\
    &= \sup_{n\in\NN} \ \frac{1}{n} \mDiv^{\uparrow,\downarrow\downarrow}\infdiv*{\cN^{\ox n}}{\cM^{\ox n}} \\
    &= \mDiv^{\infty,\uparrow,\downarrow\downarrow}\infdiv*{\cN}{\cM}.
\end{align}
The first equality follows from Lemma~\ref{lem:superadditive:mrDiv}. The second
equality exchanges the two suprema. The third equality uses the continuity
established in Lemma~\ref{lem: DM alpha jam continuity}, and the final equality
is the definition of the regularized measured divergence.
This completes the proof.
\end{proof}

\subsection{Proof of Theorem~\ref{thm:general-general-Stein}}
\label{sec:proof-of-general-general-Stein}

We are now ready to prove Theorem~\ref{thm:general-general-Stein}.

\begin{proof}[Proof of Theorem~\ref{thm:general-general-Stein}]
For every discrimination setting 
\begin{align}
(\star)\in
\bigl\{
(\downarrow,\uparrow),
(\downarrow\downarrow,\uparrow),
(\uparrow,\downarrow),
(\uparrow,\downarrow\downarrow)
\bigr\},
\end{align}
Theorem~\ref{thm:unified-operational-correspondence} and the minimax inequality give
\begin{align}
\hDiv{\ve}^{\uparrow,\downarrow\downarrow}
\infdiv*{\cN^{\ox n}}{\cM^{\ox n}}
&\le
\Stein_{\ve,n}^{(\star)}(\cN,\cM)
\le
\hDiv{\ve}^{\downarrow,\uparrow}
\infdiv*{\cN^{\ox n}}{\cM^{\ox n}}.
\label{eq:AEP-sandwich-all-variants-revised}
\end{align}
It therefore suffices to prove the achievability bound for
$\hDiv{\ve}^{\uparrow,\downarrow\downarrow}$ and the converse bound for
$\hDiv{\ve}^{\downarrow,\uparrow}$. The result for every operational setting
then follows by sandwiching.

\paragraph{Achievability bound.}
For every $\alpha\in(1/2,1)$, the one-shot bound in
Lemma~\ref{lem:one-shot-hypothesis-testing-bounds} and
$D_{\Meas,\alpha}\leq D_{\Sand,\alpha}$ give
\begin{align}
\hDiv{\ve}^{\uparrow,\downarrow\downarrow}
\infdiv*{\cN^{\ox n}}{\cM^{\ox n}}
\ge
D_{\Meas,\alpha}^{\uparrow,\downarrow\downarrow}
\infdiv*{\cN^{\ox n}}{\cM^{\ox n}}
+
\frac{\alpha}{\alpha-1}\log\frac1\ve.
\end{align}
Dividing by $n$ and taking the limit inferior yields
\begin{align}
\liminf_{n\to\infty}
\frac1n
\hDiv{\ve}^{\uparrow,\downarrow\downarrow}
\infdiv*{\cN^{\ox n}}{\cM^{\ox n}}
\ge
D_{\Meas,\alpha}^{\infty,\uparrow,\downarrow\downarrow}
(\cN\|\cM).
\end{align}
Taking the supremum over $\alpha\in(1/2,1)$ and using
Eq.~\eqref{eq: DM alpha cont udd reg} give
\begin{align}
\liminf_{n\to\infty}
\frac1n
\hDiv{\ve}^{\uparrow,\downarrow\downarrow}
\infdiv*{\cN^{\ox n}}{\cM^{\ox n}}
\ge
D^{\infty,\uparrow,\downarrow\downarrow}\infdiv*{\cN}{\cM}.
\label{eq:AEP-lower-revised}
\end{align}

\paragraph{Converse bound.}
Let $\mu$ be as in
Lemma~\ref{lem:uniform-block-bounds}.  The trace estimate in Eq.~\eqref{eq: trace estimation} and
Lemma~\ref{lem:one-shot-hypothesis-testing-bounds} give
\begin{align}
\hDiv{\ve}^{\downarrow,\uparrow}
\infdiv*{\cN^{\ox n}}{\cM^{\ox n}}
\le
\frac{
1}
{1-\ve}\Big[D^{\downarrow,\uparrow}
\infdiv*{\cN^{\ox n}}{\cM^{\ox n}}
+h(\ve)+\ve n\log\mu\Big].
\end{align}
Lemma~\ref{lem: DM D jam infty} at $\alpha=1$ gives the existence of
$D^{\infty,\uparrow,\downarrow\downarrow}$.  Thus
Proposition~\ref{prop:asymptotic_equivalence_variants} applies.  After division
by $n$ and taking the limit superior, 
\begin{align}
\limsup_{n\to\infty}
\frac1n
\hDiv{\ve}^{\downarrow,\uparrow}
\infdiv*{\cN^{\ox n}}{\cM^{\ox n}}
\le
\frac{1}
{1-\ve}\Big[
D^{\infty,\uparrow,\downarrow\downarrow}\infdiv*{\cN}{\cM}+\ve\log\mu\Big].
\label{eq:AEP-upper-revised}
\end{align}
Letting $\ve\to0^+$ gives the desired converse bound. 
\end{proof}

\section{IID tester against entangled jammer}
\label{sec:iid-tester-unrestricted-jammer}

This section considers an IID tester against an entangled jammer. As in the
previous section, all four information patterns have the same asymptotic Stein
exponent. However, because the tester's IID input set is nonconvex, the minimax
argument used in the previous section does not apply. The proof of the main
result therefore requires a different approach.

\begin{boxtheorem}[IID tester against entangled jammer]
\label{thm:iid-tester-general-jammer-Stein}
Let $\cN\in\CPTP(AE\!:\!B)$ and $\cM\in\CP(AE\!:\!B)$ satisfy
$D_{\max}^{\uparrow}(\cN\|\cM)<\infty$.  For each of the four information patterns
\begin{align}
(\star)\in
\bigl\{
(\downarrow,\emptyuparrow),
(\downarrow\downarrow,\emptyuparrow),
(\emptyuparrow,\downarrow),
(\emptyuparrow,\downarrow\downarrow)
\bigr\},
\end{align}
the asymptotic Stein exponent is given by
\begin{align}
&\lim_{\ve\to0^+}\liminf_{n\to\infty}
\frac1n
\Stein_{\ve,n}^{(\star)}(\cN,\cM)
=
\lim_{\ve\to0^+}\limsup_{n\to\infty}
\frac1n
\Stein_{\ve,n}^{(\star)}(\cN,\cM)
=
D^{\infty,\emptyuparrow,\downarrow\downarrow}
(\cN\|\cM).
\label{eq:public-iid-tester-unrestricted-jammer-AEP}
\end{align}
\end{boxtheorem}

The operational correspondence in
Theorem~\ref{thm:unified-operational-correspondence}, together with the minimax
inequality, bounds every operational exponent between two hypothesis-testing quantities:
\begin{align}
\hDiv{\ve}^{\emptyuparrow,\downarrow\downarrow}
\infdiv*{\cN^{\ox n}}{\cM^{\ox n}}
&\le
\Stein_{\ve,n}^{(\star)}(\cN,\cM)
\le
\hDiv{\ve}^{\downarrow,\emptyuparrow}
\infdiv*{\cN^{\ox n}}{\cM^{\ox n}}.
\label{eq:AEP-sandwich-all-variants-revised 1}
\end{align}
It therefore suffices to establish achievability for the left endpoint and a
converse bound for the right endpoint; the remaining cases then follow from the
same sandwich. Achievability is comparatively direct because the left endpoint
has the same optimization order and feasible state sets as the target exponent. The converse is more challenging because both the optimization order and the
feasible state sets differ. Consequently, repeating the argument from
Subsection~\ref{sec:proof-of-general-general-Stein} would yield a converse in
terms of $D^{\infty,\downarrow,\emptyuparrow}$ rather than the target rate
$D^{\infty,\emptyuparrow,\downarrow\downarrow}$. Because the IID input family
is nonconvex, the equality of these two rates cannot be established by the
minimax argument used in the previous section.

We instead prove the converse directly at the operational level. For each tester
input, we select an optimal jammer input pair and average these pairs over a
finite subcover of the tester's input space. The averaging produces a jammer input
that is independent of the tester's choice. A continuity argument shows that
the resulting approximation incurs no asymptotic loss. This construction plays
the role of a minimax argument without requiring a minimax identity for the
Umegaki divergence.

The proof uses the following two ingredients. The first supports the finite-subcover
argument; the second quantifies the
information loss caused by averaging the jammer inputs.

\begin{boxlemma}
\label{lem:fixed-block-continuity-general-jammer}
Let $\cN\in\CPTP(AE\!:\!B)$ and $\cM\in\CP(AE\!:\!B)$ satisfy
$D_{\max}^{\uparrow}(\cN\|\cM)<\infty$.  Then, for every fixed $m\in\NN$,
the function
$\rho\mapsto
D^{\downarrow\downarrow}(\cN_\rho^{\ox m}\|\cM_\rho^{\ox m})$
is finite and continuous on $\density(AR)$.
\end{boxlemma}

\begin{proof}
Choose $\kappa$ as in Lemma~\ref{lem:uniform-block-bounds}.  The uniform
divergence bounds in that lemma prove finiteness, so it remains to establish
continuity.
For $t\in(0,1)$, define
\begin{align}
V(\rho)
&:=
\inf_{\sigma,\omega\in\density(E^m)}
D\!\left(
\cN_\rho^{\ox m}(\sigma)
\,\|\,
\cM_\rho^{\ox m}(\omega)\right),
\\
V_t(\rho)
&:=
\inf_{\sigma,\omega\in\density(E^m)}
D\!\left(
\cN_\rho^{\ox m}(\sigma)
\,\|\,
\cM_\rho^{\ox m}(\omega_t)
\right),
\label{eq:fixed-block-mixed-value}
\end{align}
where $\omega_t:= (1-t)\omega+t\sigma$.
A similar argument in
Lemma~\ref{lem:uniform-block-bounds} and linearity give
\begin{align}
\cN_\rho^{\ox m}(\sigma)
&\le
\kappa^m\cM_\rho^{\ox m}(\sigma)
\le
\frac{\kappa^m}{t}
\cM_\rho^{\ox m}(\omega_t)
\label{eq:fixed-block-mixed-jammer-domination}
\end{align}
for every $\rho$, $\sigma$, and $\omega$.  For fixed $t$, both outputs depend
continuously on $(\rho,\sigma,\omega)$; the first is a state, and the second
is nonzero by the preceding domination.  Proposition~37(ii)
of~\cite{mosonyi2023some}, applied at $\alpha=1$ with
the function $f(s)=(\kappa^m/t)s$, shows that the objective defining
$V_t$ is jointly continuous.  Lemma~3(iii) therein and compactness of the
jammer domain then imply that $V_t$ is continuous in $\rho$.

Note that the mixed alternative input $\omega_t$ is feasible in the optimization $V(\rho)$, so
$V(\rho)\le V_t(\rho)$. 
Moreover,
$\cM_\rho^{\ox m}(\omega_t)
\ge(1-t)\cM_\rho^{\ox m}(\omega)$. This gives
\begin{align}
D\!\left(
\cN_\rho^{\ox m}(\sigma)
\,\|\,
\cM_\rho^{\ox m}(\omega_t)
\right)
&\le
D\!\left(
\cN_\rho^{\ox m}(\sigma)
\,\|\,
\cM_\rho^{\ox m}(\omega)
\right)
-\log(1-t).
\label{eq:fixed-block-mixed-objective-upper-bound}
\end{align}
 Taking the infimum in
Eq.~\eqref{eq:fixed-block-mixed-objective-upper-bound} gives the reverse
comparison up to $-\log(1-t)$.  Hence
\begin{align}
0
\le
V_t(\rho)-V(\rho)
\le
-\log(1-t)
\label{eq:fixed-block-uniform-continuity-approximation}
\end{align}
for every $\rho\in\density(AR)$.  Thus $V_t\to V$ uniformly as $t\to0^+$.
Since each $V_t$ is continuous, so is $V$.
\end{proof}

\begin{boxlemma}
\label{lem:dominated-component-hypothesis-testing-bound}
Let $M\ge 1$ and $N>0$.  Suppose that $\rho,\rho'\in\density$ and
$\sigma,\sigma'\in\PSD$ satisfy $M \rho\ge \rho'$ and $N \sigma\ge \sigma'$.
Then, for every $\ve\in[0,1/M)$,
\begin{align}
D_{\Hypo,\ve}(\rho\|\sigma)
&\le
D_{\Hypo,M\ve}(\rho'\|\sigma')+\log N.
\label{eq:dominated-component-hypothesis-testing-bound}
\end{align}
\end{boxlemma}

\begin{proof}
Let $T$ be feasible for the optimal type-II error $\beta_\ve(\rho\|\sigma)$.  The two domination
assumptions give
\begin{align}
\tr[\rho'(I-T)]
&\le
M\tr[\rho(I-T)]
\le
M\ve,
\notag\\
N\tr[\sigma T]
&\ge
\tr[\sigma'T]
\ge
\beta_{M\ve}(\rho'\|\sigma').
\end{align}
The first line shows that $T$ is feasible for
$\beta_{M\ve}(\rho'\|\sigma')$.  Since the second line holds for every test
feasible for $\beta_\ve(\rho\|\sigma)$, this gives
$N\beta_\ve(\rho\|\sigma)
\ge
\beta_{M\ve}(\rho'\|\sigma')$.
Taking $-\log$ proves
the asserted result.
\end{proof}

\bigskip
We are now ready to prove Theorem~\ref{thm:iid-tester-general-jammer-Stein}.

\begin{proof}[Proof of Theorem~\ref{thm:iid-tester-general-jammer-Stein}]

The proof proceeds in three parts. We first establish the existence of the
asymptotic Stein exponents, then prove achievability, and finally establish the
converse.

\prooftag{Limit existence.}
Fix $\rho\in\density(AR)$. The subadditivity of jammer-input divergences gives
\begin{align}
D^{\downarrow\downarrow}
(\cN_\rho^{\ox(n+m)}\|\cM_\rho^{\ox(n+m)})
&\le
D^{\downarrow\downarrow}
(\cN_\rho^{\ox n}\|\cM_\rho^{\ox n})
+D^{\downarrow\downarrow}
(\cN_\rho^{\ox m}\|\cM_\rho^{\ox m}).
\label{eq:fixed-tester-subadditivity-general-jammer}
\end{align}
Taking the supremum over $\rho$ and using
$\sup_\rho(f(\rho)+g(\rho))\le\sup_\rho f(\rho)+\sup_\rho g(\rho)$ proves
subadditivity of $D^{\emptyuparrow,\downarrow\downarrow}$. This indicates the existence of the regularized limit
\begin{align}
D^{\infty,\emptyuparrow,\downarrow\downarrow}(\cN\|\cM)
=
\lim_{n\to\infty}\frac1n
D^{\emptyuparrow,\downarrow\downarrow}
(\cN^{\ox n}\|\cM^{\ox n})
=
\inf_{n\ge1}\frac1n
D^{\emptyuparrow,\downarrow\downarrow}
(\cN^{\ox n}\|\cM^{\ox n}).\label{eq: emptyup downdown limit}
\end{align}

\prooftag{Achievability bound.}
For every $\ve\in(0,1)$ and $\alpha\in(1/2,1)$,
Lemma~\ref{lem:one-shot-hypothesis-testing-bounds} and
$D_{\Meas,\alpha}\leq D_{\Sand,\alpha}$ give
\begin{align}
&D_{\Hypo,\ve}^{\emptyuparrow,\downarrow\downarrow}
 (\cN^{\ox n}\|\cM^{\ox n})\ge
 \sup_{\rho\in\density(AR)}
 D_{\Meas,\alpha}^{\downarrow\downarrow}
 (\cN_\rho^{\ox n}\|\cM_\rho^{\ox n})
 +\frac{\alpha}{\alpha-1}\log\frac1\ve.
 \label{eq:detailed-full-jammer-one-shot-achievability}
\end{align}
After division by $n$, the last term vanishes:
\begin{align}
  \liminf_{n\to\infty}\frac1n
  D_{\Hypo,\ve}^{\emptyuparrow,\downarrow\downarrow}
  (\cN^{\ox n}\|\cM^{\ox n})
  & \ge
  \liminf_{n\to\infty}\frac1n
  \sup_{\rho\in\density(AR)}
  D_{\Meas,\alpha}^{\downarrow\downarrow}
  (\cN_\rho^{\ox n}\|\cM_\rho^{\ox n})
  \label{eq:detailed-full-jammer-one-shot-achievability-vanish}\\
  & \geq \sup_{\rho\in\density(AR)}\ \  \liminf_{n\to\infty} \ \ \frac1n
  D_{\Meas,\alpha}^{\downarrow\downarrow}
  (\cN_\rho^{\ox n}\|\cM_\rho^{\ox n})\\
  & \geq \sup_{\rho\in\density(AR)}\ \ \frac1m
  D_{\Meas,\alpha}^{\downarrow\downarrow}
  (\cN_\rho^{\ox m}\|\cM_\rho^{\ox m}),
\end{align}
where the last inequality above holds for any positive integer $m \in \NN$ and follows from the superadditivity of the jammer-input measured channel divergence~\cite[Lemma~2]{fang2025adversarial}.

Taking the supremum over
$\alpha\in(1/2,1)$, commuting the two suprema, and applying the order-one
continuity for compact convex sets in
\cite[Lemma~22]{fang2024generalized} to the image sets yields
\begin{align}
&\liminf_{n\to\infty}\frac1n
D_{\Hypo,\ve}^{\emptyuparrow,\downarrow\downarrow}
 (\cN^{\ox n}\|\cM^{\ox n})
\ge
\frac1m\sup_{\rho\in\density(AR)}
D_{\Meas}^{\downarrow\downarrow}
 (\cN_\rho^{\ox m}\|\cM_\rho^{\ox m}).
 \label{eq:detailed-full-jammer-measured-block}
\end{align}

Since $\cN_\rho^{\ox m}$ and $\cM_\rho^{\ox m}$ are permutation invariant, the optimal solution in $D_{\Meas}^{\downarrow\downarrow}
 (\cN_\rho^{\ox m}\|\cM_\rho^{\ox m})$ can be chosen as permutation invariant states as well. Then Lemma~\ref{Lemma: DM and Sandwiched relation} gives
\begin{align}
&\liminf_{n\to\infty}\frac1n
D_{\Hypo,\ve}^{\emptyuparrow,\downarrow\downarrow}
 (\cN^{\ox n}\|\cM^{\ox n})
\ge
\frac1m
D^{\emptyuparrow,\downarrow\downarrow}
 (\cN^{\ox m}\|\cM^{\ox m})
-\frac{f(m)}m.
 \label{eq:detailed-full-jammer-Umegaki-block}
\end{align}
where $f(m) = 2\log[(m+1)^{d}(m+d)^{d^2}]$ and $d = \dim(\cH_{BR})$.
Finally, letting $m\to\infty$ we obtain
\begin{align}
\liminf_{n\to\infty}\frac1n
D_{\Hypo,\ve}^{\emptyuparrow,\downarrow\downarrow}
 (\cN^{\ox n}\|\cM^{\ox n})
\ge
D^{\infty,\emptyuparrow,\downarrow\downarrow}(\cN\|\cM).
 \label{eq:detailed-full-jammer-achievability}
\end{align}

\prooftag{Converse bound.}
As mentioned before, the challenge for the converse part is to bound the public hypothesis-unaware setting $(\downarrow,\emptyuparrow)$ by a target divergence with $(\emptyuparrow,\downarrow\downarrow)$. For this, we need to construct a jammer state that universally works well for all tester inputs. 

We first justify fixing $R\simeq A$, independently of the blocklength,
before invoking compactness.  Let $S$ be any finite-dimensional reference
and $\rho_{AS}$ any tester state.  Choose a pure state $\psi_{AR}$
purifying $\rho_A$.  There exists a channel $\Lambda\in\CPTP(R\!:\!S)$
such that, for every $n\in\NN$, $\rho_{AS}=(\id_A\ox\Lambda)(\psi_{AR})$. This follows by purifying $\rho_{AS}$ further, relating
purifications of $\rho_A$ by an isometry, and discarding the extra
purifying system. Then we have
\begin{align}
\cG_\rho^{\ox n}(\eta_n)
&=(\id_{B^n}\ox\Lambda^{\ox n})
\bigl(\cG_\psi^{\ox n}(\eta_n)\bigr),
\qquad \cG\in\{\cN,\cM\},\quad \eta_n\in\density(E^n).
\label{eq:iid-reference-output-recovery}
\end{align}
This holds because the reference channel commutes
with both channel hypotheses.
Consequently, data processing gives, for every jammer pair
$\sigma_n,\omega_n\in\density(E^n)$,
\begin{align}
&D_{\Hypo,\ve}\!\left(
\cN_\rho^{\ox n}(\sigma_n)
\,\big\|\,
\cM_\rho^{\ox n}(\omega_n)\right)
\le
D_{\Hypo,\ve}\!\left(
\cN_\psi^{\ox n}(\sigma_n)
\,\big\|\,
\cM_\psi^{\ox n}(\omega_n)\right).
\label{eq:iid-reference-testing-domination}
\end{align}
The same inequality holds for $D$.  Since $\psi_{AR}^{\ox n}$ is an
admissible IID input and $\Lambda$ is independent of the jammer pair,
both tester optimizations can therefore be restricted to this fixed $R$.
The reverse inequalities follow because states on $AR$ are included among
the unrestricted tester inputs.
This leaves the converse endpoint and the target rate unchanged.
We may thus work on the fixed compact set $\density(AR)$.

We start from the target divergence of $(\emptyuparrow,\downarrow\downarrow)$. For any tester input state, find the corresponding optimal jammer pair and then perturb this pair a bit so that it works for a neighborhood of the original tester input. Choose $\kappa$ and $\mu$ as in
Lemma~\ref{lem:uniform-block-bounds}.  Fix $\gamma>0$. By Eq.~\eqref{eq: emptyup downdown limit}, choose $m$ such
that
\begin{align}
\frac1m
D^{\emptyuparrow,\downarrow\downarrow}
(\cN^{\ox m}\|\cM^{\ox m})
\le
D^{\infty,\emptyuparrow,\downarrow\downarrow}(\cN\|\cM)
+\gamma.
\end{align}
For each $\rho_0\in\density(AR)$, choose a minimizing pair
$(\sigma_{\rho_0},\omega_{\rho_0})$ such that
\begin{align}
  D^{\downarrow\downarrow}(\cN_{\rho_0}^{\ox m}\|\cM_{\rho_0}^{\ox m}) = D(\cN_{\rho_0}^{\ox m}(\sigma_{\rho_0})\|\cM_{\rho_0}^{\ox m}(\omega_{\rho_0})).
\end{align}  Choose
$t\in(0,1)$, independent of $\rho_0$, such that
$-\log(1-t)\le m\gamma$, and set the perturbed state
$\widetilde\omega_{\rho_0}:=(1-t)\omega_{\rho_0}+t\sigma_{\rho_0}$.
For this fixed pair, define
\begin{align}
g_{\rho_0}(\rho)
:=
D\!\left(
\cN_\rho^{\ox m}(\sigma_{\rho_0})
\,\|\,
\cM_\rho^{\ox m}(\widetilde\omega_{\rho_0})
\right).
\end{align}
Eq.~\eqref{eq:fixed-block-mixed-jammer-domination} gives, uniformly
for every $\rho\in\density(AR)$,
\begin{align}
\cN_\rho^{\ox m}(\sigma_{\rho_0})
\le
\frac{\kappa^m}{t}
\cM_\rho^{\ox m}(\widetilde\omega_{\rho_0}).
\end{align}
The two outputs depend continuously on $\rho$ and satisfy the uniform
domination relation above.  Proposition~37(ii) of~\cite{mosonyi2023some},
applied at $\alpha=1$ with $f(s)=(\kappa^m/t)s$, therefore implies that
$g_{\rho_0}$ is continuous on $\density(AR)$.  At the center $\rho_0$,
Eq.~\eqref{eq:fixed-block-mixed-objective-upper-bound} and the choice of
$m$ give
\begin{align}
g_{\rho_0}(\rho_0)
&\leq D(\cN_{\rho_0}^{\ox m}(\sigma_{\rho_0})\|\cM_{\rho_0}^{\ox m}(\omega_{\rho_0}))-\log(1-t)\\
& =
D^{\downarrow\downarrow}
(\cN_{\rho_0}^{\ox m}\|\cM_{\rho_0}^{\ox m})
-\log(1-t)
\\
&\le
D^{\emptyuparrow,\downarrow\downarrow}
(\cN^{\ox m}\|\cM^{\ox m})
-\log(1-t)
\\
&\le
m\!\left(
D^{\infty,\emptyuparrow,\downarrow\downarrow}(\cN\|\cM)
+2\gamma
\right).
\end{align}
Consequently, continuity implies that the relatively open set
\begin{align}
U_{\rho_0}
:=\left\{
\rho\in\density(AR):
g_{\rho_0}(\rho)
<m\!\left(
D^{\infty,\emptyuparrow,\downarrow\downarrow}(\cN\|\cM)
+3\gamma
\right)
\right\}
\end{align}
contains $\rho_0$.  Thus
$\{U_{\rho_0}:\rho_0\in\density(AR)\}$ is an open cover of
$\density(AR)$.  The choice $R\simeq A$ made above fixes
this state space independently of $m$ and $n$.  By compactness, there are finitely many centers
$\rho_1,\ldots,\rho_L$ whose neighborhoods cover $\density(AR)$.  Setting
$U_i:=U_{\rho_i}$,
$\sigma_i:=\sigma_{\rho_i}$, and
$\widetilde\omega_i:=\widetilde\omega_{\rho_i}$ gives
\begin{align}
\rho\in U_i
\quad\Longrightarrow\quad
D\!\left(
\cN_\rho^{\ox m}(\sigma_i)
\,\|\,
\cM_\rho^{\ox m}(\widetilde\omega_i)
\right)
\le
m\!\left(
D^{\infty,\emptyuparrow,\downarrow\downarrow}(\cN\|\cM)
+3\gamma
\right).
\label{eq:public-full-jammer-finite-cover}
\end{align}
This indicates that the fixed pair $(\sigma_i,\widetilde \omega_i)$ provides the required bound for every tester input $\rho \in U_i$.
It is important that all objects used to construct this cover are fixed after choosing $\gamma$
and $m$, before the overall blocklength $n$ and testing error $\ve$ enter
the argument.  

Now we can construct a universal jammer state based on the finite subcover. Write $n=qm+r$, where $0\le r<m$, and choose
$\tau_r\in\density(E^r)$.  For
$1\le i\le L$, define
\begin{align}
\sigma_{i,n}
&:=\sigma_i^{\ox q}\ox\tau_r, \qquad
\omega_{i,n}
 :=(\widetilde\omega_i)^{\ox q}\ox\tau_r,
\end{align}
and consider the averaged state
\begin{align}
  \widebar\tau_n
:=\frac1{2L}\sum_{i=1}^L
\left(\sigma_{i,n}+\omega_{i,n}\right).
\label{eq:public-full-jammer-correlated-common-state}
\end{align}
This average constructs a state that is independent of $\rho$ and is used under both
hypotheses; this is the crucial reduction to the public hypothesis-unaware game. Next, we will show that this is indeed a good choice for the asymptotic estimation.

Fix $\rho$ and choose $i$ such that $\rho\in U_i$.  Each relevant output
contains its selected component with weight $1/(2L)$, so positivity gives
\begin{align}
\cN_\rho^{\ox n}(\widebar\tau_n)
&\ge
\frac1{2L}\cN_\rho^{\ox n}(\sigma_{i,n}),\\
\cM_\rho^{\ox n}(\widebar\tau_n)
& \ge
\frac1{2L}\cM_\rho^{\ox n}(\omega_{i,n}).
\end{align}
Applying
Lemma~\ref{lem:dominated-component-hypothesis-testing-bound} with
$M=N=2L$ gives, whenever $2L\ve<1$,
\begin{align}
D_{\Hypo,\ve}\!\left(
\cN_\rho^{\ox n}(\widebar\tau_n)
\,\|\,
\cM_\rho^{\ox n}(\widebar\tau_n)
\right)
\le
D_{\Hypo,2L\ve}\!\left(
\cN_\rho^{\ox n}(\sigma_{i,n})
\,\|\,
\cM_\rho^{\ox n}(\omega_{i,n})
\right)
+\log(2L).
\label{eq:public-full-jammer-common-state-DH-bound}
\end{align}

For each
$\rho\in\density(AR)$, retain a cover index $i=i(\rho)$ such that
$\rho\in U_i$. Applying Lemma~\ref{lem:one-shot-hypothesis-testing-bounds} to the right-hand side of
Eq.~\eqref{eq:public-full-jammer-common-state-DH-bound} gives
\begin{align}
D_{\Hypo,\ve}\!\left(
\cN_\rho^{\ox n}(\widebar\tau_n)
\,\|\,
\cM_\rho^{\ox n}(\widebar\tau_n)
\right)
& \le
\frac{
D\!\left(
\cN_\rho^{\ox n}(\sigma_{i,n})
\,\|\,
\cM_\rho^{\ox n}(\omega_{i,n})\right)
}{1-2L\ve} 
+ c.
\end{align}
where
\begin{align}
  c = \frac{
h(2L\ve)
+2L\ve n\log \mu
}{1-2L\ve}
+\log(2L).
\end{align}

Fix $\rho$ and choose $i$ such that $\rho\in U_i$.  The tensor-product
forms of $\sigma_{i,n}$ and $\omega_{i,n}$, additivity of relative entropy,
Eq.~\eqref{eq:public-full-jammer-finite-cover}, and the remainder bounds in
the proof of Lemma~\ref{lem:uniform-block-bounds} give
\begin{align}
D\!\left(
\cN_\rho^{\ox n}(\sigma_{i,n})
\,\|\,
\cM_\rho^{\ox n}(\omega_{i,n})
\right)
&=
qD\!\left(
\cN_\rho^{\ox m}(\sigma_i)
\,\|\,
\cM_\rho^{\ox m}(\widetilde\omega_i)
\right)
+
D\!\left(
\cN_\rho^{\ox r}(\tau_r)
\,\|\,
\cM_\rho^{\ox r}(\tau_r)
\right)
\\
&\le
qm\!\left(
D^{\infty,\emptyuparrow,\downarrow\downarrow}(\cN\|\cM)
+3\gamma
\right)
+r\log\kappa.
\label{eq:public-full-jammer-component-bounds}
\end{align}
This implies that
\begin{align}
 D_{\Hypo,\ve}\!\left(
\cN_\rho^{\ox n}(\widebar\tau_n)
\,\|\,
\cM_\rho^{\ox n}(\widebar\tau_n)
\right) & \leq \frac{qm\!\left(
D^{\infty,\emptyuparrow,\downarrow\downarrow}(\cN\|\cM)
+3\gamma
\right)
+r\log\kappa}{1-2L\ve} + c
\end{align}
Crucially, $\widebar\tau_n$ is independent of $\rho$, and the preceding
bound therefore holds simultaneously for all tester inputs.  We may first
take the tester supremum at this fixed jammer state and then evaluate the
outer jammer infimum at the feasible choice $\widebar\tau_n$.  Explicitly,
\begin{align}
\frac1n D_{\Hypo,\ve}^{\downarrow,\emptyuparrow}(\cN^{\ox n}\|\cM^{\ox n})
&=\frac1n
\inf_{\tau_n\in\density(E^n)}
\sup_{\rho\in\density(AR)}
D_{\Hypo,\ve}\!\left(
\cN_\rho^{\ox n}(\tau_n)
\,\big\|\,
\cM_\rho^{\ox n}(\tau_n)
\right)
\\
&\le\frac1n
\sup_{\rho\in\density(AR)}
D_{\Hypo,\ve}\!\left(
\cN_\rho^{\ox n}(\widebar\tau_n)
\,\big\|\,
\cM_\rho^{\ox n}(\widebar\tau_n)
\right)
\\
&\le
\frac{
qm\!\left(
D^{\infty,\emptyuparrow,\downarrow\downarrow}(\cN\|\cM)
+3\gamma
\right)
+r\log\kappa
}{
n(1-2L\ve)
}
+ \frac{c}{n}.
\end{align}
For fixed $\gamma$, the quantities $m$ and $L$ are fixed.  Since
$qm/n\to1$ and $r/n\to0$, taking $n\to\infty$ gives, for $2L\ve<1$,
\begin{align}
\limsup_{n\to\infty}\frac1n
D_{\Hypo,\ve}^{\downarrow,\emptyuparrow}(\cN^{\ox n}\|\cM^{\ox n})
\le
\frac{
D^{\infty,\emptyuparrow,\downarrow\downarrow}(\cN\|\cM)
+3\gamma+2L\ve\log\mu}
{1-2L\ve}.
\end{align}
Letting $\ve\to0^+$ while $\gamma$, and hence $L$, remains fixed, and only
then letting $\gamma\to0^+$ gives
\begin{align}
\lim_{\ve\to0^+}\limsup_{n\to\infty}
\frac1n
D_{\Hypo,\ve}^{\downarrow,\emptyuparrow}(\cN^{\ox n}\|\cM^{\ox n})
\le
D^{\infty,\emptyuparrow,\downarrow\downarrow}(\cN\|\cM).
\end{align}
This completes the proof.
\end{proof}

\section{IID tester against IID jammer}\label{sec: iid inputs}

In this section, we consider an IID tester against an IID jammer.
Unlike in the previous settings, the asymptotic Stein exponent can differ significantly across the different information patterns.

The two public-jammer cases are relatively straightforward by applying Lemma~\ref{lem:one-shot-hypothesis-testing-bounds}, the established continuity in Lemma~\ref{lem: DM alpha jam continuity} and the minimax property in Proposition~\ref{prop: jammer divergence minimax}.

\begin{boxtheorem}[IID tester against public IID jammer]
\label{thm: public hypo hypothesis-unaware iid}
Let $\cN\in\CPTP(AE\!:\!B)$ and $\cM\in\CP(AE\!:\!B)$ satisfy
$D_{\max}^{\uparrow}(\cN\|\cM)<\infty$.  The following asymptotic Stein exponents hold
\begin{align}
\lim_{n\to\infty}
\frac1n
\Stein_{\ve,n}^{\emptydownarrow,\emptyuparrow}(\cN,\cM)
&=
D^{\uparrow,\downarrow}(\cN\|\cM),\qquad \ \forall \ve\in(0,1),
\label{eq:public-hypothesis-unaware-iid-Stein}
\\
\lim_{n\to\infty}
\frac1n
\Stein_{\ve,n}^{\emptydownarrow\emptydownarrow,\emptyuparrow}(\cN,\cM)
&=
D^{\uparrow,\downarrow\downarrow}(\cN\|\cM), \qquad \forall \ve\in(0,1).
\label{eq:public-hypothesis-aware-iid-Stein}
\end{align}
\end{boxtheorem}

\begin{proof}
We first prove the hypothesis-unaware case; the proof works the same for the other case. Let
$\alpha\in(1/2,1)$ and $\alpha'>1$.  Applying Lemma~\ref{lem:one-shot-hypothesis-testing-bounds} gives
\begin{align}
D_{\Sand,\alpha}^{\emptydownarrow,\emptyuparrow}(\cN^{\ox n}\|\cM^{\ox n})
+\frac{\alpha}{\alpha-1}\log\frac1\ve
&\le
D_{\Hypo,\ve}^{\emptydownarrow,\emptyuparrow}(\cN^{\ox n}\|\cM^{\ox n})
\le
D_{\Sand,\alpha'}^{\emptydownarrow,\emptyuparrow}(\cN^{\ox n}\|\cM^{\ox n})
+\frac{\alpha'}{\alpha'-1}\log\frac1{1-\ve}.
\label{eq:public-hypothesis-unaware-iid-renyi-sandwich}
\end{align}
As the optimizations are restricted to the tensor-power states, we have 
\begin{align}
D_{\Sand,\alpha}^{\emptydownarrow,\emptyuparrow}(\cN^{\ox n}\|\cM^{\ox n})
&=
n D_{\Sand,\alpha}^{\downarrow,\uparrow}(\cN\|\cM).
\end{align}
Dividing by $n$ and taking $n\to\infty$, we have 
\begin{align}
D_{\Sand,\alpha}^{\downarrow,\uparrow}(\cN\|\cM) \leq \liminf_{n\to\infty} \frac1n D_{\Hypo,\ve}^{\emptydownarrow,\emptyuparrow}(\cN^{\ox n}\|\cM^{\ox n}) \leq \limsup_{n\to\infty} \frac1n D_{\Hypo,\ve}^{\emptydownarrow,\emptyuparrow}(\cN^{\ox n}\|\cM^{\ox n}) \leq D_{\Sand,\alpha'}^{\downarrow,\uparrow}(\cN \|\cM).
\end{align}
Taking the supremum over $\alpha$, the infimum over $\alpha'$, and applying
Proposition~\ref{prop: jammer divergence minimax} and
Lemma~\ref{lem: DM alpha jam continuity}, we have
\begin{align}
\lim_{n\to\infty} \frac1n D_{\Hypo,\ve}^{\emptydownarrow,\emptyuparrow}(\cN^{\ox n}\|\cM^{\ox n}) = D^{\downarrow,\uparrow}(\cN\|\cM) = D^{\uparrow,\downarrow}(\cN\|\cM).
\end{align}
By Theorem~\ref{thm:unified-operational-correspondence}, the left-hand side is
the hypothesis-unaware Stein exponent. The
hypothesis-aware case follows identically with the corresponding double-arrow
divergences.
\end{proof}

We next consider the two secret-jammer cases, whose Stein exponents are
characterized by regularized channel divergences.

\begin{boxtheorem}[IID tester against secret IID jammer]
\label{thm: secret hypo hypothesis-aware iid}
Let $\cN\in\CPTP(AE\!:\!B)$ and $\cM\in\CP(AE\!:\!B)$ satisfy
$D_{\max}^{\uparrow}(\cN\|\cM)<\infty$.  For
$(\star)\in
\bigl\{
(\emptyuparrow,\emptydownarrow),
(\emptyuparrow,\emptydownarrow\emptydownarrow)
\bigr\}$,
the asymptotic Stein exponents are
\begin{align}
&\lim_{\ve\to0^+}\liminf_{n\to\infty}
\frac1n
\Stein_{\ve,n}^{(\star)}(\cN,\cM)
=
\lim_{\ve\to0^+}\limsup_{n\to\infty}
\frac1n
\Stein_{\ve,n}^{(\star)}(\cN,\cM)
=
D^{\infty,
\emptyuparrow,
\convdownarrow\convdownarrow}(\cN\|\cM).
\label{eq:theorem73-vanishing-error-Stein-revised}
\end{align}
\end{boxtheorem}

\begin{proof}
By Theorem~\ref{thm:unified-operational-correspondence}, for either
$(\star)\in\{(\emptyuparrow,\emptydownarrow),
(\emptyuparrow,\emptydownarrow\emptydownarrow)\}$,
\begin{align}
\Stein_{\ve,n}^{(\star)}(\cN,\cM)
&=
D_{\Hypo,\ve}^{\emptyuparrow,\convdownarrow\convdownarrow}
(\cN^{\ox n}\|\cM^{\ox n}).
\end{align}
It suffices to establish the asserted asymptotic formula for the
divergence on the right. We first prove the existence of the regularized
Umegaki divergence and then the achievability and converse bounds.

\prooftag{Limit existence.} 
Fix $\rho\in\density(AR)$.  For
$X\in\PSD((BR)^n)$ and $Y\in\PSD((BR)^m)$, linearity gives
\begin{align}
&\sup_{\xi_{n+m}\in\conv(\sI(E^{n+m}))}
\tr\!\left[
(X\ox Y)\cN_\rho^{\ox(n+m)}(\xi_{n+m})
\right]
\\
&\quad=
\sup_{\tau\in\density(E)}
\tr\!\left[X\cN_\rho^{\ox n}(\tau^{\ox n})\right]
\tr\!\left[Y\cN_\rho^{\ox m}(\tau^{\ox m})\right]
\\
&\quad\le
\sup_{\xi_n\in\conv(\sI(E^n))}
\tr\!\left[X\cN_\rho^{\ox n}(\xi_n)\right]
\sup_{\xi_m\in\conv(\sI(E^m))}
\tr\!\left[Y\cN_\rho^{\ox m}(\xi_m)\right].
\label{eq:iid-iid-output-support-submultiplicativity}
\end{align}
The same calculation holds with $\cN$ replaced by $\cM$.  Thus the
positive polars of both compact convex output-set sequences are stable
under tensor products by~\cite[Lemma~7]{fang2024generalized}. The
superadditivity results for $D_{\Meas,\alpha}$ and $D_{\Meas}$
in~\cite[Lemmas~21 and~23, respectively]{fang2024generalized} therefore
give, for every $\alpha\in(0,1)$,
\begin{align}
D_{\Meas,\alpha}^{\convdownarrow\convdownarrow}
(\cN_\rho^{\ox(n+m)}\|\cM_\rho^{\ox(n+m)})
&\ge
D_{\Meas,\alpha}^{\convdownarrow\convdownarrow}
(\cN_\rho^{\ox n}\|\cM_\rho^{\ox n})
+
D_{\Meas,\alpha}^{\convdownarrow\convdownarrow}
(\cN_\rho^{\ox m}\|\cM_\rho^{\ox m}),
\label{eq:iid-iid-fixed-tester-measured-Renyi-superadditivity}
\\
D_{\Meas}^{\convdownarrow\convdownarrow}
(\cN_\rho^{\ox(n+m)}\|\cM_\rho^{\ox(n+m)})
&\ge
D_{\Meas}^{\convdownarrow\convdownarrow}
(\cN_\rho^{\ox n}\|\cM_\rho^{\ox n})
+
D_{\Meas}^{\convdownarrow\convdownarrow}
(\cN_\rho^{\ox m}\|\cM_\rho^{\ox m}).
\label{eq:iid-iid-fixed-tester-measured-superadditivity}
\end{align}
Then we have 
\begin{align}
  \liminf_{n\to \infty} \frac1n & \sup_{\rho\in\density(AR)}  D_{\Meas}^{\convdownarrow\convdownarrow}(\cN_\rho^{\ox n}\|\cM_\rho^{\ox n})\notag\\
   & \geq \sup_{\rho\in\density(AR)} \liminf_{n\to \infty} \frac1n  D_{\Meas}^{\convdownarrow\convdownarrow}(\cN_\rho^{\ox n}\|\cM_\rho^{\ox n})\\
  & = \sup_{\rho\in\density(AR)} \sup_{n\in \NN} \frac1n  D_{\Meas}^{\convdownarrow\convdownarrow}(\cN_\rho^{\ox n}\|\cM_\rho^{\ox n}),
\end{align}
where the equality follows from the superadditivity in Eq.~\eqref{eq:iid-iid-fixed-tester-measured-superadditivity} and Fekete's lemma.  The limit exists and is equal to the supremum over $n$.
On the other hand, every limit superior is bounded by the supremum over $k$ of the finite-block quantities,
\begin{align}
&\limsup_{n\to\infty}\frac1n
\sup_{\rho\in\density(AR)}
D_{\Meas}^{\convdownarrow\convdownarrow}
(\cN_\rho^{\ox n}\|\cM_\rho^{\ox n})
\le
\sup_{k\in\NN}\frac1k
\sup_{\rho\in\density(AR)}
D_{\Meas}^{\convdownarrow\convdownarrow}
(\cN_\rho^{\ox k}\|\cM_\rho^{\ox k}).
\end{align}
This proves that the limit exists and gives
\begin{align}
&\lim_{n\to\infty}\frac1n
\sup_{\rho\in\density(AR)}
D_{\Meas}^{\convdownarrow\convdownarrow}
(\cN_\rho^{\ox n}\|\cM_\rho^{\ox n})
=
\sup_{k\in\NN}\frac1k
\sup_{\rho\in\density(AR)}
D_{\Meas}^{\convdownarrow\convdownarrow}
(\cN_\rho^{\ox k}\|\cM_\rho^{\ox k}).
\label{eq:iid-iid-measured-rate-existence}
\end{align}

It remains to compare this limit with the Umegaki divergence. Set
$d:=\dim\cH_{BR}$ and
$f(n):=2\log[(n+1)^{d}(n+d)^{d^2}]$.
Every alternative output obtained from $\conv(\sI(E^n))$ is permutation
invariant on $\cH_{BR}^{\ox n}$. At order one,
Lemma~\ref{Lemma: DM and Sandwiched relation} gives the measured-to-Umegaki
comparison. Taking the two jammer infima and the tester supremum gives
\begin{align}
&\sup_{\rho\in\density(AR)}
D_{\Meas}^{\convdownarrow\convdownarrow}
(\cN_\rho^{\ox n}\|\cM_\rho^{\ox n})
\le
D^{\emptyuparrow,\convdownarrow\convdownarrow}
(\cN^{\ox n}\|\cM^{\ox n})
\le
\sup_{\rho\in\density(AR)}
D_{\Meas}^{\convdownarrow\convdownarrow}
(\cN_\rho^{\ox n}\|\cM_\rho^{\ox n})+f(n).
\label{eq:iid-iid-measured-Umegaki-comparison}
\end{align}
Since $f(n)/n\to0$, Eq.~\eqref{eq:iid-iid-measured-rate-existence}
and the squeeze theorem give the existence of the limit
\begin{align}
D^{\infty,\emptyuparrow,\convdownarrow\convdownarrow}(\cN\|\cM)
&:=
\lim_{n\to\infty}\frac1n
D^{\emptyuparrow,\convdownarrow\convdownarrow}
(\cN^{\ox n}\|\cM^{\ox n})
=
\sup_{k\in\NN}\frac1k
\sup_{\rho\in\density(AR)}
D_{\Meas}^{\convdownarrow\convdownarrow}
(\cN_\rho^{\ox k}\|\cM_\rho^{\ox k}).
\label{eq:iid-iid-convexified-rate-existence}
\end{align}
Lemma~\ref{lem:uniform-block-bounds} implies that the limit is finite.

\bigskip
\prooftag{Achievability bound.} Fix $\ve\in(0,1)$ and $\alpha\in(1/2,1)$.  The one-shot lower bound in
Lemma~\ref{lem:one-shot-hypothesis-testing-bounds}, together with
$D_{\Meas,\alpha}\le D_{\Sand,\alpha}$, gives
\begin{align}
&D_{\Hypo,\ve}^{\emptyuparrow,
\convdownarrow\convdownarrow}
(\cN^{\ox n}\|\cM^{\ox n})
\ge
\sup_{\rho\in\density(AR)}
D_{\Meas,\alpha}^{\convdownarrow\convdownarrow}
(\cN_\rho^{\ox n}\|\cM_\rho^{\ox n})
+
\frac{\alpha}{\alpha-1}\log\frac1\ve.
\label{eq:theorem73-achievability-one-shot-revised}
\end{align}
After division by $n$, the final term in
Eq.~\eqref{eq:theorem73-achievability-one-shot-revised} vanishes. For every
fixed $\rho$, Eq.~\eqref{eq:iid-iid-fixed-tester-measured-Renyi-superadditivity}
and Fekete's lemma show that the corresponding normalized sequence converges
to the supremum of its finite-block values. Hence, for every $k\in\NN$,
\begin{align}
 \liminf_{n\to\infty}&\frac1n
D_{\Hypo,\ve}^{\emptyuparrow,
\convdownarrow\convdownarrow}
(\cN^{\ox n}\|\cM^{\ox n})\notag\\
&\ge
\liminf_{n\to\infty}\frac1n
\sup_{\rho\in\density(AR)}
D_{\Meas,\alpha}^{\convdownarrow\convdownarrow}
(\cN_\rho^{\ox n}\|\cM_\rho^{\ox n})
\\
&\ge
\sup_{\rho\in\density(AR)}
\lim_{n\to\infty}\frac1n
D_{\Meas,\alpha}^{\convdownarrow\convdownarrow}
(\cN_\rho^{\ox n}\|\cM_\rho^{\ox n})
\\
&=
\sup_{\rho\in\density(AR)}\sup_{m\in\NN}\frac1m
D_{\Meas,\alpha}^{\convdownarrow\convdownarrow}
(\cN_\rho^{\ox m}\|\cM_\rho^{\ox m})
\\
&\ge
\frac1k\sup_{\rho\in\density(AR)}
D_{\Meas,\alpha}^{\convdownarrow\convdownarrow}
(\cN_\rho^{\ox k}\|\cM_\rho^{\ox k}).
\label{eq:theorem73-achievability-fixed-block}
\end{align}
The left-hand side is independent of $\alpha$. Taking the supremum over
$\alpha\in(1/2,1)$, exchanging the two suprema, and applying the order-one
continuity for compact convex sets in~\cite[Lemma~22]{fang2024generalized}
to the fixed-$\rho$ output sets yield
\begin{align}
\liminf_{n\to\infty}\frac1n
D_{\Hypo,\ve}^{\emptyuparrow,
\convdownarrow\convdownarrow}
(\cN^{\ox n}\|\cM^{\ox n})
&\ge
\frac1k\sup_{\rho\in\density(AR)}
D_{\Meas}^{\convdownarrow\convdownarrow}
(\cN_\rho^{\ox k}\|\cM_\rho^{\ox k}).
\label{eq:theorem73-achievability-measured-block}
\end{align}
Since this holds for every $k\in\NN$,
Eq.~\eqref{eq:iid-iid-convexified-rate-existence} gives
\begin{align}
\liminf_{n\to\infty}\frac1n
D_{\Hypo,\ve}^{\emptyuparrow,
\convdownarrow\convdownarrow}
(\cN^{\ox n}\|\cM^{\ox n})
&\ge
D^{\infty,\emptyuparrow,\convdownarrow\convdownarrow}(\cN\|\cM).
\label{eq:theorem73-achievability-final-revised}
\end{align}

\prooftag{Converse bound.}
Let $\mu$ be as in Lemma~\ref{lem:uniform-block-bounds}.
The trace estimate in the proof of that lemma and
Lemma~\ref{lem:one-shot-hypothesis-testing-bounds}, applied pointwise, give
the following bound after taking the infimum over the jammer pair and then the
supremum over the tester:
\begin{align}
&D_{\Hypo,\ve}^{\emptyuparrow,
\convdownarrow\convdownarrow}
(\cN^{\ox n}\|\cM^{\ox n})
\le
\frac{
D^{\emptyuparrow,
\convdownarrow\convdownarrow}
(\cN^{\ox n}\|\cM^{\ox n})
+h(\ve)+\ve n\log \mu}
{1-\ve}.
\label{eq:theorem73-composite-one-shot-upper}
\end{align}
Dividing by $n$, taking the limit superior, and using
Eq.~\eqref{eq:iid-iid-convexified-rate-existence} gives
\begin{align}
&\limsup_{n\to\infty}\frac1n
D_{\Hypo,\ve}^{\emptyuparrow,
\convdownarrow\convdownarrow}
(\cN^{\ox n}\|\cM^{\ox n})
\le
\frac{
D^{\infty,\emptyuparrow,
\convdownarrow\convdownarrow}(\cN\|\cM)
+\ve\log\mu}
{1-\ve}.
\end{align}
Letting $\ve\to0^+$ proves the desired converse direction.
\end{proof}

\begin{boxremark}
\label{rem:iid-exponent-separation}
In general, neither equality among the three asymptotic Stein exponents in this
section holds universally.  Appendix~\ref{sec:strict-iid-gap} gives two qubit
channel pairs for strict differences.  Example~\ref{exam:iid-awareness-gap} shows that
$D^{\uparrow,\downarrow}>D^{\uparrow,\downarrow\downarrow}=0$.
Example~\ref{exam:strict-iid-gap} gives a different channel pair for which
$D^{\uparrow,\downarrow\downarrow}>
D^{\infty,\emptyuparrow,\convdownarrow\convdownarrow}$.
\end{boxremark}

\section{Replacer alternatives}
\label{sec:replacer-alternatives}

Replacer channels form a structured class of alternatives whose output is
fixed independently of all channel inputs.  In the split-control setting, this
property removes the alternative-side jammer variable from the optimization.
We prove that all eleven minimax/maximin Umegaki divergences in Definition~\ref{def:resource-dependent-block-divergences} are additive
against a replacer alternative and, consequently, that the twelve
vanishing-error Stein exponents coincide at one single-letter value.
 
For a fixed state $\tau_B$, the replacer channel
$\cR_\tau\in\CPTP(AE\!:\!B)$ discards its input and prepares $\tau_B$:
\begin{align}
\cR_\tau(X_{AE})&:=\tr[X_{AE}]\tau_B.
\end{align}
If the tester retains a reference $R$, then for every
$\rho_{AR}\in\density(AR)$ and $\omega_E\in\density(E)$,
\begin{align}
(\cR_\tau\ox\id_R)(\rho_{AR}\ox\omega_E)
&=\tau_B\ox\rho_R.
\label{eq:replacer-fixed-output}
\end{align}
Thus the output retains the tester reference but is independent of the
jammer input.  

The following theorem states the finite-block additivity of the eleven
Umegaki divergences and its asymptotic operational consequence.

\begin{boxtheorem}[Universal collapse for replacer alternatives]
\label{thm:replacer-universal-single-letter}
Let $\cN\in\CPTP(AE\!:\!B)$ and let
$\cR_\tau\in\CPTP(AE\!:\!B)$ be the replacer with output
$\tau_B\in\density(B)$.  Suppose that
$D_{\max}^{\uparrow}(\cN\|\cR_\tau)<\infty$.  Let $(\diamond)$ denote any
divergence index in Definition~\ref{def:resource-dependent-block-divergences},
with the underlying
quantum divergence chosen as the Umegaki relative entropy. Then, for every
$n\in\NN$,
\begin{align}
D^{(\diamond)}(\cN^{\ox n}\|\cR_\tau^{\ox n})
&=nD^{\uparrow,\downarrow}(\cN\|\cR_\tau),\qquad
D^{\infty,(\diamond)}(\cN\|\cR_\tau)
=D^{\uparrow,\downarrow}(\cN\|\cR_\tau).
\label{eq:replacer-universal-additivity}
\end{align}
Consequently, for every operational arrow code $(\star)$ in
Table~\ref{tab:finite-block-Stein-quantities}, the asymptotic Stein exponent has the same
single-letter value:
\begin{align}
&\lim_{\ve\to0^+}\liminf_{n\to\infty}
\frac1n\Stein_{\ve,n}^{(\star)}(\cN,\cR_\tau)
=
\lim_{\ve\to0^+}\limsup_{n\to\infty}
\frac1n\Stein_{\ve,n}^{(\star)}(\cN,\cR_\tau)
=D^{\uparrow,\downarrow}(\cN\|\cR_\tau).
\label{eq:replacer-universal-Stein-exponent}
\end{align}
\end{boxtheorem}

\begin{proof}
For every $\rho_n\in\density((AR)^n)$ and
$\omega_n\in\density(E^n)$, Eq.~\eqref{eq:replacer-fixed-output} gives
\begin{align}
(\cR_\tau^{\ox n}\ox\id_{R^n})(\rho_n\ox\omega_n)
=\tau_B^{\ox n}\ox(\rho_n)_{R^n}.
\label{eq:replacer-block-fixed-output}
\end{align}
Thus the alternative jammer input is irrelevant, and the common- and
separate-input variants have the same objective for every fixed tester input
and null-hypothesis jammer input.  At one channel use, we consequently have
\begin{align}
D^{\uparrow,\downarrow}(\cN\|\cR_\tau)
&=\sup_{\rho\in\density(AR)}\inf_{\sigma\in\density(E)}
D\!\left(\cN_\rho(\sigma)\middle\|\tau_B\ox\rho_R\right)
=\inf_{\sigma\in\density(E)}\sup_{\rho\in\density(AR)}
D\!\left(\cN_\rho(\sigma)\middle\|\tau_B\ox\rho_R\right).
\label{eq:replacer-one-shot-game-value}
\end{align}
The first equality is the definition of the maximin divergence after applying
Eq.~\eqref{eq:replacer-block-fixed-output}; the second is the minimax identity
in Proposition~\ref{prop: jammer divergence minimax}.

At blocklength $n$, the relevant resource domains satisfy
\begin{align}
\sI((AR)^n)&\subseteq\density((AR)^n),
\qquad
\sI(E^n)\subseteq\conv(\sI(E^n))\subseteq\density(E^n).
\end{align}
For every divergence index $(\diamond)$ in
Definition~\ref{def:resource-dependent-block-divergences}, domain monotonicity
and the minimax inequality give the universal sandwich
\begin{align}
D^{\emptyuparrow,\downarrow}(\cN^{\ox n}\|\cR_\tau^{\ox n})
&\le D^{(\diamond)}(\cN^{\ox n}\|\cR_\tau^{\ox n})
\le D^{\emptydownarrow,\uparrow}(\cN^{\ox n}\|\cR_\tau^{\ox n}),
\label{eq:replacer-resource-domain-sandwich}
\end{align}
where
\begin{align}
D^{\emptydownarrow,\uparrow}(\cN^{\ox n}\|\cR_\tau^{\ox n}) := \inf_{\sigma\in\density(E)}
\sup_{\rho_n\in\density((AR)^n)}
D\!\left(
(\cN^{\ox n}\ox\id_{R^n})(\rho_n\ox\sigma^{\ox n})
\middle\|
\tau_B^{\ox n}\ox(\rho_n)_{R^n}
\right)
\end{align}
is a variant of minimax divergence not explicitly defined in
Definition~\ref{def:resource-dependent-block-divergences}.

It remains to evaluate the two endpoints in
Eq.~\eqref{eq:replacer-resource-domain-sandwich}.  We first show that the
jammer-input and tester-input channel divergences are additive against replacers.
For arbitrary
$\cN_i\in\CPTP(E_i\!:\!B_i)$ and replacers with
outputs $\tau_i$, we have
\begin{align}
D^{\downarrow}(\cN_1\ox\cN_2\|\cR_{\tau_1}\ox\cR_{\tau_2})
&=D^{\downarrow}(\cN_1\|\cR_{\tau_1})
+D^{\downarrow}(\cN_2\|\cR_{\tau_2}).
\label{eq:replacer-jammer-input-additivity}
\end{align}
One direction of the equality is immediate from the definition of the jammer-input divergence.  The other direction follows from the inequality for the Umegaki relative entropy:
$D(\zeta_{B_1B_2}\|\tau_1\ox\tau_2)
\geq D(\zeta_{B_1}\|\tau_1)+D(\zeta_{B_2}\|\tau_2)$.
The tester-input divergence is also additive against replacers:
\begin{align}
D^{\uparrow}(\cN_1\ox\cN_2\|\cR_{\tau_1}\ox\cR_{\tau_2})
&=D^{\uparrow}(\cN_1\|\cR_{\tau_1})+D^{\uparrow}(\cN_2\|\cR_{\tau_2}).
\label{eq:replacer-tester-input-additivity}
\end{align}
Indeed, for any joint input and resulting output $\zeta_{RB_1B_2}$, the chain
rule gives
\begin{align}
D(\zeta_{RB_1B_2}\|\zeta_R\ox\tau_1\ox\tau_2)
&=D(\zeta_{RB_1B_2}\|\zeta_{RB_2}\ox\tau_1)
+D(\zeta_{RB_2}\|\zeta_R\ox\tau_2).
\end{align}
Regarding $RB_2$ as the reference for the first channel bounds the first term
by $D^{\uparrow}(\cN_1\|\cR_{\tau_1})$, and the second term is bounded by
$D^{\uparrow}(\cN_2\|\cR_{\tau_2})$.  Product inputs give the reverse
inequality.

For $\rho_{AR}\in\density(AR)$, let $\cN_\rho:E\to BR$ denote the induced
channel $\cN_\rho(X_E):=(\cN\ox\id_R)(\rho_{AR}\ox X_E)$.  Let
$\cR_{\tau\ox\rho_R}:E\to BR$ be the replacer with output
$\tau_B\ox\rho_R$.  For $\sigma_E\in\density(E)$, let
$\cN_\sigma:A\to B$ denote the induced channel
$\cN_\sigma(X_A):=\cN(X_A\ox\sigma_E)$.  In the latter comparison,
$\cR_\tau:A\to B$ denotes the replacer with output $\tau_B$.

The lower endpoint admits the chain
\begin{align}
D^{\emptyuparrow,\downarrow}(\cN^{\ox n}\|\cR_\tau^{\ox n})
&=\sup_{\rho\in\density(AR)}
D^{\downarrow}\!\left(
\cN_\rho^{\ox n}\middle\|\cR_{\tau\ox\rho_R}^{\ox n}
\right)\\
&=n\sup_{\rho\in\density(AR)}
D^{\downarrow}\!\left(
\cN_\rho\middle\|\cR_{\tau\ox\rho_R}
\right)\\
&=nD^{\uparrow,\downarrow}(\cN\|\cR_\tau).
\label{eq:replacer-lower-endpoint}
\end{align}
The first equality is by the
definition of $D^{\downarrow}$.  The second follows by iterating
Eq.~\eqref{eq:replacer-jammer-input-additivity}, and the third follows from
Remark~\ref{rem: jammer divergence sup of jammer-input}.

Similarly, the explicit upper bound admits the chain
\begin{align}
D^{\emptydownarrow,\uparrow}(\cN^{\ox n}\|\cR_\tau^{\ox n})
&=\inf_{\sigma\in\density(E)}
D^{\uparrow}\!\left(
\cN_\sigma^{\ox n}\middle\|\cR_\tau^{\ox n}
\right)\\
&=n\inf_{\sigma\in\density(E)}
D^{\uparrow}\!\left(
\cN_\sigma\middle\|\cR_\tau
\right)\\
&=nD^{\uparrow,\downarrow}(\cN\|\cR_\tau).
\label{eq:replacer-upper-endpoint}
\end{align}
The first equality is the definition of $D^{\uparrow}$.  The second follows by
iterating Eq.~\eqref{eq:replacer-tester-input-additivity}, and the third follows
from Remark~\ref{rem: jammer divergence sup of tester-input}.

Combining Eqs.~\eqref{eq:replacer-resource-domain-sandwich},
\eqref{eq:replacer-lower-endpoint}, and
\eqref{eq:replacer-upper-endpoint} proves the finite-block identity in
Eq.~\eqref{eq:replacer-universal-additivity} for all eleven divergences.
Dividing by $n$ and taking the regularizations proves the second identity.

Finally, Theorem~\ref{thm:unified-operational-correspondence} gives the
finite-block operational correspondences, while Theorems
\ref{thm:general-general-Stein},
\ref{thm:iid-tester-general-jammer-Stein},
\ref{thm: public hypo hypothesis-unaware iid}, and
\ref{thm: secret hypo hypothesis-aware iid} express the  asymptotic Stein exponents through the
corresponding Umegaki channel divergences.  Since all of these divergences are
covered by Eq.~\eqref{eq:replacer-universal-additivity}, they give
Eq.~\eqref{eq:replacer-universal-Stein-exponent} for all games.
\end{proof}

\begin{boxremark}[Implication in quantum illumination]
Quantum illumination aims to determine whether a distant object is present by using quantum light together with a quantum detection strategy~\cite{Lloyd2008QuantumIllumination}. Cooney, Mosonyi, and Wilde studied a variant of the quantum-illumination scenario that can be modeled as a channel-discrimination problem between a general channel and a replacer channel~\cite{cooney2016strong}. The minimax framework developed here, together with Theorem~\ref{thm:replacer-universal-single-letter}, extends this perspective to a more versatile setting in which a jammer may also participate and adversarially influence the channel. This captures practical scenarios where the sensing or communication environment is uncertain, contested, or deliberately manipulated. A full application to quantum illumination would require extending the present analysis to bosonic modes under energy constraints, which we leave for future work.
\end{boxremark}

\section{Strong-converse upgrades}
\label{sec:strong-converse-enhancement}

Establishing an exact asymptotic rate in hypothesis testing typically requires
two complementary arguments. Achievability shows that a target exponent can be
attained by a feasible sequence of tests, whereas the converse shows that no
larger exponent is possible. A strong converse gives a sharper threshold: any
attempt to operate above the threshold forces the type-I error to converge to
one. 

In this section, we develop a general technique that upgrades suitable
achievability statements into strong converses.
Section~\ref{subsec:information-spectrum-enhancement} establishes the
mathematical upgrade for state sequences and then extends it to sequences of
state sets.  We apply these results to an IID tester against a secret entangled
jammer in Section~\ref{subsec:strong-converse-iid-tester-general-jammer}
and to composite state hypotheses in
Section~\ref{subsec:composite-state-strong-converse}, including the
trivial-tester game.  In particular, the composite IID specialization resolves the
strong-converse question posed in~\cite[Remark~2.2]{berta2021composite} and strengthens a few results studied recently in~\cite{Lam25_Sanov}.

\subsection{General principles}
\label{subsec:information-spectrum-enhancement}

We first establish the principle for sequences of classical distributions in
Lemma~\ref{lem:classical-info-spectrum-strong-conv-general-scale}. We then use
the Nussbaum--Szkola distributions in
Lemma~\ref{lem:NS-positive-alternative-hypothesis-comparison} to lift the result
to general quantum state sequences in
Theorem~\ref{thm:quantum-info-spectrum-strong-conv-general-scale}. Finally, we
extend the result to sets of quantum states in
Theorem~\ref{thm:composite-state-strong-converse}.

\subsubsection*{Upgrade for classical sequences}

Let $\mathbf c=(c_k)_{k\in\NN}$ be positive with $c_k\to\infty$.
The usual blocklength setting has $c_k=k$, while a blocklength subsequence
$(n_k)_k$ is covered by $c_k=n_k$.  For state sequences
$\brho=(\rho_k)_k$ and positive semidefinite alternatives
$\bsigma=(\sigma_k)_k$, the vanishing-error and strong-converse exponents
are, respectively,
\begin{align}
\lim_{\eta\to0^+}\liminf_{k\to\infty}
\frac1{c_k}D_{\Hypo,\eta}(\rho_k\|\sigma_k),
\qquad\text{and}\qquad
\lim_{\eta\to1^-}\limsup_{k\to\infty}
\frac1{c_k}D_{\Hypo,\eta}(\rho_k\|\sigma_k).
\label{eq:testing-relative-entropy-endpoint-rates}
\end{align}
Appendix~\ref{sec:operational-testing-exponents} proves their exact
equivalence to the operational definitions in~\cite{NH07}, including the
one-sided error limits needed at fixed error.  Thus a Stein lemma with a
strong converse at rate $d$ is expressed by
$\lim_{k\to\infty}c_k^{-1}D_{\Hypo,\ve}(\rho_k\|\sigma_k)=d$
for every fixed $\ve\in(0,1)$.

We retain the shorthand for the upper relative entropy rate,
\begin{align}
\overline D_{\mathbf c}(\brho\|\bsigma)
&:=\limsup_{k\to\infty}\frac1{c_k}D(\rho_k\|\sigma_k),
&
\overline D_{\mathbf c}(\bp\|\bq)
&:=\limsup_{k\to\infty}\frac1{c_k}D(p_k\|q_k).
\label{eq:general-scale-relative-entropy-definitions}
\end{align}
The second expression is its classical counterpart for probability
distributions $\bp=(p_k)_k$ and finite nonnegative measures $\bq=(q_k)_k$.

\begin{boxlemma}[Strong converse upgrade for classical sequences]
\label{lem:classical-info-spectrum-strong-conv-general-scale}
Let $p_k$ be a probability distribution and $q_k$ a finite nonnegative
measure on a finite set $\Omega_k$, and let $c_k>0$ with $c_k\to\infty$.
Assume that
\begin{align}\label{eq: mass growth assumption}
\limsup_{k\to\infty}\frac1{c_k}\log q_k(\Omega_k)<\infty.
\end{align}
If the vanishing error exponent achieves the relative entropy rate
\begin{align}
\lim_{\eta\to0^+}\liminf_{k\to\infty}\frac1{c_k}
D_{\Hypo,\eta}(p_k\|q_k)\ge\overline D_{\mathbf c}(\bp\|\bq),
\label{eq:classical-testing-achievability}
\end{align}
then
\begin{align}
\lim_{k\to\infty}\frac1{c_k}D_{\Hypo,\ve}(p_k\|q_k)
=\overline D_{\mathbf c}(\bp\|\bq),
\qquad \forall\,\ve\in(0,1).
\label{eq:general-scale-classical-conclusion}
\end{align}
\end{boxlemma}

\begin{proof}
Set $d:=\overline D_{\mathbf c}(\bp\|\bq)$.
The log-sum inequality gives $D(p_k\|q_k)\ge-\log q_k(\Omega_k)$, with
$-\log0=+\infty$, so the mass-growth assumption in Eq.~\eqref{eq: mass growth assumption} excludes $d=-\infty$.
Monotonicity in the error parameter gives the lower bound for every fixed
$\ve\in(0,1)$ and proves the claim when $d=+\infty$.  We may therefore
assume $d\in\RR$.  

Suppose that the upper bound direction of Eq.~\eqref{eq:general-scale-classical-conclusion} fails for some $\ve$. That is, there exists $\ve$ such that
\begin{align}
  \limsup_{k\to\infty}\frac1{c_k}D_{\Hypo,\ve}(p_k\|q_k)
>\overline D_{\mathbf c}(\bp\|\bq).
\end{align}  
Define
\begin{align} 
  L_k(a):=\{x:p_k(x)>2^{c_ka}q_k(x)\}.
\end{align} 
By monotonicity and the upper
information-spectrum identity in
Eq.~\eqref{eq:testing-spectrum-upper-identity}, there is $b>d$ such that
\begin{align}
r:=\limsup_{k\to\infty}p_k(L_k(b))>0.
\end{align}
Choose $0<\delta<r(b-d)$ and partition $\Omega_k$ into
\begin{align}
\Omega_{1,k}:=L_k(b),\qquad
\Omega_{2,k}:=L_k(d-\delta)\setminus L_k(b),\qquad
\Omega_{3,k}:=L_k(d-\delta)^c.
\end{align}
By Eq.~\eqref{eq:classical-testing-achievability} and the lower
information-spectrum identity in
Eq.~\eqref{eq:testing-spectrum-lower-identity},
$p_k(\Omega_{3,k})\to0$.  Thus asymptotically no probability mass lies below
$d-\delta$, while a positive mass lies above $b>d$.

We may discard a finite initial segment so that $D(p_k\|q_k)$ is finite.
The likelihood-ratio bounds on the first two regions and the log-sum
inequality on the third give
\begin{align}
\frac1{c_k}D(p_k\|q_k)
&\ge b\,p_k(\Omega_{1,k})+(d-\delta)p_k(\Omega_{2,k})
+\frac{p_k(\Omega_{3,k})}{c_k}
\log\frac{p_k(\Omega_{3,k})}{q_k(\Omega_k)}
\\
&=d-\delta+p_k(\Omega_{1,k})(b-d+\delta)
\\
&\quad+p_k(\Omega_{3,k})\left[
\frac1{c_k}\log p_k(\Omega_{3,k})
-\frac1{c_k}\log q_k(\Omega_k)-(d-\delta)
\right].
\label{eq:classical-three-region-bound}
\end{align}
Here we use $0\log0=0$.  The last term has limit inferior at least zero:
$x\log x$ is bounded below on $[0,1]$, the positive part of
$c_k^{-1}\log q_k(\Omega_k)$ is bounded, and $p_k(\Omega_{3,k})\to0$.

Choose a subsequence $(k_j)_j$ along which $p_{k_j}(\Omega_{1,k_j})\to r$.
Then Eq.~\eqref{eq:classical-three-region-bound} yields
\begin{align}
d
&\ge\liminf_{j\to\infty}\frac1{c_{k_j}}D(p_{k_j}\|q_{k_j})
\ge d-\delta+r(b-d+\delta)>d.
\end{align}
The last inequality follows from $\delta<r(b-d)$ and $r>0$.
This contradiction proves the upper bound and hence the claim.
\end{proof}

\subsubsection*{Upgrade for quantum state sequences}

We associate quantum pairs with classical finite measures through the
Nussbaum--Szko\l a construction~\cite{Nussbaum2009}.
For $\rho\in\density(\cH)$ and $\sigma\in\PSD(\cH)$, consider their spectral decompositions
\begin{align}
\rho=\sum_jr_j\proj{e_j},\qquad
\sigma=\sum_\ell s_\ell\proj{f_\ell}.
\end{align}
The eigenvalues are repeated according to multiplicity, including zeros.
On the finite set
$\Omega_{\rho,\sigma}:=\{(j,\ell)\}$, define the Nussbaum--Szko\l a
measures
\begin{align}
P_{\rho,\sigma}(j,\ell)
&:=r_j\abs{\langle e_j|f_\ell\rangle}^2,
&
Q_{\rho,\sigma}(j,\ell)
&:=s_\ell\abs{\langle e_j|f_\ell\rangle}^2.
\end{align}
Direct substitution gives
\begin{align}
P_{\rho,\sigma}(\Omega_{\rho,\sigma})=1,\qquad
Q_{\rho,\sigma}(\Omega_{\rho,\sigma})=\tr\sigma,\qquad
D(P_{\rho,\sigma}\|Q_{\rho,\sigma})=D(\rho\|\sigma).
\label{eq:NS-positive-alternative-identities}
\end{align}

The next lemma is the key step in the quantum
strong-converse upgrade.  It follows from the one-shot Nussbaum--Szko\l a
converse bound of~\cite[Theorem~1]{lizarribar2026converse}.

\begin{boxlemma}
\label{lem:NS-positive-alternative-hypothesis-comparison}
For every $\rho\in\density(\cH)$,
$\sigma\in\PSD(\cH)$, and every $0<\delta<\ve<1$,
\begin{align}
D_{\Hypo,\ve}(P_{\rho,\sigma}\|Q_{\rho,\sigma})
\ge
D_{\Hypo,\ve-\delta}(\rho\|\sigma)
+\log\frac{\delta}{\ve}.
\label{eq:NS-positive-alternative-hypothesis-bound}
\end{align}
\end{boxlemma}

\begin{proof}
Set $t:=\tr\sigma$, $\widehat\sigma:=\sigma/t$, and
$\widehat Q:=Q_{\rho,\sigma}/t$.
Then $t>0$, and $\rho$ and $\widehat\sigma$ are normalized states.
Using the same eigenbases, their Nussbaum--Szko\l a distributions are
$P_{\rho,\sigma}$ and $\widehat Q$.
After complementing the test effect to match our convention,
\cite[Theorem~1]{lizarribar2026converse} gives
\begin{align}
\beta_a(\rho\|\widehat\sigma)
&\ge
s\,
\beta_{a/(1-s)}(P_{\rho,\sigma}\|\widehat Q),
\qquad
0\le s<1,\quad 0\le a\le1-s.
\label{eq:NS-normalized-beta-comparison}
\end{align}
For $a:=\ve-\delta$ and $s:=\delta/\ve$, we have
$a/(1-s)=\ve$ and $a=\ve(1-s)\le1-s$.  Therefore,
\begin{align}
\beta_{\ve-\delta}(\rho\|\sigma)
&=
t\,\beta_{\ve-\delta}(\rho\|\widehat\sigma)
\ge
\frac{\delta}{\ve}\,
t\,\beta_\ve(P_{\rho,\sigma}\|\widehat Q)
=
\frac{\delta}{\ve}\,
\beta_\ve(P_{\rho,\sigma}\|Q_{\rho,\sigma}).
\label{eq:NS-unnormalized-beta-comparison}
\end{align}
Here the equalities use the linear scaling of $\beta_\eta$ in its
alternative, and the inequality follows from
Eq.~\eqref{eq:NS-normalized-beta-comparison}.
Taking negative logarithms proves the claim.
\end{proof}

The decisive feature of
Lemma~\ref{lem:NS-positive-alternative-hypothesis-comparison} is that its
correction depends only on the two error parameters.  After normalization by
$c_k$, this correction vanishes.  We can
therefore lift the classical result in
Lemma~\ref{lem:classical-info-spectrum-strong-conv-general-scale}. 

\begin{boxtheorem}[Strong converse upgrade for state sequences]
\label{thm:quantum-info-spectrum-strong-conv-general-scale}
Let $\rho_k\in\density(\cH_k)$ and $\sigma_k\in\PSD(\cH_k)$, where each
$\cH_k$ is finite-dimensional, and let $c_k>0$ with $c_k\to\infty$.
Write $\brho=(\rho_k)_k$, $\bsigma=(\sigma_k)_k$, and
$\mathbf c=(c_k)_k$.  Assume that
\begin{align}
\limsup_{k\to\infty}\frac1{c_k}\log\tr\sigma_k<\infty.
\label{eq:quantum-sequence-trace-growth}
\end{align}
If the vanishing-error exponent reaches the upper relative entropy rate,
\begin{align}
\lim_{\eta\to0^+}\liminf_{k\to\infty}\frac1{c_k}
D_{\Hypo,\eta}(\rho_k\|\sigma_k)
\ge\overline D_{\mathbf c}(\brho\|\bsigma),
\label{eq:quantum-sequence-testing-achievability}
\end{align}
then
\begin{align}
\lim_{k\to\infty}\frac1{c_k}D_{\Hypo,\ve}(\rho_k\|\sigma_k)
=\overline D_{\mathbf c}(\brho\|\bsigma),
\qquad \forall\,\ve\in(0,1).
\label{eq:general-scale-quantum-conclusion}
\end{align}
\end{boxtheorem}

\begin{proof}
Set $d:=\overline D_{\mathbf c}(\brho\|\bsigma)$.
The trace bound and $D(\rho_k\|\sigma_k)\ge-\log\tr\sigma_k$ imply
$d>-\infty$.  Monotonicity in the error parameter gives
$\liminf_k c_k^{-1}D_{\Hypo,\ve}(\rho_k\|\sigma_k)\ge d$
for every fixed $\ve\in(0,1)$.  This proves the claim if $d=+\infty$,
so assume $d\in\RR$.

Set $p_k:=P_{\rho_k,\sigma_k}$ and $q_k:=Q_{\rho_k,\sigma_k}$ on
$\Omega_k:=\Omega_{\rho_k,\sigma_k}$.
By Eq.~\eqref{eq:NS-positive-alternative-identities},
\begin{align}
q_k(\Omega_k)=\tr\sigma_k,\qquad
\overline D_{\mathbf c}(\bp\|\bq)
=\overline D_{\mathbf c}(\brho\|\bsigma)=d.
\label{eq:general-scale-NS-Dbar}
\end{align}
For every $\eta\in(0,1)$,
Lemma~\ref{lem:NS-positive-alternative-hypothesis-comparison} with
$\delta=\eta/2$ gives
\begin{align}
\liminf_{k\to\infty}\frac1{c_k}D_{\Hypo,\eta}(p_k\|q_k)
&\ge\liminf_{k\to\infty}\frac1{c_k}
D_{\Hypo,\eta/2}(\rho_k\|\sigma_k)
\ge d.
\label{eq:quantum-to-classical-testing-achievability}
\end{align}
The constant correction vanishes after division by $c_k$; the last
inequality follows from Eq.~\eqref{eq:quantum-sequence-testing-achievability}.
Thus Lemma~\ref{lem:classical-info-spectrum-strong-conv-general-scale}
applies to $(p_k,q_k)$.

Fix $\ve\in(0,1)$ and choose $\ve'\in(\ve,1)$.
Applying the same one-shot comparison with $\delta=\ve'-\ve$ yields
\begin{align}
d
&\le\liminf_{k\to\infty}\frac1{c_k}
D_{\Hypo,\ve}(\rho_k\|\sigma_k)
\le\limsup_{k\to\infty}\frac1{c_k}
D_{\Hypo,\ve}(\rho_k\|\sigma_k)\le\lim_{k\to\infty}\frac1{c_k}
D_{\Hypo,\ve'}(p_k\|q_k)=d.
\label{eq:general-scale-fixed-error-limsup}
\end{align}
Again the correction vanishes because $c_k\to\infty$, and the final
equality is the classical strong converse.  This proves the fixed-error
limit directly.
\end{proof}

\subsubsection*{Upgrade for sets of quantum states}

Theorem~\ref{thm:quantum-info-spectrum-strong-conv-general-scale} extends
to sequences of sets by selecting a relative-entropy-minimizing pair at
each index.
For nonempty sets $\cP_k\subseteq\density(\cH_k)$ and
$\cQ_k\subseteq\PSD(\cH_k)$, we use the pairwise set divergences
\begin{align}
\DD(\cP_k\|\cQ_k)
&:=\inf_{\substack{\rho_k\in\cP_k\\\sigma_k\in\cQ_k}}
\DD(\rho_k\|\sigma_k),
\qquad
\DD\in\{D,D_{\Hypo,\ve}\}.
\label{eq:composite-state-pairwise-divergences}
\end{align}
For the sequences $\cP=(\cP_k)_{k\in\NN}$ and
$\cQ=(\cQ_k)_{k\in\NN}$ and a positive scale
$\mathbf c=(c_k)_{k\in\NN}$ with $c_k\to\infty$, define
\begin{align}
\overline D_{\mathbf c}(\cP\|\cQ)
&:=\limsup_{k\to\infty}\frac1{c_k}D(\cP_k\|\cQ_k).
\label{eq:composite-state-relative-entropy-rate}
\end{align}
For the standard scale $c_k=k$, we write
$\overline D^{\infty}(\cP\|\cQ)$, or $D^{\infty}(\cP\|\cQ)$ when the
limsup is a limit.

\begin{boxtheorem}[Strong-converse upgrade for set sequences]
\label{thm:composite-state-strong-converse}
Let $\mathbf c=(c_k)_{k\in\NN}$ with $c_k>0$ and $c_k\to\infty$.
For every $k\in\NN$, let $\cH_k$ be finite-dimensional, and let
$\cP_k\subseteq\density(\cH_k)$ and $\cQ_k\subseteq\PSD(\cH_k)$
be nonempty compact sets.  Assume that
\begin{align}
\limsup_{k\to\infty}\frac1{c_k}
\log\sup_{\sigma_k\in\cQ_k}\tr\sigma_k
&<\infty.
\label{eq:set-sequence-alternative-trace-growth}
\end{align}
If the vanishing-error exponent reaches the upper relative entropy rate,
\begin{align}
\lim_{\delta\to0^+}\liminf_{k\to\infty}\frac1{c_k}
D_{\Hypo,\delta}(\cP_k\|\cQ_k)
&\ge\overline D_{\mathbf c}(\cP\|\cQ),
\label{eq:set-sequence-upgrade-assumption}
\end{align}
then
\begin{align}
\lim_{k\to\infty}\frac1{c_k}D_{\Hypo,\ve}(\cP_k\|\cQ_k)
&=\overline D_{\mathbf c}(\cP\|\cQ),
\qquad \forall\,\ve\in(0,1).
\label{eq:set-sequence-fixed-error-strong-converse}
\end{align}
\end{boxtheorem}

\begin{proof}
Set $d:=\overline D_{\mathbf c}(\cP\|\cQ)$.
The trace-growth assumption and $D(\rho\|\sigma)\ge-\log\tr\sigma$
exclude $d=-\infty$.  Monotonicity in the error parameter gives
\begin{align}
\liminf_{k\to\infty}\frac1{c_k}D_{\Hypo,\eta}(\cP_k\|\cQ_k)
&\ge
\lim_{\delta\to0^+}\liminf_{k\to\infty}\frac1{c_k}
D_{\Hypo,\delta}(\cP_k\|\cQ_k)
\ge d,
\qquad \forall\,\eta\in(0,1).
\label{eq:composite-state-fixed-error-achievability}
\end{align}
This proves the claim when $d=+\infty$, so assume $d\in\RR$.

By compactness and joint lower semicontinuity of the relative entropy,
choose $\rho_k\in\cP_k$ and $\sigma_k\in\cQ_k$ such that
\begin{align}
D(\rho_k\|\sigma_k)
&=D(\cP_k\|\cQ_k),
&
\limsup_{k\to\infty}\frac1{c_k}D(\rho_k\|\sigma_k)
&=d.
\label{eq:composite-state-selected-relative-entropy-pair}
\end{align}
This choice is independent of the testing error.  Since
$D_{\Hypo,\eta}(\cP_k\|\cQ_k)\le D_{\Hypo,\eta}(\rho_k\|\sigma_k)$,
Eq.~\eqref{eq:composite-state-fixed-error-achievability} implies
\begin{align}
\lim_{\eta\to0^+}\liminf_{k\to\infty}\frac1{c_k}
D_{\Hypo,\eta}(\rho_k\|\sigma_k)
&\ge d.
\label{eq:composite-state-selected-achievability}
\end{align}
The selected pair thus satisfies the achievability premise of
Theorem~\ref{thm:quantum-info-spectrum-strong-conv-general-scale} at scale
$\mathbf c$, and its trace bound follows from
Eq.~\eqref{eq:set-sequence-alternative-trace-growth}.  That theorem gives
\begin{align}
\lim_{k\to\infty}\frac1{c_k}D_{\Hypo,\ve}(\rho_k\|\sigma_k)
&=d,
\qquad \forall\,\ve\in(0,1).
\label{eq:composite-state-selected-fixed-error-converse}
\end{align}
Consequently,
\begin{align}
d
&\le\liminf_{k\to\infty}\frac1{c_k}D_{\Hypo,\ve}(\cP_k\|\cQ_k)
\le\limsup_{k\to\infty}\frac1{c_k}D_{\Hypo,\ve}(\cP_k\|\cQ_k)
\le\lim_{k\to\infty}\frac1{c_k}D_{\Hypo,\ve}(\rho_k\|\sigma_k)
=d,
\label{eq:composite-state-fixed-error-equality-chain}
\end{align}
which proves Eq.~\eqref{eq:set-sequence-fixed-error-strong-converse}.
\end{proof}

\subsection{Upgrade for the minimax game in Section~\ref{sec:iid-tester-unrestricted-jammer}}
\label{subsec:strong-converse-iid-tester-general-jammer}
\label{subsec:unrestricted-jammer-converse-detailed}

We now apply Theorem~\ref{thm:quantum-info-spectrum-strong-conv-general-scale}
to strengthen the two cases in
Theorem~\ref{thm:iid-tester-general-jammer-Stein} to strong-converse results. The
main challenge is to construct a sequence of states that is independent of the
vanishing-error parameter $\ve$ while still achieving the asymptotic relative-entropy
rate, so that Theorem~\ref{thm:quantum-info-spectrum-strong-conv-general-scale}
can be applied. The achievability result in Theorem~\ref{thm:iid-tester-general-jammer-Stein}
is insufficient because the optimal state sequence obtained from the hypothesis-testing
relative entropy may depend on the error parameter. We therefore establish a
more suitable achievability result, which enables the application of
Theorem~\ref{thm:quantum-info-spectrum-strong-conv-general-scale}.

\begin{boxtheorem}[IID tester
against secret entangled jammer]
\label{thm: secret hypo hypothesis-aware iid vs general}
\label{cor: secret hypo hypothesis-unaware iid vs general}
Let $\cN\in\CPTP(AE\!:\!B)$ and $\cM\in\CP(AE\!:\!B)$ satisfy
$D_{\max}^{\uparrow}(\cN\|\cM)<\infty$. Then for $(\star)\in
\bigl\{
(\emptyuparrow,\downarrow),
(\emptyuparrow,\downarrow\downarrow)
\bigr\}$,
\begin{align}
\lim_{n\to\infty}
\frac1n
\Stein_{\ve,n}^{(\star)}(\cN,\cM)
=
D^{\infty,\emptyuparrow,\downarrow\downarrow}
(\cN\|\cM),\qquad \forall \ve\in(0,1).
\end{align}
\end{boxtheorem}

\begin{proof}
Theorem~\ref{thm:unified-operational-correspondence} gives, for every $n$,
\begin{align}
\Stein_{\ve,n}^{\emptyuparrow,\downarrow}(\cN,\cM)
=\Stein_{\ve,n}^{\emptyuparrow,\downarrow\downarrow}(\cN,\cM)
=D_{\Hypo,\ve}^{\emptyuparrow,\downarrow\downarrow}(\cN^{\ox n}\|\cM^{\ox n}).
\end{align}
It suffices to prove 
\begin{align}
  L_\ve:= \limsup_{n\to\infty}\frac1n
  D_{\Hypo,\ve}^{\emptyuparrow,\downarrow\downarrow}
  (\cN^{\ox n}\|\cM^{\ox n})
  \le
  D^{\infty,\emptyuparrow,\downarrow\downarrow}(\cN\|\cM),\qquad \forall \ve\in(0,1).
\end{align}

\prooftag{State subsequence construction.}
Choose $\kappa$ and $\mu$ as in
Lemma~\ref{lem:uniform-block-bounds} and a similar analysis therein implies that $L_\ve$ is finite.  For every $n$, choose
$\rho_{n,\ve}\in\density(AR)$ such that
\begin{align}
D_{\Hypo,\ve}^{\downarrow\downarrow}
(\cN_{\rho_{n,\ve}}^{\ox n}\|\cM_{\rho_{n,\ve}}^{\ox n})
\ge
D_{\Hypo,\ve}^{\emptyuparrow,\downarrow\downarrow}
(\cN^{\ox n}\|\cM^{\ox n})-1.
\end{align}

Choose an increasing sequence $n_k\to\infty$ realizing $L_\ve$.
After passing to a further subsequence, compactness of $\density(AR)$ gives
\begin{align}
\rho_{n_k,\ve}\longrightarrow\rho_0,
\qquad \text{and} \qquad
\frac1{n_k}
D_{\Hypo,\ve}^{\downarrow\downarrow}
(\cN_{\rho_{n_k,\ve}}^{\ox n_k}
 \|\cM_{\rho_{n_k,\ve}}^{\ox n_k})
\longrightarrow L_\ve.
\label{eq:strong-converse-selected-tester-subsequence}
\end{align}
This subsequence is now fixed; in particular, it will not depend on the
auxiliary error introduced next.

For every $k$, compactness and lower semicontinuity provide jammer states
$\sigma_{n_k},\omega_{n_k}$ such that
\begin{align}
D^{\downarrow\downarrow}
(\cN_{\rho_{n_k,\ve}}^{\ox n_k}
 \|\cM_{\rho_{n_k,\ve}}^{\ox n_k}) = 
D
(\cN_{\rho_{n_k,\ve}}^{\ox n_k}(\sigma_{n_k})
 \|\cM_{\rho_{n_k,\ve}}^{\ox n_k}(\omega_{n_k})).
\end{align}
Define two sequences of states
\begin{align}
\widehat\rho_k
&:=
\cN^{\ox n_k}
(\rho_{n_k,\ve}^{\ox n_k}\ox\sigma_{n_k}),
\qquad \text{and} \qquad
\widehat\sigma_k
:=
\cM^{\ox n_k}
(\rho_{n_k,\ve}^{\ox n_k}\ox\omega_{n_k}).
\end{align}
We aim to show that for every $\eta\in(0,1)$,
\begin{align}
\liminf_{k\to\infty}\frac1{n_k}
D_{\Hypo,\eta}(\widehat\rho_k\|\widehat\sigma_k)
\ge
\limsup_{k\to\infty}\frac1{n_k}
D(\widehat\rho_k\|\widehat\sigma_k).
\label{eq:strong-converse-output-information-spectrum-premise}
\end{align}
By their definitions, we have
\begin{align}
D(\widehat\rho_k\|\widehat\sigma_k)
&=
D^{\downarrow\downarrow}
(\cN_{\rho_{n_k,\ve}}^{\ox n_k}
 \|\cM_{\rho_{n_k,\ve}}^{\ox n_k}),
\\
D_{\Hypo,\eta}(\widehat\rho_k\|\widehat\sigma_k)
&\ge
D_{\Hypo,\eta}^{\downarrow\downarrow}
(\cN_{\rho_{n_k,\ve}}^{\ox n_k}
 \|\cM_{\rho_{n_k,\ve}}^{\ox n_k}).
\end{align}
So it suffices to show that
\begin{align}
&\liminf_{k\to\infty}\frac1{n_k}
D_{\Hypo,\eta}^{\downarrow\downarrow}
(\cN_{\rho_{n_k,\ve}}^{\ox n_k}
 \|\cM_{\rho_{n_k,\ve}}^{\ox n_k})
\ge
\limsup_{k\to\infty}\frac1{n_k}
D^{\downarrow\downarrow}
(\cN_{\rho_{n_k,\ve}}^{\ox n_k}
 \|\cM_{\rho_{n_k,\ve}}^{\ox n_k}).
\end{align}

\prooftag{Achievability for the fixed tester sequence.} 
Fix $\eta\in(0,1)$, $m\in\NN$, and $\alpha\in(1/2,1)$.  Write
$n_k=q_km+r_k$, where $0\le r_k<m$, and set
\begin{align}
\cE_k:=\cN_{\rho_{n_k,\ve}},\qquad
\cF_k:=\cM_{\rho_{n_k,\ve}},\qquad
c_{\alpha,\eta}:=\frac{\alpha}{\alpha-1}\log\frac1\eta.
\end{align}
We interpret the divergence of the remainder block as zero when $r_k=0$.
Then
\begin{align}
&\liminf_{k\to\infty}\frac1{n_k}
D_{\Hypo,\eta}^{\downarrow\downarrow}
(\cE_k^{\ox n_k}\|\cF_k^{\ox n_k})
\\
&\ge
\liminf_{k\to\infty}
\left[
\frac1{n_k}D_{\Meas,\alpha}^{\downarrow\downarrow}
(\cE_k^{\ox n_k}\|\cF_k^{\ox n_k})
+\frac{c_{\alpha,\eta}}{n_k}
\right]
\\
&\ge
\liminf_{k\to\infty}
\left[
\frac{q_k}{n_k}D_{\Meas,\alpha}^{\downarrow\downarrow}
(\cE_k^{\ox m}\|\cF_k^{\ox m})
+\frac1{n_k}D_{\Meas,\alpha}^{\downarrow\downarrow}
(\cE_k^{\ox r_k}\|\cF_k^{\ox r_k})
+\frac{c_{\alpha,\eta}}{n_k}
\right]
\\
&\ge
\liminf_{k\to\infty}
\left[
\frac{q_k}{n_k}D_{\Meas,\alpha}^{\downarrow\downarrow}
(\cE_k^{\ox m}\|\cF_k^{\ox m})
-\frac{r_k}{n_k}\log\mu
+\frac{c_{\alpha,\eta}}{n_k}
\right]
\\
&\ge
\frac1mD_{\Meas,\alpha}^{\downarrow\downarrow}
(\cN_{\rho_0}^{\ox m}\|\cM_{\rho_0}^{\ox m}).
\label{eq:strong-converse-common-rate-comparison}
\end{align}
The first inequality follows from the one-shot lower bound in
Lemma~\ref{lem:one-shot-hypothesis-testing-bounds} and
$D_{\Meas,\alpha}\le D_{\Sand,\alpha}$.  The second follows by iterating the
measured-R\'enyi chain rule~\cite[Lemma~2]{fang2025adversarial}; this is done
pointwise in $k$, so $(\cE_k,\cF_k)$ is fixed throughout each application.
The third inequality uses the uniform lower bound in
Eq.~\eqref{eq:fixed-block-uniform-finite-bounds} for the remainder block.
For the last inequality, Proposition~18(iii) of~\cite{mosonyi2023some} gives
joint lower semicontinuity of $D_{\Meas,\alpha}^{\downarrow\downarrow}$.  The induced outputs depend
continuously on $(\rho,\sigma,\omega)$; the first is a state, and the second
is nonzero by the completely positive domination established in the proof of
Lemma~\ref{lem:uniform-block-bounds}.  Since
$\density(E^m)\times\density(E^m)$ is fixed and compact, the minimum is the
negative of the supremum of the negative objective.  Lemma~3(ii) therein thus
shows that
$\rho\mapsto D_{\Meas,\alpha}^{\downarrow\downarrow}
(\cN_\rho^{\ox m}\|\cM_\rho^{\ox m})$ is lower semicontinuous.  Combining
this with
$q_k/n_k\to1/m$, $r_k/n_k\to0$, $c_{\alpha,\eta}/n_k\to0$, and
$\rho_{n_k,\ve}\to\rho_0$ proves the last inequality above.

Taking the supremum over
$\alpha\in(1/2,1)$ and applying the same measured-to-Umegaki comparison as in
the proof of Eq.~\eqref{eq:detailed-full-jammer-achievability}, we have
\begin{align}
&\liminf_{k\to\infty}\frac1{n_k}
D_{\Hypo,\eta}^{\downarrow\downarrow}
(\cN_{\rho_{n_k,\ve}}^{\ox n_k}
 \|\cM_{\rho_{n_k,\ve}}^{\ox n_k})
\ge
D^{\infty,\downarrow\downarrow}
(\cN_{\rho_0}\|\cM_{\rho_0}).\label{eq: DH rho0}
\end{align}

For the converse comparison, fix $\ell\in\NN$ and block
$n_k=a_k\ell+s_k$, where $0\le s_k<\ell$.  We again interpret the remainder
divergence as zero when $s_k=0$.  Then
\begin{align}
&\limsup_{k\to\infty}\frac1{n_k}
D^{\downarrow\downarrow}
(\cE_k^{\ox n_k}\|\cF_k^{\ox n_k})\notag
\\
&\le
\limsup_{k\to\infty}
\left[
\frac{a_k}{n_k}D^{\downarrow\downarrow}
(\cE_k^{\ox\ell}\|\cF_k^{\ox\ell})
+\frac1{n_k}D^{\downarrow\downarrow}
(\cE_k^{\ox s_k}\|\cF_k^{\ox s_k})
\right]
\\
&\le
\limsup_{k\to\infty}
\left[
\frac{a_k}{n_k}D^{\downarrow\downarrow}
(\cE_k^{\ox\ell}\|\cF_k^{\ox\ell})
+\frac{s_k}{n_k}\log\kappa
\right]
\\
&=
\frac1\ell D^{\downarrow\downarrow}
(\cN_{\rho_0}^{\ox \ell}\|\cM_{\rho_0}^{\ox \ell}).
\end{align}
The first inequality follows by iterating the fixed-tester subadditivity in
Eq.~\eqref{eq:fixed-tester-subadditivity-general-jammer} for the pair
$(\cE_k,\cF_k)$ at each fixed $k$.  The second uses the upper bound in
Eq.~\eqref{eq:fixed-block-uniform-finite-bounds} on the remainder block.
Finally, the equality follows from $a_k/n_k\to1/\ell$,
$s_k/n_k\to0$, $\rho_{n_k,\ve}\to\rho_0$, and the fixed-blocklength
continuity in Lemma~\ref{lem:fixed-block-continuity-general-jammer}.
Taking the infimum over $\ell$ and applying Fekete's lemma yields
\begin{align}
\limsup_{k\to\infty}\frac1{n_k}
D^{\downarrow\downarrow}
(\cN_{\rho_{n_k,\ve}}^{\ox n_k}
 \|\cM_{\rho_{n_k,\ve}}^{\ox n_k})
\le D^{\infty,\downarrow\downarrow}(\cN_{\rho_0}\|\cM_{\rho_0}).\label{eq:strong-converse-selected-tester-subsequence 1}
\end{align}

Combining Eqs.~\eqref{eq: DH rho0} and~\eqref{eq:strong-converse-selected-tester-subsequence 1}
gives
\begin{align}
&\liminf_{k\to\infty}\frac1{n_k}
D_{\Hypo,\eta}^{\downarrow\downarrow}
(\cN_{\rho_{n_k,\ve}}^{\ox n_k}
 \|\cM_{\rho_{n_k,\ve}}^{\ox n_k})
\ge
\limsup_{k\to\infty}\frac1{n_k}
D^{\downarrow\downarrow}
(\cN_{\rho_{n_k,\ve}}^{\ox n_k}
 \|\cM_{\rho_{n_k,\ve}}^{\ox n_k}).
\label{eq:strong-converse-common-sequence-comparison}
\end{align}

\prooftag{Strong converse upgrade.}
Use the scale $c_k=n_k$ and write
$\widehat{\brho}:=(\widehat\rho_k)_k$ and
$\widehat{\bsigma}:=(\widehat\sigma_k)_k$.  Here $c_k\to\infty$,
$\widehat\rho_k$ is a state because $\cN$ is trace preserving, and
$\widehat\sigma_k$ is positive semidefinite because $\cM$ is completely
positive.  The divergence bounds in
Lemma~\ref{lem:uniform-block-bounds} and the trace estimate in its proof
give
\begin{align}
-n_k\log\mu
&\le D(\widehat\rho_k\|\widehat\sigma_k)
\le n_k\log\kappa,
\\
\kappa^{-n_k}
&\le\tr\widehat\sigma_k\le\mu^{n_k}.
\end{align}
Hence
$\overline D_{\mathbf c}
(\widehat{\brho}\|\widehat{\bsigma})$ is finite, and
$\limsup_k n_k^{-1}\log\tr\widehat\sigma_k<\infty$.
Equation~\eqref{eq:strong-converse-output-information-spectrum-premise}
verifies the achievability premise of
Theorem~\ref{thm:quantum-info-spectrum-strong-conv-general-scale}.
Applying that theorem at the original error $\ve$ gives
\begin{align}
\lim_{k\to\infty}\frac1{n_k}
D_{\Hypo,\ve}(\widehat\rho_k\|\widehat\sigma_k)
=
\overline D_{\mathbf c}
(\widehat{\brho}\|\widehat{\bsigma}). \label{eq: DH Dc tmp}
\end{align}

Combining the preceding estimates, we obtain
\begin{align}
L_\ve
&\le
\lim_{k\to\infty}\frac1{n_k}
D_{\Hypo,\ve}(\widehat\rho_k\|\widehat\sigma_k)
=
\overline D_{\mathbf c}
(\widehat{\brho}\|\widehat{\bsigma})
=
\limsup_{k\to\infty}\frac1{n_k}
D(\widehat\rho_k\|\widehat\sigma_k)\notag
\\
&=
\limsup_{k\to\infty}\frac1{n_k}
D^{\downarrow\downarrow}
(\cN_{\rho_{n_k,\ve}}^{\ox n_k}
 \|\cM_{\rho_{n_k,\ve}}^{\ox n_k})
\le
D^{\infty,\downarrow\downarrow}
(\cN_{\rho_0}\|\cM_{\rho_0})
\le
D^{\infty,\emptyuparrow,\downarrow\downarrow}(\cN\|\cM).
\label{eq:strong-converse-final-rate-chain}
\end{align}
The first inequality follows from
Eq.~\eqref{eq:strong-converse-selected-tester-subsequence}, because the
chosen jammer states are feasible for the jammer-input hypothesis-testing
divergence and hence
$D_{\Hypo,\ve}(\widehat\rho_k\|\widehat\sigma_k)$ is no smaller than that
jammer-input value.  The first equality is
Eq.~\eqref{eq: DH Dc tmp}; the next two use the definition of
$\overline D_{\mathbf c}$ and the choice of
$(\sigma_{n_k},\omega_{n_k})$ as a minimizing jammer pair.
The next inequality follows from
Eq.~\eqref{eq:strong-converse-selected-tester-subsequence 1}.
Finally, the last inequality holds blockwise because the tester
supremum includes the fixed state $\rho_0$; regularization preserves this
inequality.  This completes the proof.
\end{proof}

\subsection{Upgrade for composite state hypotheses}
\label{subsec:composite-state-strong-converse}

We now apply the set-sequence upgrade in
Theorem~\ref{thm:composite-state-strong-converse} to composite hypothesis
testing. This yields strong converse upgrades of results due to
Berta, Brand\~ao, and Hirche~\cite{berta2021composite} and
Lami~\cite{Lam25_Sanov}, and recovers the trivial-tester channel discrimination game as
a specialization.

For a sequence of nonempty compact sets in finite-dimensional spaces,
$\cC=(\cC_n)_{n\in\NN}$, write
\begin{align}
\conv(\cC)
&:=\bigl(\conv(\cC_n)\bigr)_{n\in\NN}.
\label{eq:composite-state-convexified-sequence}
\end{align}
For every $n$, let $\cH_n$ be finite-dimensional and let
$\cP_n\subseteq\density(\cH_n)$ and $\cQ_n\subseteq\PSD(\cH_n)$
be nonempty compact sets.  Write $\cP=(\cP_n)_{n\in\NN}$ and
$\cQ=(\cQ_n)_{n\in\NN}$.

For $\ve\in(0,1)$, define the composite hypothesis-testing exponent by
\begin{align}
\Stein_{\ve,n}(\cP,\cQ)
&:=-\log
\inf_{\substack{T_n\in\sE(\cH_n)\\
\sup_{\rho_n\in\cP_n}\tr[\rho_n(I-T_n)]\le\ve}}
\ \sup_{\sigma_n\in\cQ_n}\tr[\sigma_nT_n].
\label{eq:composite-state-fixed-error-Stein-exponent}
\end{align}
Taking convex hulls leaves both worst-case testing functionals
unchanged.  The composite hypothesis-testing identity
of~\cite[Lemma 31]{fang2024generalized} gives
\begin{align}
\Stein_{\ve,n}(\cP,\cQ)
&=D_{\Hypo,\ve}\bigl(\conv(\cP_n)\|\conv(\cQ_n)\bigr).
\label{eq:composite-state-testing-identity}
\end{align}
Thus the operational exponent is the pairwise divergence of the convex
hulls. The following is a direct consequence of Theorem~\ref{thm:composite-state-strong-converse}.

\begin{boxcorollary}[Strong-converse upgrade for composite hypotheses]
\label{cor:composite-state-operational-upgrade}
For every $n\in\NN$, let $\cH_n$ be finite-dimensional and let
$\cP_n\subseteq\density(\cH_n)$ and $\cQ_n\subseteq\PSD(\cH_n)$
be nonempty compact sets.  Write $\cP=(\cP_n)_{n\in\NN}$ and
$\cQ=(\cQ_n)_{n\in\NN}$.  Assume that
\begin{align}
\limsup_{n\to\infty}\frac1n
\log\sup_{\sigma_n\in \cQ_n}\tr\sigma_n
&<\infty.
\label{eq:composite-state-alternative-trace-growth}
\end{align}
If the vanishing-error exponent reaches the upper relative entropy rate,
\begin{align}
\lim_{\delta\to0^+}\liminf_{n\to\infty}\frac1n
\Stein_{\delta,n}(\cP,\cQ)
&\ge\overline D^{\infty}(\conv(\cP)\|\conv(\cQ)),
\label{eq:composite-state-upgrade-assumption}
\end{align}
then
\begin{align}
\lim_{n\to\infty}\frac1n
\Stein_{\ve,n}(\cP,\cQ)
&=\overline D^{\infty}(\conv(\cP)\|\conv(\cQ)),
\qquad \forall\,\ve\in(0,1).
\label{eq:composite-state-fixed-error-strong-converse}
\end{align}
\end{boxcorollary}

The applications below use this corollary with different choices of
composite hypotheses.

\subsubsection*{Composite IID hypotheses and the trivial-tester specialization}

Let $\cH$ be finite-dimensional.  For a nonempty closed set
$\cC\subseteq\density(\cH)$, define its composite IID and arbitrarily varying
families by
\begin{align}
\cC_n^{\mathsf{IID}}
&:=\bigl\{\rho^{\ox n}:\rho\in\cC\bigr\},
&
\cC_n^{\mathsf{AV}}
&:=\bigl\{\rho_1\ox\cdots\ox\rho_n:
\rho_1,\ldots,\rho_n\in\cC\bigr\}.
\label{eq:composite-IID-AV-families}
\end{align}
We write $\cC^b:=(\cC_n^b)_{n\in\NN}$ for
$b\in\{\mathsf{IID},\mathsf{AV}\}$.
Thus $\conv(\cC^b)$ denotes the sequence of convex hulls of the
$n$-copy sets $\cC_n^b$, rather than the family $(\conv(\cC))^b$.

We first obtain a fixed-error strengthening of the composite IID Stein lemma
in~\cite{berta2021composite}.

\begin{boxcorollary}[Composite IID strong converse]
\label{cor:BBH-composite-state-strong-converse}
Let $\cS,\cT\subseteq\density(\cH)$ be nonempty closed convex sets.
Then the relative entropy rate is an ordinary limit and satisfies
\begin{align}
D^\infty\bigl(\conv(\cS^{\mathsf{IID}})\|
\conv(\cT^{\mathsf{IID}})\bigr)
&=\lim_{n\to\infty}\frac1n
\inf_{\substack{\rho\in\cS\\
\sigma_n\in\conv(\cT_n^{\mathsf{IID}})}}
D(\rho^{\ox n}\|\sigma_n).
\label{eq:BBH21-rate-as-composite-special-case}
\end{align}
For every fixed $\ve\in(0,1)$,
\begin{align}
\lim_{n\to\infty}\frac1n
\Stein_{\ve,n}(\cS^{\mathsf{IID}},\cT^{\mathsf{IID}})
&=D^\infty\bigl(\conv(\cS^{\mathsf{IID}})\|
\conv(\cT^{\mathsf{IID}})\bigr).
\label{eq:BBH21-composite-state-fixed-error-strong-converse}
\end{align}
\end{boxcorollary}

\begin{proof}
The achievability bound in~\cite[Lemmas~2.3--2.4 and Eq.~(49)]{berta2021composite}
gives Eq.~\eqref{eq:composite-state-upgrade-assumption} with
$\cP=\cS^{\mathsf{IID}}$ and $\cQ=\cT^{\mathsf{IID}}$.
Proposition~2.1 and Lemma~2.5 of the same reference, combined as in
Eqs.~(56)--(59) therein, show that the relative entropy rate is an ordinary
limit and establish Eq.~\eqref{eq:BBH21-rate-as-composite-special-case}.
In particular, Lemma~2.5 removes convexification of the IID null states
at a cost of order $\log n/n$ in the normalized divergence.

These steps do not require the pairwise support assumption imposed in
\cite[Theorem~1.1]{berta2021composite}.
The classical achievability step in Lemma~2.3 and the minimax identity in
Lemma~A.2 accommodate infinite relative entropies; Lemma~2.4 explicitly
covers pairs without support inclusion.  The weak-converse and
convexification bounds also remain valid with infinite values.

Since all alternative hypotheses are normalized, the trace-growth condition
is automatic.  Corollary~\ref{cor:composite-state-operational-upgrade}
therefore gives Eq.~\eqref{eq:BBH21-composite-state-fixed-error-strong-converse},
resolving the strong-converse question posed in
\cite[Remark~2.2]{berta2021composite}.
\end{proof}

\begin{boxcorollary}[Trivial tester against secret IID jammer]
\label{cor:trivial-tester-IID-jammer-strong-converse}
Let $\cN\in\CPTP(E\!:\!B)$ and $\cM\in\CP(E\!:\!B)$ satisfy
$D_{\max}^{\uparrow}(\cN\|\cM)<\infty$.  Then, for
$(\star)\in\bigl\{\emptydownarrow,
\emptydownarrow\emptydownarrow\bigr\}$ and every $\ve\in(0,1)$,
\begin{align}
\lim_{n\to\infty}\frac1n\Stein_{\ve,n}^{(\star)}(\cN,\cM)
&=D^{\infty,\convdownarrow\convdownarrow}(\cN\|\cM).
\label{eq:trivial-tester-IID-jammer-strong-converse}
\end{align}
\end{boxcorollary}

\begin{proof}
For every $n$, define the channel-generated composite hypotheses
\begin{align}
\cP_n
&:=\cN^{\ox n}\bigl(\sI(E^n)\bigr),
&
\cQ_n
&:=\cM^{\ox n}\bigl(\sI(E^n)\bigr).
\label{eq:trivial-tester-channel-generated-hypotheses}
\end{align}
The null set consists of states, the alternative set consists of positive
semidefinite operators, and both sets are compact.  By linearity,
\begin{align}
\conv(\cP_n)
&=\cN^{\ox n}\bigl(\conv(\sI(E^n))\bigr),
&
\conv(\cQ_n)
&=\cM^{\ox n}\bigl(\conv(\sI(E^n))\bigr).
\label{eq:trivial-tester-channel-generated-convex-hulls}
\end{align}
Let $\mu:=\|\cM^\dagger(I_B)\|_\infty$.  The finiteness assumption implies
$\mu>0$, and every $\sigma_n\in\conv(\cQ_n)$ satisfies
\begin{align}
\tr\sigma_n&\le\mu^n.
\end{align}
Thus Eq.~\eqref{eq:composite-state-alternative-trace-growth} holds.

Since the tester input is trivial,
Theorem~\ref{thm:unified-operational-correspondence} gives, for either choice
of $(\star)$,
\begin{align}
\Stein_{\ve,n}^{(\star)}(\cN,\cM)
&=D_{\Hypo,\ve}^{\convdownarrow\convdownarrow}
(\cN^{\ox n}\|\cM^{\ox n})
=\Stein_{\ve,n}(\cP,\cQ).
\label{eq:trivial-tester-composite-operational-identity}
\end{align}
Moreover, the limit-existence argument in the proof of
Theorem~\ref{thm: secret hypo hypothesis-aware iid} gives
\begin{align}
\overline D^{\infty}(\conv(\cP)\|\conv(\cQ))
&=D^{\infty}(\conv(\cP)\|\conv(\cQ))
=D^{\infty,\convdownarrow\convdownarrow}(\cN\|\cM).
\label{eq:trivial-tester-composite-relative-entropy-rate}
\end{align}
The vanishing-error result in that theorem, specialized to a trivial tester,
verifies Eq.~\eqref{eq:composite-state-upgrade-assumption}.  The operational
upgrade in Eq.~\eqref{eq:composite-state-fixed-error-strong-converse}, together
with Eq.~\eqref{eq:trivial-tester-composite-operational-identity}, proves the claim.
\end{proof}

\subsubsection*{Correlated and arbitrarily varying hypotheses}

The same parent theorem also upgrades the composite Stein exponents identified
by Lami~\cite{Lam25_Sanov}. 

\begin{boxcorollary}[Correlated and IID/AV composite hypotheses]
\label{cor:Lami-composite-state-strong-converse}
Let $\cB\subseteq\density(\cH)$ be nonempty and closed, and let
$\cA=(\cA_n)_{n\in\NN}$ be a sequence of nonempty closed convex sets
$\cA_n\subseteq\density(\cH^{\ox n})$ satisfying Axioms Q.I--Q.IV
of~\cite[Section~2]{Lam25_Sanov}.  Then, for every
$b\in\{\mathsf{IID},\mathsf{AV}\}$ and every fixed $\ve\in(0,1)$,
\begin{align}
\lim_{n\to\infty}\frac1n
\Stein_{\ve,n}(\cA,\cB^b)
&=D^{\infty}(\cA\|\conv(\cB^b)).
\label{eq:Lami-correlated-fixed-error-strong-converse}
\end{align}
Moreover, for any nonempty closed sets
$\cS,\cT\subseteq\density(\cH)$, any
$a,b\in\{\mathsf{IID},\mathsf{AV}\}$, and every fixed $\ve\in(0,1)$,
\begin{align}
\lim_{n\to\infty}\frac1n
\Stein_{\ve,n}(\cS^a,\cT^b)
&=D^{\infty}(\conv(\cS^a)\|\conv(\cT^b)).
\label{eq:composite-state-IID-AV-fixed-error-strong-converse}
\end{align}
\end{boxcorollary}

\begin{proof}
By~\cite[Eq.~(87)]{Lam25_Sanov}, the vanishing-error exponent equals the
relative entropy rate in
Eq.~\eqref{eq:Lami-correlated-fixed-error-strong-converse}.
For an IID null hypothesis, \cite[Eq.~(116)]{Lam25_Sanov} gives the
corresponding identity in
Eq.~\eqref{eq:composite-state-IID-AV-fixed-error-strong-converse}.  For an
arbitrarily varying null hypothesis, the discussion following
Eq.~(111) in~\cite[Section~4.4]{Lam25_Sanov} verifies Axioms Q.I--Q.IV for
$\conv(\cS^{\mathsf{AV}})$, so Eq.~(87) therein applies.
The alternative hypotheses are normalized, so
the trace-growth condition is automatic.  In every case,
Eq.~\eqref{eq:composite-state-fixed-error-strong-converse} gives the asserted
fixed-error limit.
\end{proof}

\section{Conclusion}
\label{sec:conclusion}

We have developed a game-theoretic framework for binary quantum channel
discrimination in which a tester and a jammer independently control two input
systems of the channel.  By combining three input structures with four
information patterns, the framework yields twelve operational games
that interpolate between tester-input and jammer-input discrimination.  We
characterized the asymptotic Stein exponent for all twelve games through maximin and minimax channel divergences.  The results show
that an entangled jammer erases the asymptotic effect of its visibility and
hypothesis-awareness, whereas an IID jammer can make the information
pattern operationally different. For replacer alternatives, all twelve vanishing-error
Stein exponents collapse to the same additive, single-letter value. We further develop a general argument that upgrades achievability results
to strong converse results, thereby establishing strong converse properties for several models,
resolving an open problem in composite hypothesis testing posed in~\cite[Remark~2.2]{berta2021composite} and strengthening a few results studied recently in~\cite{Lam25_Sanov}. 

Several questions remain open.  The present work treats discrimination with parallel channel uses, so a natural next step is to
extend the framework to adaptive or sequential strategies.  It would also be
important to determine when regularized divergences admit single-letter
forms beyond replacer alternatives, and whether the remaining vanishing-error
results can be strengthened to fixed-error or strong-converse statements. We also leave the application in quantum illumination as a promising further study.

\bigskip
\paragraph{Acknowledgements.}
K.F. is supported in part by the National Natural Science Foundation of China (Grants No. 92470113 and 12404569), the Shenzhen Science and Technology Program (Grants No. QNXMB20250701091826036 and JCYJ20240813113519025), the Shenzhen International Quantum Academy (Grant No. SIQA2025KFKT03), the Shenzhen Fundamental Research Program (Grant No. JCYJ20241202124023031), the Guangdong Provincial Quantum Science Strategic Initiative (Grant No. GDZX2503001), and the University Development Fund (Grant No. UDF01003565). 
M.C.~acknowledges support from the European Research Council (ERC Grant Agreement No.~948139), the Excellence Cluster Matter and Light for Quantum Computing (ML4Q), and the New Faculty Start-up Fund at City University of Hong Kong (Dongguan).
H.C.~acknowledges support from National Science and Technology Council (NSTC 115-2628-E-002-005, NSTC 114-2119-M-001-002, and NSTC 115-2124-M-002-014) and Ministry of Education (NTU-115V2016-1, NTU-CC115L893705, and NTU-115L900702). 
L.G. is
partially supported by the National Natural Science Foundation of China (Grant No. 12401163), by the Hubei
Provincial International Collaboration (Project No. 2025EHA041), and Hubei Provincial Innovation research
group (Project No. 2025AFA044).
M.H. was supported in part by the Guangdong Provincial Quantum Science
Strategic Initiative (Grant No. GDZX2505003), the General R\&D Projects of 1+1+1 CUHK-CUHK(SZ)-GDST Joint Collaboration Fund (Grant No. GRDP2025-022), and the Shenzhen International Quantum Academy (Grant
No. SIQA2025KFKT07).

We acknowledge that OpenAI GPT-5.6 Sol was used to explore some candidate proof strategies and support the writing of the manuscript. We
have verified AI-assisted material to the best of our knowledge
and take full responsibility for the content of this work.

\paragraph{Data availability statement.} There is no conflict of interest in this work.
No datasets were generated or analysed during the current study.

\bibliographystyle{alpha_abbrv}
\bibliography{Bib}

\newcommand{\etalchar}[1]{$^{#1}$}
\begin{thebibliography}{MLDS{\etalchar{+}}13}

\bibitem[Aci01]{acin2001statistical}
A.~Acin.
\newblock Statistical distinguishability between unitary operations.
\newblock {\em Physical Review Letters}, 87(17):177901, 2001.

\bibitem[BBH21]{berta2021composite}
M.~Berta, F.~G. Brandão, and C.~Hirche.
\newblock On composite quantum hypothesis testing.
\newblock {\em Communications in Mathematical Physics}, 385(1):55--77, 2021.

\bibitem[BCP19]{bae_more_2019}
J.~Bae, D.~Chruściński, and M.~Piani.
\newblock More {Entanglement} {Implies} {Higher} {Performance} in {Channel} {Discrimination} {Tasks}.
\newblock {\em Physical Review Letters}, 122(14):140404, April 2019.

\bibitem[BDNW18]{boche2018fully}
H.~Boche, C.~Deppe, J.~N{\"o}tzel, and A.~Winter.
\newblock Fully quantum arbitrarily varying channels: Random coding capacity and capacity dichotomy.
\newblock In {\em 2018 IEEE International Symposium on Information Theory (ISIT)}, pages 2012--2016. IEEE, 2018.

\bibitem[BDSW24]{bergh_parallelization_2024}
B.~Bergh, N.~Datta, R.~Salzmann, and M.~M. Wilde.
\newblock Parallelization of {Adaptive} {Quantum} {Channel} {Discrimination} in the {Non}-{Asymptotic} {Regime}.
\newblock {\em IEEE Transactions on Information Theory}, 70(4):2617--2636, April 2024.

\bibitem[Bel24]{belzig2024fully}
P.~Belzig.
\newblock Fully quantum arbitrarily varying channel coding for entanglement-assisted communication.
\newblock In {\em 2024 IEEE International Symposium on Information Theory (ISIT)}, pages 1011--1016. IEEE, 2024.

\bibitem[BFT17]{Berta2017}
M.~Berta, O.~Fawzi, and M.~Tomamichel.
\newblock On variational expressions for quantum relative entropies.
\newblock {\em Letters in Mathematical Physics}, 107(12):2239--2265, 2017.

\bibitem[BMQ21]{bavaresco2021strict}
J.~Bavaresco, M.~Murao, and M.~T. Quintino.
\newblock Strict hierarchy between parallel, sequential, and indefinite-causal-order strategies for channel discrimination.
\newblock {\em Physical Review Letters}, 127(20):200504, 2021.

\bibitem[BP10]{brandao2010generalization}
F.~G. Brandao and M.~B. Plenio.
\newblock {A generalization of quantum Stein’s lemma}.
\newblock {\em Communications in Mathematical Physics}, 295(3):791--828, 2010.

\bibitem[CDP08]{chiribella_memory_2008}
G.~Chiribella, G.~M. D’Ariano, and P.~Perinotti.
\newblock Memory {Effects} in {Quantum} {Channel} {Discrimination}.
\newblock {\em Physical Review Letters}, 101(18):180501, October 2008.

\bibitem[CMW16]{cooney2016strong}
T.~Cooney, M.~Mosonyi, and M.~M. Wilde.
\newblock Strong converse exponents for a quantum channel discrimination problem and quantum-feedback-assisted communication.
\newblock {\em Communications in Mathematical Physics}, 344(3):797--829, 2016.

\bibitem[CYB25]{cao2025channel}
M.~X. Cao, Y.~Yao, and M.~Berta.
\newblock Channel coding against quantum jammers via minimax.
\newblock {\em arXiv preprint arXiv:2505.11362}, 2025.

\bibitem[D'A01]{dariano_using_2001}
G.~M. D'Ariano.
\newblock Using {Entanglement} {Improves} the {Precision} of {Quantum} {Measurements}.
\newblock {\em Physical Review Letters}, 87(27), 2001.

\bibitem[Dat09]{datta2009min}
N.~Datta.
\newblock Min-and max-relative entropies and a new entanglement monotone.
\newblock {\em IEEE Transactions on Information Theory}, 55(6):2816--2826, 2009.

\bibitem[DFY07]{duan2007entanglement}
R.~Duan, Y.~Feng, and M.~Ying.
\newblock Entanglement is not necessary for perfect discrimination between unitary operations.
\newblock {\em Physical Review Letters}, 98(10):100503, 2007.

\bibitem[DFY09]{duan2009perfect}
R.~Duan, Y.~Feng, and M.~Ying.
\newblock Perfect distinguishability of quantum operations.
\newblock {\em Physical Review Letters}, 103(21):210501, 2009.

\bibitem[DKQ{\etalchar{+}}23]{ding2023bounding}
D.~Ding, S.~Khatri, Y.~Quek, P.~W. Shor, X.~Wang, and M.~M. Wilde.
\newblock Bounding the forward classical capacity of bipartite quantum channels.
\newblock {\em IEEE Transactions on Information Theory}, 69(5):3034--3061, 2023.

\bibitem[DMHB13]{datta_smooth_2013}
N.~Datta, M.~Mosonyi, M.-H. Hsieh, and F.~G. S.~L. Brandao.
\newblock A {Smooth} {Entropy} {Approach} to {Quantum} {Hypothesis} {Testing} and the {Classical} {Capacity} of {Quantum} {Channels}.
\newblock {\em IEEE Transactions on Information Theory}, 59(12):8014--8026, December 2013.

\bibitem[Don86]{donald1986relative}
M.~J. Donald.
\newblock On the relative entropy.
\newblock {\em Communications in Mathematical Physics}, 105:13--34, 1986.

\bibitem[DSSB{\etalchar{+}}23]{debry2023experimental}
K.~DeBry, J.~Sinanan-Singh, C.~D. Bruzewicz, D.~Reens, M.~E. Kim, M.~P. Roychowdhury, R.~McConnell, I.~L. Chuang, and J.~Chiaverini.
\newblock Experimental quantum channel discrimination using metastable states of a trapped ion.
\newblock {\em Physical Review Letters}, 131(17):170602, 2023.

\bibitem[DWH25]{Dasgupta2025}
A.~Dasgupta, N.~A. Warsi, and M.~Hayashi.
\newblock Universal tester for multiple independence testing and classical-quantum arbitrarily varying multiple access channel.
\newblock {\em IEEE Transations on Information Theory}, 71(5):3719–3765, May 2025.

\bibitem[FF21]{fang2021geometric}
K.~Fang and H.~Fawzi.
\newblock {Geometric R{\'e}nyi divergence and its applications in quantum channel capacities}.
\newblock {\em Communications in Mathematical Physics}, 384(3):1615--1677, 2021.

\bibitem[FFF25]{fang2025adversarial}
K.~Fang, H.~Fawzi, and O.~Fawzi.
\newblock Adversarial quantum channel discrimination.
\newblock {\em Physical Review Letters}, 135:260201, Dec 2025.

\bibitem[FFF26]{fang2024generalized}
K.~Fang, H.~Fawzi, and O.~Fawzi.
\newblock Generalized quantum asymptotic equipartition.
\newblock {\em Communications in Mathematical Physics}, 407(10):208, 2026.

\bibitem[FFRS20]{fang2020chain}
K.~Fang, O.~Fawzi, R.~Renner, and D.~Sutter.
\newblock Chain rule for the quantum relative entropy.
\newblock {\em Physical Review Letters}, 124(10):100501, 2020.

\bibitem[FGW25]{fang2025towards}
K.~Fang, G.~Gour, and X.~Wang.
\newblock Towards the ultimate limits of quantum channel discrimination and quantum communication.
\newblock {\em SCIENCE CHINA Information Sciences}, 68(8):180509, 2025.

\bibitem[FR06]{farkas2006potential}
B.~Farkas and S.~G. R{\'e}v{\'e}sz.
\newblock Potential theoretic approach to rendezvous numbers.
\newblock {\em Monatshefte f{\"u}r mathematik}, 148(4):309--331, 2006.

\bibitem[Hay09]{hayashi_discrimination_2009}
M.~Hayashi.
\newblock Discrimination of {Two} {Channels} by {Adaptive} {Methods} and {Its} {Application} to {Quantum} {System}.
\newblock {\em IEEE Transactions on Information Theory}, 55(8):3807--3820, August 2009.

\bibitem[HHLW10]{harrow_adaptive_2010}
A.~W. Harrow, A.~Hassidim, D.~W. Leung, and J.~Watrous.
\newblock Adaptive versus nonadaptive strategies for quantum channel discrimination.
\newblock {\em Physical Review A}, 81(3):032339, March 2010.

\bibitem[HP91]{hiai1991proper}
F.~Hiai and D.~Petz.
\newblock The proper formula for relative entropy and its asymptotics in quantum probability.
\newblock {\em Communications in Mathematical Physics}, 143:99--114, 1991.

\bibitem[HY25]{Hayashi2025}
M.~Hayashi and H.~Yamasaki.
\newblock {The generalized quantum Stein's lemma and the second law of quantum resource theories}.
\newblock {\em Nature Physics}, 21(12):1988--1993, 2025.

\bibitem[KW24]{khatri2024principlesquantumcommunicationtheory}
S.~Khatri and M.~M. Wilde.
\newblock Principles of quantum communication theory: A modern approach.
\newblock {\em arXiv: 2011.04672}, 2024.

\bibitem[Lam25a]{lami2024solutiongeneralisedquantumsteins}
L.~Lami.
\newblock {A solution of the generalized quantum Stein’s lemma}.
\newblock {\em IEEE Transactions on Information Theory}, 71(6):4454–4484, June 2025.

\bibitem[Lam25b]{Lam25_Sanov}
L.~Lami.
\newblock Generalised quantum {Sanov} theorem revisited.
\newblock {\em arXiv preprint arXiv:2510.06340}, 2025.
\newblock https://arxiv.org/abs/2510.06340.

\bibitem[LCVVK26]{lizarribar2026converse}
J.~Lizarribar-Carrillo, G.~Vazquez-Vilar, and T.~Koch.
\newblock A converse bound via the {Nussbaum--Szko{\l}a} mapping for quantum hypothesis testing.
\newblock arXiv:2601.13970, 2026.

\bibitem[LKDW18]{leditzky2018approaches}
F.~Leditzky, E.~Kaur, N.~Datta, and M.~M. Wilde.
\newblock Approaches for approximate additivity of the {H}olevo information of quantum channels.
\newblock {\em Physical Review A}, 97(1):012332, 2018.

\bibitem[Llo08]{Lloyd2008QuantumIllumination}
S.~Lloyd.
\newblock Enhanced sensitivity of photodetection via quantum illumination.
\newblock {\em Science}, 321(5895):1463--1465, September 2008.

\bibitem[MH24]{mosonyi2023some}
M.~Mosonyi and F.~Hiai.
\newblock Some continuity properties of quantum r{\'e}nyi divergences.
\newblock {\em IEEE Transactions on Information Theory}, 70(4):2674--2700, 2024.

\bibitem[MLDS{\etalchar{+}}13]{muller2013quantum}
M.~M{\"u}ller-Lennert, F.~Dupuis, O.~Szehr, S.~Fehr, and M.~Tomamichel.
\newblock On quantum {R}{\'e}nyi entropies: {A} new generalization and some properties.
\newblock {\em Journal of Mathematical Physics}, 54(12), 2013.

\bibitem[NH07]{NH07}
H.~Nagaoka and M.~Hayashi.
\newblock An information-spectrum approach to classical and quantum hypothesis testing for simple hypotheses.
\newblock {\em {IEEE} Transactions on Information Theory}, 53(2):534--549, Feb 2007.

\bibitem[NS09]{Nussbaum2009}
M.~Nussbaum and A.~Szko{\l}a.
\newblock {The Chernoff Lower Bound For Symmetric Quantum Hypothesis Testing}.
\newblock {\em The Annals of Statistics}, 37(2):1040--1057, 2009.

\bibitem[ON00]{Ogawa2000}
T.~Ogawa and H.~Nagaoka.
\newblock {Strong converse and {Stein's} lemma in quantum hypothesis testing}.
\newblock {\em IEEE Transactions on Information Theory}, 46(7):2428--2433, feb 2000.

\bibitem[PBG{\etalchar{+}}18]{pirandola_advances_2018}
S.~Pirandola, B.~R. Bardhan, T.~Gehring, C.~Weedbrook, and S.~Lloyd.
\newblock Advances in photonic quantum sensing.
\newblock {\em Nature Photonics}, 12(12):724--733, December 2018.

\bibitem[Per20]{pereira2020quantum}
J.~Pereira.
\newblock {\em Quantum channel simulation and discrimination with applications to quantum communications}.
\newblock PhD thesis, University of York, 2020.

\bibitem[PLLP19]{pirandola_fundamental_2019}
S.~Pirandola, R.~Laurenza, C.~Lupo, and J.~L. Pereira.
\newblock Fundamental limits to quantum channel discrimination.
\newblock {\em npj Quantum Information}, 5(1):50, June 2019.

\bibitem[PW09]{piani2009all}
M.~Piani and J.~Watrous.
\newblock All entangled states are useful for channel discrimination.
\newblock {\em Physical Review Letters}, 102(25):250501, 2009.

\bibitem[QWW18]{qi2018applications}
H.~Qi, Q.~Wang, and M.~M. Wilde.
\newblock Applications of position-based coding to classical communication over quantum channels.
\newblock {\em Journal of Physics A: Mathematical and Theoretical}, 51(44):444002, 2018.

\bibitem[SDM{\etalchar{+}}24]{sugiura2024power}
S.~Sugiura, A.~Dutt, W.~J. Munro, S.~Zeytino{\u{g}}lu, and I.~L. Chuang.
\newblock Power of sequential protocols in hidden quantum channel discrimination.
\newblock {\em Physical Review Letters}, 132(24):240805, 2024.

\bibitem[SL19]{skrzypczyk_robustness_2019}
P.~Skrzypczyk and N.~Linden.
\newblock Robustness of {Measurement}, {Discrimination} {Games}, and {Accessible} {Information}.
\newblock {\em Physical Review Letters}, 122(14):140403, April 2019.

\bibitem[SPL{\etalchar{+}}20]{spedalieri_detecting_2020}
G.~Spedalieri, L.~Piersimoni, O.~Laurino, S.~L. Braunstein, and S.~Pirandola.
\newblock Detecting and tracking bacteria with quantum light.
\newblock {\em Physical Review Research}, 2(4):043260, November 2020.

\bibitem[TRB{\etalchar{+}}19]{takagi_operational_2019}
R.~Takagi, B.~Regula, K.~Bu, Z.-W. Liu, and G.~Adesso.
\newblock Operational advantage of quantum resources in subchannel discrimination.
\newblock {\em Physical Review Letters}, 122:140402, Apr 2019.

\bibitem[Ume54]{umegaki1954conditional}
H.~Umegaki.
\newblock Conditional expectation in an operator algebra.
\newblock {\em Tohoku Mathematical Journal, Second Series}, 6(2-3):177--181, 1954.

\bibitem[Wat18]{watrous2018theory}
J.~Watrous.
\newblock {\em The theory of quantum information}.
\newblock Cambridge University Press, 2018.

\bibitem[WBHK20]{wilde2020amortized}
M.~M. Wilde, M.~Berta, C.~Hirche, and E.~Kaur.
\newblock Amortized channel divergence for asymptotic quantum channel discrimination.
\newblock {\em Letters in Mathematical Physics}, 110:2277--2336, 2020.

\bibitem[WFD19]{Wang2019}
X.~Wang, K.~Fang, and R.~Duan.
\newblock {Semidefinite Programming Converse Bounds for Quantum Communication}.
\newblock {\em IEEE Transactions on Information Theory}, 65(4):2583--2592, apr 2019.

\bibitem[WFT19]{wang2019converse}
X.~{Wang}, K.~{Fang}, and M.~{Tomamichel}.
\newblock On converse bounds for classical communication over quantum channels.
\newblock {\em IEEE Transactions on Information Theory}, 65(7):4609--4619, July 2019.

\bibitem[WR12]{wang2012one}
L.~Wang and R.~Renner.
\newblock One-shot classical-quantum capacity and hypothesis testing.
\newblock {\em Physical Review Letters}, 108(20):200501, 2012.

\bibitem[WW19]{WW2019}
X.~Wang and M.~M. Wilde.
\newblock {Resource theory of asymmetric distinguishability for quantum channels}.
\newblock {\em Physical Review Research}, 1(3):033169, dec 2019.

\bibitem[WWY14]{wilde2014strong}
M.~M. Wilde, A.~Winter, and D.~Yang.
\newblock Strong converse for the classical capacity of entanglement-breaking and {H}adamard channels via a sandwiched {R}{\'e}nyi relative entropy.
\newblock {\em Communications in Mathematical Physics}, 331:593--622, 2014.

\bibitem[ZH19]{zhu2019efficient}
H.~Zhu and M.~Hayashi.
\newblock Efficient verification of pure quantum states in the adversarial scenario.
\newblock {\em Physical Review Letters}, 123(26):260504, 2019.

\bibitem[ZP20]{zhuang_ultimate_2020}
Q.~Zhuang and S.~Pirandola.
\newblock Ultimate {Limits} for {Multiple} {Quantum} {Channel} {Discrimination}.
\newblock {\em Physical Review Letters}, 125(8):080505, August 2020.

\end{thebibliography}

\appendix

\section{Strict separations between IID jammer models}
\label{sec:strict-iid-gap}

This appendix provides the two channel-level separations invoked in
Remark~\ref{rem:iid-exponent-separation}.  Both examples use qubit channels
and satisfy the max-relative-entropy finiteness assumption imposed in the
main text.

\begin{example}
\label{exam:iid-awareness-gap}
Let $A\simeq\mathbb C$ and $E\simeq B\simeq\mathbb C^2$, and write
$\pi_2:=I_B/2$.  Consider the measure-and-prepare channels
\begin{align}
\cN(\omega)
&:=\langle0|\omega|0\rangle
\left(\frac34\ket0\!\bra0+\frac14\ket1\!\bra1\right)
+\langle1|\omega|1\rangle\pi_2,
\\
\cM(\omega)
&:=\langle0|\omega|0\rangle\pi_2
+\langle1|\omega|1\rangle
\left(\frac14\ket0\!\bra0+\frac34\ket1\!\bra1\right).
\label{eq:iid-awareness-gap-channels}
\end{align}
They satisfy
\begin{align}
D_{\max}^{\uparrow}(\cN\|\cM)&\le1,
&
D^{\uparrow,\downarrow}(\cN\|\cM)
&>D^{\uparrow,\downarrow\downarrow}(\cN\|\cM)=0.
\label{eq:iid-awareness-gap-conclusion}
\end{align}
\end{example}

\begin{proof}
In the computational product basis, their unnormalized Choi operators satisfy
\begin{align}
J_{\cN}&=\frac14\operatorname{diag}(3,1,2,2),\\
J_{\cM}&=\frac14\operatorname{diag}(2,2,1,3)>0,
\\
2J_{\cM}-J_{\cN}&=\frac14\operatorname{diag}(1,3,0,4)\ge0.
\end{align}
Thus $2\cM-\cN$ is completely positive, which gives
$D_{\max}^{\uparrow}(\cN\|\cM)\le\log2=1$.
For separate jammer inputs, we have
\begin{align}
\cN(\ket1\!\bra1)&=\cM(\ket0\!\bra0)=\pi_2,
& D^{\downarrow\downarrow}(\cN\|\cM)&=0,
\end{align}
where the divergence equality follows from nonnegativity of relative entropy.
In contrast, every common input $\tau\in\density(E)$ satisfies
\begin{align}
\cN(\tau)-\cM(\tau)
&=\frac14\left(\ket0\!\bra0-\ket1\!\bra1\right),
\\
D(\cN(\tau)\|\cM(\tau))
&\ge\frac{\|\cN(\tau)-\cM(\tau)\|_1^2}{2\ln 2}
=\frac{1}{8\ln 2}>0.
\end{align}
Here the inequality is the quantum Pinsker inequality
\cite[Theorem~5.38]{watrous2018theory}, with base-two logarithms.
The lower bound is independent of $\tau$, so the common-input minimum
is at least $1/(8\ln 2)$.
Since the tester input is trivial, these common- and separate-input minima
equal $D^{\uparrow,\downarrow}$ and $D^{\uparrow,\downarrow\downarrow}$,
respectively.  Theorem~\ref{thm: public hypo hypothesis-unaware iid} therefore
gives a strict separation of the corresponding public IID-jammer exponents.
\end{proof}

\begin{example}
\label{exam:strict-iid-gap}
Let $A\simeq\mathbb C$ and $E\simeq B\simeq\mathbb C^2$.  Define
$\rho:=\ket{\psi}\!\bra{\psi}$, where
\begin{align}
\ket{\psi}
&:=\sqrt{\frac78}\ket{0}+\sqrt{\frac18}\ket{1},
\end{align}
and let $X$, $Y$, and $Z$ denote the Pauli matrices.  Consider the channels
\begin{align}
\cN(\omega)
&:=\tr(\omega)\rho,
\label{eq:strict-iid-gap-replacer-channel}
\\
\cM(\omega)
&:=\frac7{16}\omega
+\frac14X\omega X
+\frac14Y\omega Y
+\frac1{16}Z\omega Z.
\label{eq:strict-iid-gap-Pauli-channel}
\end{align}
Then $D_{\max}^{\uparrow}(\cN\|\cM)<\infty$ and
\begin{align}
D^{\uparrow,\downarrow\downarrow}(\cN\|\cM)
&>
D^{\infty,\emptyuparrow,
\convdownarrow\convdownarrow}(\cN\|\cM).
\label{eq:strict-iid-gap-conclusion}
\end{align}
\end{example}

\begin{proof}
\prooftag{Channel properties and single-letter value.}
All four Pauli weights of $\cM$ are strictly positive, so its Choi operator is
positive definite.  Hence $J_{\cN}\le\kappa J_{\cM}$ for some finite
$\kappa>0$, which implies
$D_{\max}^{\uparrow}(\cN\|\cM)<\infty$.  The action of $\cM$ on Bloch
vectors is
\begin{align}
\cM\!\left(\frac{I+xX+yY+zZ}{2}\right)
&=\frac{I+\frac38xX+\frac38yY}{2}.
\label{eq:strict-iid-gap-Bloch-action}
\end{align}
Consequently, the output range of $\cM$ is the equatorial disk
\begin{align}
\cM(\density(E))
&=
\left\{
\tau_{r,\phi}:=
\frac{I+r\cos\phi\,X+r\sin\phi\,Y}{2}:
0\le r\le\frac38,\ 0\le\phi<2\pi
\right\}.
\label{eq:strict-iid-gap-output-disk}
\end{align}
The Bloch vector of $\rho$ is $(\sqrt7/4,0,3/4)$.  Since $\rho$ is pure,
\begin{align}
D(\rho\|\tau_{r,\phi})
&=
-\frac{1+\frac{\sqrt7}{4}\cos\phi}{2}
 \log\frac{1+r}{2}
-\frac{1-\frac{\sqrt7}{4}\cos\phi}{2}
 \log\frac{1-r}{2}.
\label{eq:strict-iid-gap-one-letter-objective}
\end{align}
For fixed $r$, this expression is minimized at $\phi=0$.  Moreover,
\begin{align}
\frac{\partial}{\partial r}D(\rho\|\tau_{r,0})
&=
\frac{r-\frac{\sqrt7}{4}}{(1-r^2)\ln 2}<0,
\qquad 0\le r\le\frac38.
\label{eq:strict-iid-gap-radial-derivative}
\end{align}
The minimizing output is therefore $\tau_{3/8,0}$.  Since the tester is
trivial and $\cN$ is a replacer channel,
\begin{align}
D^{\uparrow,\downarrow\downarrow}(\cN\|\cM)
&=d,
\label{eq:strict-iid-gap-single-letter-value}
\\
d
&:=
\log16
-\frac{4+\sqrt7}{8}\log11
-\frac{4-\sqrt7}{8}\log5
\notag\\
&=0.733126213324922\ldots.
\label{eq:strict-iid-gap-d-value}
\end{align}

\prooftag{A feasible convexified-IID input.}
For $\phi\in[0,2\pi)$, define the pure state and its channel
output by
\begin{align}
\eta_\phi
&:=\frac{I+\cos\phi\,X+\sin\phi\,Y}{2},
&
\sigma_\phi
&:=\cM(\eta_\phi)
=\frac{I+\frac38\cos\phi\,X+\frac38\sin\phi\,Y}{2}.
\label{eq:strict-iid-gap-phase-states}
\end{align}
For each $n$, let $\phi_{j,n}:=2\pi j/(n+1)$ and set
\begin{align}
\omega_n
&:=\frac1{n+1}\sum_{j=0}^{n}\eta_{\phi_{j,n}}^{\ox n}
\in\conv(\sI(E^n)),
\label{eq:strict-iid-gap-convex-IID-input}
\\
\tau_n
&:=\cM^{\ox n}(\omega_n)
=\frac1{n+1}\sum_{j=0}^{n}\sigma_{\phi_{j,n}}^{\ox n}.
\label{eq:strict-iid-gap-phase-twirled-output}
\end{align}
The root-of-unity average removes matrix elements between computational-basis
strings of different Hamming weights.  If $x,y\in\{0,1\}^n$ have equal
weight and Hamming distance $2s$, then
\begin{align}
\bra{x}\tau_n\ket{y}
&=2^{-n}\left(\frac9{64}\right)^s.
\label{eq:strict-iid-gap-twirled-matrix-elements}
\end{align}
Let $\ket{D_{n,k}}$ be the normalized Dicke vector of Hamming weight $k$.
Counting strings of weight $k$ at distance $2s$ from a fixed string shows that
$\ket{D_{n,k}}$ is an eigenvector of $\tau_n$ with eigenvalue
\begin{align}
\lambda_{n,k}
&:=2^{-n}
\sum_{s=0}^{\min\{k,n-k\}}
\binom{k}{s}\binom{n-k}{s}
\left(\frac9{64}\right)^s.
\label{eq:strict-iid-gap-Dicke-eigenvalue}
\end{align}
Writing $q:=1/8$, we have
\begin{align}
\ket{\psi}^{\ox n}
&=
\sum_{k=0}^{n}
\sqrt{\binom nk(1-q)^{n-k}q^k}\,\ket{D_{n,k}}.
\label{eq:strict-iid-gap-Dicke-expansion}
\end{align}
It follows that
\begin{align}
D(\rho^{\ox n}\|\tau_n)
&=-\sum_{k=0}^{n}
\binom nk(1-q)^{n-k}q^k\log\lambda_{n,k}.
\label{eq:strict-iid-gap-binomial-expectation}
\end{align}

\prooftag{Asymptotic evaluation.}
The standard type bounds for binomial coefficients imply, uniformly when
$k/n$ ranges over a fixed neighborhood of $q$,
\begin{align}
&\frac1n\log
\sum_{s=0}^{\min\{k,n-k\}}
\binom{k}{s}\binom{n-k}{s}
\left(\frac9{64}\right)^s
=
G\!\left(\frac kn\right)+o(1),
\label{eq:strict-iid-gap-type-asymptotics}
\\
G(a)
&:=
\max_{0\le t\le\min\{a,1-a\}}
\left\{
a h\!\left(\frac ta\right)
+(1-a)h\!\left(\frac{t}{1-a}\right)
-t\log\frac{64}{9}
\right\}.
\label{eq:strict-iid-gap-type-rate-function}
\end{align}
Indeed, this follows by applying
\begin{align}
\frac1{m+1}2^{m h(\ell/m)}
&\le\binom m\ell\le2^{m h(\ell/m)}
\label{eq:strict-iid-gap-type-bounds}
\end{align}
to both binomial coefficients and observing that the sum contains at most
$n+1$ terms.  Equation~\eqref{eq:strict-iid-gap-Dicke-eigenvalue} also gives
$\lambda_{n,k}\ge2^{-n}$, while $\lambda_{n,k}\le1$.  Hence
\begin{align}
0
&\le-\frac1n\log\lambda_{n,k}\le1.
\label{eq:strict-iid-gap-uniform-eigenvalue-bound}
\end{align}
Combining the uniform type estimate with the concentration of a binomial
random variable of parameter $q$ in
Eq.~\eqref{eq:strict-iid-gap-binomial-expectation} yields
\begin{align}
\lim_{n\to\infty}\frac1nD(\rho^{\ox n}\|\tau_n)
&=R,
&
R
&:=1-G(q).
\label{eq:strict-iid-gap-twirl-rate}
\end{align}

The function optimized in $G(q)$ is strictly concave on $(0,q)$, and its
derivative vanishes precisely when
\begin{align}
\frac9{64}(q-t)(1-q-t)
&=t^2.
\label{eq:strict-iid-gap-stationary-equation}
\end{align}
For $q=1/8$, the unique solution is $t_*=3/40$.  Substitution gives
\begin{align}
R
&=
1-\frac18\log\frac52-\frac78\log\frac{35}{32}
\notag\\
&=0.721636348312234\ldots.
\label{eq:strict-iid-gap-R-value}
\end{align}

Since $\omega_n$ is feasible for the convexified-IID alternative input and
$\cN$ ignores its jammer input,
\begin{align}
D^{\emptyuparrow,
\convdownarrow\convdownarrow}
(\cN^{\ox n}\|\cM^{\ox n})
&\le D(\rho^{\ox n}\|\tau_n).
\label{eq:strict-iid-gap-feasible-upper-bound}
\end{align}
Taking the regularized limit and using
Eq.~\eqref{eq:strict-iid-gap-twirl-rate} gives
\begin{align}
D^{\infty,\emptyuparrow,
\convdownarrow\convdownarrow}(\cN\|\cM)
&\le R.
\label{eq:strict-iid-gap-regularized-upper-bound}
\end{align}
This construction provides an upper bound on the regularized divergence; no
claim that the bound is tight is needed for the separation.

Finally, Eqs.~\eqref{eq:strict-iid-gap-d-value}
and~\eqref{eq:strict-iid-gap-R-value} give
\begin{align}
d-R
&=\frac18\log
\frac{7^7}{2^{12}(11/5)^{4+\sqrt7}}>0.
\label{eq:strict-iid-gap-exact-difference}
\end{align}
For completeness, the strict inequality follows from
\begin{align}
\left(\frac{11}{5}\right)^{4+\sqrt7}
&<
\left(\frac{11}{5}\right)^{20/3}
<114\cdot\frac{17}{10}
<200
<\frac{7^7}{2^{12}},
\label{eq:strict-iid-gap-elementary-comparison}
\end{align}
where $\sqrt7<8/3$, $(11/5)^6<114$, and
$(121/25)^{1/3}<17/10$.  Combining
Eqs.~\eqref{eq:strict-iid-gap-single-letter-value},
\eqref{eq:strict-iid-gap-regularized-upper-bound}, and
\eqref{eq:strict-iid-gap-exact-difference} proves
Eq.~\eqref{eq:strict-iid-gap-conclusion}.
\end{proof}

\section{Some useful identities}
\label{sec:operational-testing-exponents}

This appendix relates the operational exponents of Nagaoka and
Hayashi~\cite[Section~IV]{NH07} to the hypothesis-testing relative entropy
used in Section~\ref{sec:strong-converse-enhancement}. 
We allow an arbitrary scale $c_k\to\infty$ and prove the exact relations,
including the one-sided limits required at fixed error.

Let $\rho_k\in\density(\cH_k)$ and $\sigma_k\in\PSD(\cH_k)$, where each
$\cH_k$ is finite-dimensional, and let $\mathbf c=(c_k)_k$ be positive
with $c_k\to\infty$.  Write $\brho=(\rho_k)_k$, $\bsigma=(\sigma_k)_k$,
and $\bT=(T_k)_k$.  For effects $0\le T_k\le I$, set
\begin{align}
\alpha_k[T_k]&:=\tr[\rho_k(I-T_k)],
&
\beta_k[T_k]&:=\tr[\sigma_kT_k].
\end{align}
The operational exponents are
\begin{align}
B_{\mathbf c}(\ve\,|\,\brho\|\bsigma)
&:=\sup_{\bT}\left\{
\liminf_{k\to\infty}-\frac1{c_k}\log\beta_k[T_k]:
\limsup_{k\to\infty}\alpha_k[T_k]\le\ve
\right\},
\\
B_{\mathbf c}^{\dagger}(\ve\,|\,\brho\|\bsigma)
&:=\sup_{\bT}\left\{
\liminf_{k\to\infty}-\frac1{c_k}\log\beta_k[T_k]:
\liminf_{k\to\infty}\alpha_k[T_k]<\ve
\right\}.
\label{eq:general-scale-operational-definitions}
\end{align}
In particular,
\begin{align}
B_{\mathbf c}(\brho\|\bsigma)
&:=B_{\mathbf c}(0\,|\,\brho\|\bsigma),
&
B_{\mathbf c}^{\dagger}(\brho\|\bsigma)
&:=B_{\mathbf c}^{\dagger}(1\,|\,\brho\|\bsigma)
\end{align}
are the vanishing-error and strong-converse exponents, respectively.
We use $-\log0=+\infty$.  Classical distributions and finite nonnegative
measures are included by taking diagonal operators; restricting tests to
diagonal effects does not change either testing value.

\begin{lemma}
\label{prop:operational-testing-exponent-identities}
For $0\le\ve<1$,
\begin{align}
B_{\mathbf c}(\ve\,|\,\brho\|\bsigma)
&=\lim_{\eta\to\ve^+}\liminf_{k\to\infty}
\frac1{c_k}D_{\Hypo,\eta}(\rho_k\|\sigma_k).
\label{eq:operational-testing-fixed-error-direct}
\end{align}
For $0<\ve\le1$, 
\begin{align}
B_{\mathbf c}^{\dagger}(\ve\,|\,\brho\|\bsigma)
&=\lim_{\eta\to\ve^-}\limsup_{k\to\infty}
\frac1{c_k}D_{\Hypo,\eta}(\rho_k\|\sigma_k).
\label{eq:operational-testing-fixed-error-strong}
\end{align}
The limits in $\eta$ are taken within $(0,1)$.  In particular,
\begin{align}
B_{\mathbf c}(\brho\|\bsigma)
&=\lim_{\eta\to0^+}\liminf_{k\to\infty}
\frac1{c_k}D_{\Hypo,\eta}(\rho_k\|\sigma_k),
\label{eq:operational-testing-vanishing-error}
\\
B_{\mathbf c}^{\dagger}(\brho\|\bsigma)
&=\lim_{\eta\to1^-}\limsup_{k\to\infty}
\frac1{c_k}D_{\Hypo,\eta}(\rho_k\|\sigma_k).
\label{eq:operational-testing-strong-converse}
\end{align}
\end{lemma}

\begin{proof}
Write $h_k(\eta):=c_k^{-1}D_{\Hypo,\eta}(\rho_k\|\sigma_k)$.
The one-sided limits exist because $h_k$ is nondecreasing.
Optimal tests exist by compactness of the feasible set and continuity of
the type-II objective.

For Eq.~\eqref{eq:operational-testing-fixed-error-direct}, an admissible
sequence satisfies $\alpha_k[T_k]\le\eta$ eventually for every $\eta>\ve$.
Consequently,
\begin{align}
\liminf_{k\to\infty}-\frac1{c_k}\log\beta_k[T_k]
\le\liminf_{k\to\infty}h_k(\eta),
\qquad \eta>\ve.
\end{align}
Taking the supremum over admissible tests and then the infimum over $\eta$
proves the upper bound.  Conversely, let $r$ be any real number strictly
below the right-hand side of
Eq.~\eqref{eq:operational-testing-fixed-error-direct}, and choose
$\eta_j\in(\ve,1)$ with $\eta_j\to\ve$.
There are strictly increasing integers $K_j$ such that
$h_k(\eta_j)\ge r$ for every $k\ge K_j$.  On each block
$K_j\le k<K_{j+1}$, choose an optimal test at error $\eta_j$, and choose
the finitely many initial tests arbitrarily.  This sequence satisfies
\begin{align}
\limsup_{k\to\infty}\alpha_k[T_k]\le\ve,
\qquad
\liminf_{k\to\infty}-\frac1{c_k}\log\beta_k[T_k]\ge r.
\end{align}
Taking the supremum over $r$ gives the reverse inequality.

For Eq.~\eqref{eq:operational-testing-fixed-error-strong}, an admissible
sequence has $\alpha_k[T_k]\le\eta$ on an infinite set $K$ for some
$\eta\in(0,\ve)$.  Hence
\begin{align}
\liminf_{k\to\infty}-\frac1{c_k}\log\beta_k[T_k]
&\le\liminf_{\substack{k\to\infty\\k\in K}}
-\frac1{c_k}\log\beta_k[T_k]
\le\limsup_{k\to\infty}h_k(\eta).
\end{align}
Taking the supremum over tests gives the upper bound.  Conversely, fix
$\eta\in(0,\ve)$ and any real $r<\limsup_k h_k(\eta)$.
There are infinitely many indices with $h_k(\eta)\ge r$.
Use an optimal test at error $\eta$ on those indices and set $T_k=0$
elsewhere.  Then
\begin{align}
\liminf_{k\to\infty}\alpha_k[T_k]\le\eta<\ve,
\qquad
\liminf_{k\to\infty}-\frac1{c_k}\log\beta_k[T_k]\ge r,
\end{align}
because the zero test has type-II value zero.  Taking the supremum over
$r$ and $\eta$ proves the reverse inequality.  The arguments using arbitrary
finite $r$ also cover a right-hand side of $+\infty$; if it is $-\infty$,
the upper bound already gives equality.  The endpoint identities follow
by taking $\ve=0$ and $\ve=1$, respectively.
\end{proof}

For a probability distribution $p_k$ and a finite nonnegative measure $q_k$
on a finite set $\Omega_k$, define
\begin{align}
L_k(a):=\{x\in\Omega_k:p_k(x)>2^{c_ka}q_k(x)\}.
\end{align}
The endpoint cases of~\cite[Theorem~1]{NH07}, expressed through
Lemma~\ref{prop:operational-testing-exponent-identities}, are
\begin{align}
\lim_{\eta\to0^+}\liminf_{k\to\infty}\frac1{c_k}
D_{\Hypo,\eta}(p_k\|q_k)
&=\sup\left\{a\in\RR:p_k(L_k(a)^c)\to0\right\},
\label{eq:testing-spectrum-lower-identity}
\\
\lim_{\eta\to1^-}\limsup_{k\to\infty}\frac1{c_k}
D_{\Hypo,\eta}(p_k\|q_k)
&=\sup\left\{a\in\RR:\limsup_{k\to\infty}p_k(L_k(a))>0\right\}.
\label{eq:testing-spectrum-upper-identity}
\end{align}

\end{document}